\PassOptionsToPackage{unicode}{hyperref}
\PassOptionsToPackage{hyphens}{url}
\PassOptionsToPackage{dvipsnames,svgnames,x11names}{xcolor}
\documentclass[
  12pt]{article}
\usepackage{soul}
\usepackage[utf8]{inputenc}
\usepackage[T1]{fontenc}
\usepackage{lmodern}
\usepackage{authblk}

\usepackage{comment}

\usepackage{algorithm}       
\usepackage{algpseudocode}
\usepackage{amsmath, amssymb}
\usepackage{bm}

\usepackage{amsmath,amssymb}
\usepackage{iftex}
\ifPDFTeX
  \usepackage[T1]{fontenc}
  \usepackage[utf8]{inputenc}
  \usepackage{textcomp} 
\else 
  \usepackage{unicode-math}
  \defaultfontfeatures{Scale=MatchLowercase}
  \defaultfontfeatures[\rmfamily]{Ligatures=TeX,Scale=1}
\fi
\usepackage{lmodern}
\ifPDFTeX\else
\fi
\IfFileExists{upquote.sty}{\usepackage{upquote}}{}
\IfFileExists{microtype.sty}{
  \usepackage[]{microtype}
  \UseMicrotypeSet[protrusion]{basicmath} 
}{}
\makeatletter
\@ifundefined{KOMAClassName}{
  \IfFileExists{parskip.sty}{%
    \usepackage{parskip}
  }{
    \setlength{\parindent}{0pt}
    \setlength{\parskip}{6pt plus 2pt minus 1pt}}
}{
  \KOMAoptions{parskip=half}}
\makeatother
\usepackage{xcolor}
\makeatletter
\ifx\paragraph\undefined\else
  \let\oldparagraph\paragraph
  \renewcommand{\paragraph}{
    \@ifstar
      \xxxParagraphStar
      \xxxParagraphNoStar
  }
  \newcommand{\xxxParagraphStar}[1]{\oldparagraph*{#1}\mbox{}}
  \newcommand{\xxxParagraphNoStar}[1]{\oldparagraph{#1}\mbox{}}
\fi
\ifx\subparagraph\undefined\else
  \let\oldsubparagraph\subparagraph
  \renewcommand{\subparagraph}{
    \@ifstar
      \xxxSubParagraphStar
      \xxxSubParagraphNoStar
  }
  \newcommand{\xxxSubParagraphStar}[1]{\oldsubparagraph*{#1}\mbox{}}
  \newcommand{\xxxSubParagraphNoStar}[1]{\oldsubparagraph{#1}\mbox{}}
\fi
\makeatother

\usepackage{longtable,booktabs,array}
\usepackage{calc} 
\usepackage{etoolbox}
\makeatletter
\patchcmd\longtable{\par}{\if@noskipsec\mbox{}\fi\par}{}{}
\makeatother
\IfFileExists{footnotehyper.sty}{\usepackage{footnotehyper}}{\usepackage{footnote}}
\makesavenoteenv{longtable}
\usepackage{graphicx}
\makeatletter
\def\maxwidth{\ifdim\Gin@nat@width>\linewidth\linewidth\else\Gin@nat@width\fi}
\def\maxheight{\ifdim\Gin@nat@height>\textheight\textheight\else\Gin@nat@height\fi}
\makeatother
\setkeys{Gin}{width=\maxwidth,height=\maxheight,keepaspectratio}
\makeatletter
\def\fps@figure{htbp}
\makeatother

\makeatletter
\@ifpackageloaded{caption}{}{\usepackage{caption}}
\AtBeginDocument{%
\ifdefined\contentsname
  \renewcommand*\contentsname{Table of contents}
\else
  \newcommand\contentsname{Table of contents}
\fi
\ifdefined\listfigurename
  \renewcommand*\listfigurename{List of Figures}
\else
  \newcommand\listfigurename{List of Figures}
\fi
\ifdefined\listtablename
  \renewcommand*\listtablename{List of Tables}
\else
  \newcommand\listtablename{List of Tables}
\fi
\ifdefined\figurename
  \renewcommand*\figurename{Figure}
\else
  \newcommand\figurename{Figure}
\fi
\ifdefined\tablename
  \renewcommand*\tablename{Table}
\else
  \newcommand\tablename{Table}
\fi
}
\@ifpackageloaded{float}{}{\usepackage{float}}
\floatstyle{ruled}
\@ifundefined{c@chapter}{\newfloat{codelisting}{h}{lop}}{\newfloat{codelisting}{h}{lop}[chapter]}
\floatname{codelisting}{Listing}

\makeatother
\makeatletter
\@ifpackageloaded{caption}{}{\usepackage{caption}}
\@ifpackageloaded{subcaption}{}{\usepackage{subcaption}}
\makeatother

\ifLuaTeX
  \usepackage{selnolig}  
\fi
\usepackage[]{natbib}
\usepackage{bookmark}

\IfFileExists{xurl.sty}{\usepackage{xurl}}{} 
\hypersetup{
  pdftitle={Title},
  pdfauthor={Author 1; Author 2},
  pdfkeywords={3 to 6 keywords, that do not appear in the title},
  colorlinks=true,
  linkcolor={blue},
  filecolor={Maroon},
  citecolor={Blue},
  urlcolor={Blue},
  pdfcreator={LaTeX via pandoc}}

\usepackage{amsthm}
\newtheorem{thm}{Theorem}[section]
\newtheorem{lem}[thm]{Lemma}
\newtheorem{proposition}[thm]{Proposition}
\newtheorem{corollary}[thm]{Corollary}

\newtheorem{definition}[thm]{Definition}
\newcommand{\Var}{\mbox{Var}}

\def\R{\mathbb{R}}
\def\E{\mathbb{E}}
\def\M{\mathbf{m}}
\def\m{\mathbf{m}}
\def\N{\mathbb{N}}
\def\X{\mathbf{x}}
\def\Y{\mathbf{y}}
\def\Z{\mathbf{z}}

\begin{document}

\def\spacingset#1{\renewcommand{\baselinestretch}%
{#1}\small\normalsize} \spacingset{1}



\date{}

\title{\vspace{-11mm}
\bf \LARGE High-dimensional Many--to--many--to--many \\ \vspace{2.5mm} Mediation Analysis \vspace{2.5mm}}

\author[ ]{%
Tien Dat Nguyen\textsuperscript{1,2}, Trung Khang Tran\textsuperscript{3}, Cong Khanh Truong\textsuperscript{1,2},
\newline Duy-Cat Can\textsuperscript{4,5,6}, Binh T. Nguyen\textsuperscript{1,2}, and Oliver Y. Chén\textsuperscript{5,6}\\
\vspace{2.5mm} for the Alzheimer's Disease Neuroimaging Initiative\thanks{See Data Availability Statement.}
\vspace{3mm}
}%


\affil[1]{\centering \normalsize Faculty of Mathematics and Computer Science, University of Science,\par\vspace{-0.65em}
Vietnam National University Ho Chi Minh City, Vietnam}
\affil[2]{\centering \normalsize Vietnam National University Ho Chi Minh City, Vietnam}
\affil[3]{\centering \normalsize National University of Singapore, Singapore}
\affil[4]{\centering \normalsize VNU University of Engineering and Technology, Hanoi, Vietnam}
\affil[5]{\centering \normalsize Lausanne University Hospital, Lausanne, Switzerland}
\affil[6]{\centering \normalsize University of Lausanne, Lausanne, Switzerland}

\maketitle

\vspace{-3.5mm}
\begin{abstract}
\vspace{-2mm}
We study high-dimensional mediation analysis in which exposures, mediators, and outcomes are all multivariate, and both exposures and mediators may be high-dimensional.
We formalize this as a many (exposures)--to--many (mediators)--to--many (outcomes) (MMM) mediation analysis problem.
Methodologically, MMM mediation analysis simultaneously performs variable selection for high-dimensional exposures and mediators, estimates the indirect effect matrix (i.e., the coefficient matrices linking exposure-to-mediator and mediator-to-outcome pathways), and enables prediction of multivariate outcomes. Theoretically, we show that the estimated indirect effect matrices are consistent and element-wise asymptotically normal, and we derive error bounds for the estimators.
To evaluate the efficacy of the  MMM mediation framework, we first investigate its finite-sample performance, including convergence properties, the behavior of the asymptotic approximations, and robustness to noise, via simulation studies.
We then apply MMM mediation analysis to data from the Alzheimer's Disease Neuroimaging Initiative to study how cortical thickness of 202 brain regions may mediate the effects of 688 genome-wide significant single nucleotide polymorphisms (SNPs) (selected from approximately 1.5 million SNPs) on eleven cognitive--behavioral and diagnostic outcomes.
The MMM mediation framework identifies biologically interpretable, many--to--many--to--many genetic--neural--cognitive pathways and improves downstream out-of-sample classification and prediction performance.
Taken together, our results demonstrate the potential of MMM mediation analysis and highlight the value of statistical methodology for investigating complex, high-dimensional multi-layer pathways in science. The MMM package is available at the \href{https://github.com/THELabTop/MMM-Mediation}
{MMM repository}.
\vspace{-6.5mm}

\end{abstract}
\vspace{2.5mm}
\noindent%
\vfill

\newpage
\spacingset{1.8} 


\section{Introduction}
\label{sec:Introduction}

\vspace{-.5em}
\subsection{Background and motivation}

\begin{figure}[htbp]
\centering
\includegraphics[width=\textwidth]{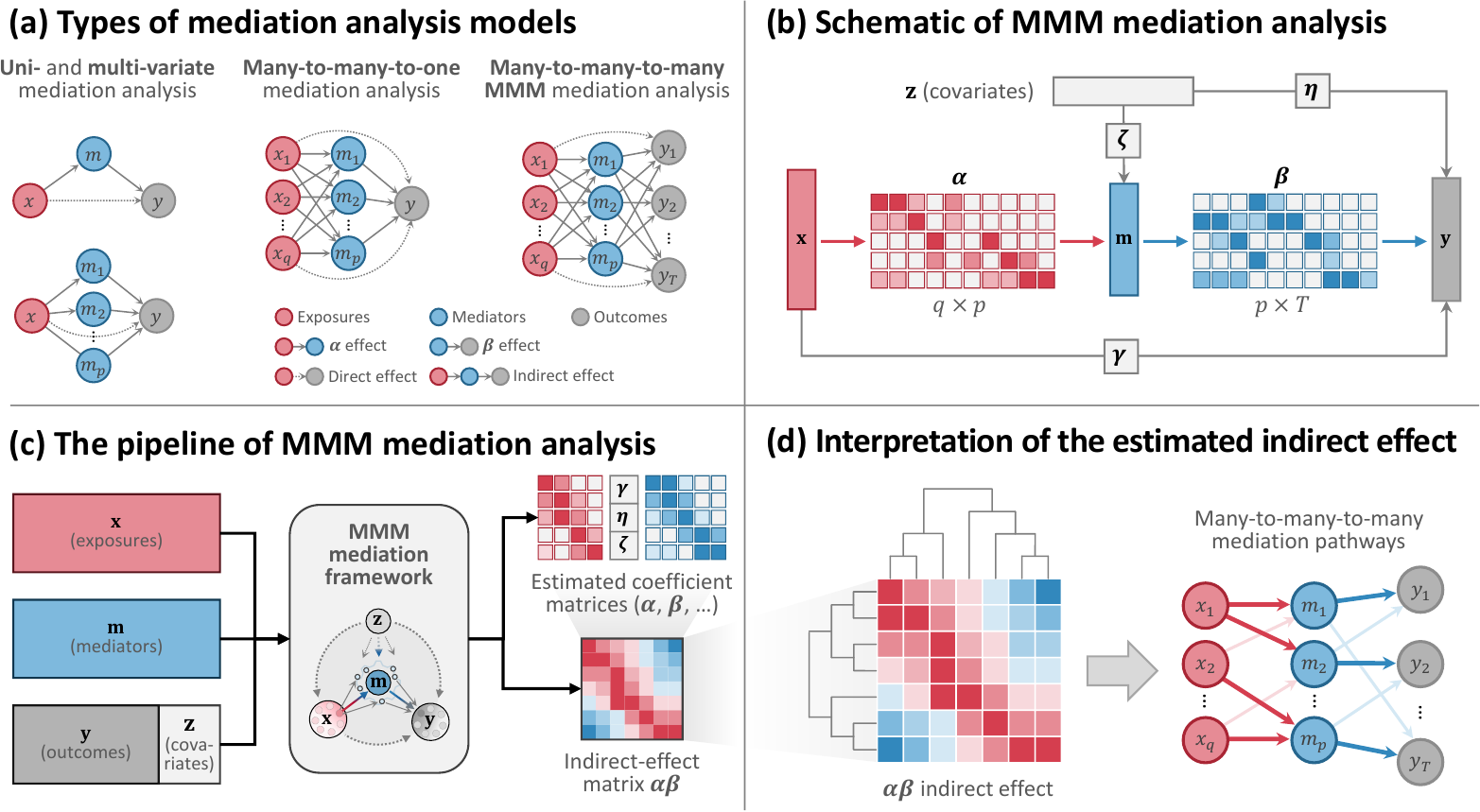}
\caption{
    \textbf{An overview of the many--to--many--to--many (MMM) mediation analysis framework.}
    {
    \footnotesize
    \textit{(a) Schematic representation of different types of mediation analysis.}
    An illustration of how classical univariate- and multivariate mediation analysis models extend to the many--to--many--to--many (MMM) setting, where multivariate exposures $\X$, mediators $\M$, and outcomes $\Y$ interact through multiple indirect pathways.
    \textit{(b) The MMM model.}
    A schematic representation of the multivariate linear structural equation model linking $\X$, $\M$, $\Y$, and covariates $\Z$ through coefficient matrices $(\bm{\alpha},\bm{\beta},\bm{\gamma},\bm{\zeta},\bm{\eta})$.
    \textit{(c) Analysis pipeline of the MMM method.}
    A high-level workflow showing input data layers, estimation of coefficient matrices, and the derivation of the indirect-effect matrix $\bm{\alpha\beta}$.
    \textit{(d) Output interpretation.}
     Estimated coefficient matrices and indirect-effect patterns from MMM mediation reveal structured many--to--many--to--many exposure--mediator--outcome pathways.
    }
}
\label{fig:framework}
\end{figure}
Mediation analysis is a central tool for studying mechanisms between layered systems, specifically how the effect of an exposure 
on an outcome 
is transmitted through an intermediate variable. 
As mediation analysis presents a simple, yet elegant way to identify, separate, and quantify an indirect effect from the exposure to the outcome that is carried through the mediator, it has interested economists~\citep{celli2022causal}, geneticists~\citep{zhang2022high}, neuroscientists~\citep{meng2023evaluating},  psychologists~\citep{rucker2011mediation}, and social scientists~\citep{figueredo2013revisiting},  no less than statisticians~\citep{baron1986moderator, pearl2001direct, vanderweele2015explanation}.

Perhaps the simplest mediation analysis is the univariate mediation analysis, consisting of the relationship between an exposure, an outcome, and a univariate mediator~\citep{baron1986moderator, robins1992identifiability} (see the top left panel of \textbf{Fig.}~\ref{fig:framework}\textit{(a)}). To estimate the univariate mediator effect, one can use the difference approach or the product approach~\citep{vanderweele2016mediation}. An extension of the univariate mediation analysis is the multivariate mediation analysis. As the name suggests, it studies how multiple mediators intermediating an exposure and an outcome~\citep{lindquist2012functional, vanderweele2014mediation} (see the lower left panel of \textbf{Fig.}~\ref{fig:framework}\textit{(a)}). To estimate the multivariate mediation effect, one can employ structural equation models (SEMs), where one can identify and separate the mediation effect due to each mediator. A special case of the multivariate mediation analysis is the high-dimensional multivariate mediation, where the number of the intermediate variables is large- or high-dimensional~\citep{huang2016hypothesis, zhang2016estimating, chen2018high, song2020bayesian, zhao2020sparse, hu2024high, cui2024mediation}. As the dimensionality of the mediator may be large and sometimes larger than the sample size, high-dimensional mediation analysis may first employ variable screening or dimension reduction to reduce dimensionality, or use regularization methods to enable estimation in high-dimensional settings. Examples of such high-dimensional mediation analyses include brain-imaging studies that evaluate the mediation effects of hundreds of thousands of voxels on the association between thermal pain stimulus and pain perception~\citep{chen2018high}, and genetic studies that investigate the mediation effects of hundreds of thousands of DNA-methylation probes on the relationship between smoking and lung function~\citep{zhang2016estimating}.

Regardless of the dimensionality of the mediators, classical high-dimensional mediation analyses typically consider a single exposure and a single outcome. This confines the analysis to settings in which there is a dominant exposure (such as an etiologic determinant or medical treatment) and a primary outcome (such as a disease endpoint). This naturally leaves out a large terrain where the high-dimensional mediators may be affected by multiple exposures which then give rise to an outcome (see the middle panel of \textbf{Fig.}~\ref{fig:framework}\textit{(a)}), or that the high-dimensional mediator affected by one exposure may result in changes in several outcomes (see the right panel of \textbf{Fig.}~\ref{fig:framework}\textit{(a)} for a general case). For the former,~\citep{zhao2022multimodal, zhao2024mediation} consider a high-dimensional mediation problem with multiple, potentially high-dimensional exposures. To estimate mediation effects in high-dimensional settings, one can learn linear projections of exposures and mediators, either unsupervised (e.g., PCA) or supervised (jointly with the outcome), and apply sparsity or thresholding to handle the high-dimensionality. For the latter, there appear to be no established methodological frameworks to address this problem, despite practical needs. For instance, in Alzheimer's disease (AD) studies, brain disruption may affect multiple cognitive scores, each reflecting different aspects of cognitive function~\citep{vu2025residual, Diaz2025OPTIMUS}. At the same time, strong genetic effects, such as those of the APOE gene, play a key role in AD~\citep{sienski2021apoe4, Diaz2025OPTIMUS}. Consequently, it is both scientifically relevant and potentially important to investigate how the brain may mediate the effects of a genetic factor on multiple cognitive outcomes. One possible approach is to separately apply mediation analysis designed for single outcomes to each cognitive score. This strategy, however, can be computationally burdensome when the number of outcomes is large; additionally, by studying univariate outcomes in isolation, one may fail to identify mediators that are both associated with the exposure and jointly influence multiple outcomes, as well as mediators whose exposure–mediator pathways uniquely target specific outcomes. As one would see below, this scenario is a special case (i.e., univariate APOE gene) of the general framework we propose, which we discuss in greater detail below.

Although recent advances in mediation analysis have substantially improved our understanding of how high-dimensional mediators operate within complex systems, investigations of high-dimensional mediation involving multiple exposures and multiple outcomes (see the lower right panel of \textbf{Fig.}~\ref{fig:framework}\textit{(a)}) remain scarce. Yet, there is a strong demand for such analyses, as many intermediate variables, whether features of the human brain or the stock performance of the S\&P 500, are influenced by multiple factors and, in turn, affect multiple outcomes. For instance, despite the prominent effect of the APOE gene, Alzheimer's disease is considered a polygenic disorder~\citep{harrisonPolygenicScoresPrecision2020, bellenguez2022new} and simultaneously impacts multiple cognitive and behavioral domains, including memory, executive function, visuospatial function, and language~\citep{Diaz2025OPTIMUS}.

In addition to facilitating the identification of effects from exposures through mediators to outcomes, developing appropriate statistical methods for many-to-many-to-many mediation also enables prediction of multivariate outcomes based on exposures and predicted mediators. Nevertheless, developing such methods presents both analytical and scientific challenges. First, this problem involves multilayer pathways, with each layer consisting of multivariate, potentially high-dimensional variables. While deep learning approaches may address the multilayer prediction task (e.g., multilayer perceptrons~\citep{rumelhart1986learning}, graph neural networks~\citep{scarselli2008graph}, graph-based Transformers~\citep{ying2021transformers}, and multilayer networks with embeddings~\citep{guillemaud2025hyperbolic}), identifying, estimating, and disentangling the many--to--many--to--many mediation effects, and making inference for them, requires new statistical methodology. Second, this problem demands a careful integration of statistical apparatuses and scientific insights. For instance, it is reasonable to posit that the structure and function of the brain are partially dictated by genetics, and that the brain influences virtually all human behaviors. To elevate statistical findings to the level of scientific evidence, one, however, must, on the one hand, demonstrate rigorous theoretical properties of the proposed methods and estimates, and, on the other hand, show that the identified mediators and the many-to-many-to-many pathways not only improve prediction of multivariate outcomes but are also scientifically or biologically explainable.

Building on the work of pioneering statisticians and biologists, we propose a mediation framework called many-to-many-to-many (MMM) analysis to (1) estimate mediation effects when exposures, mediators, and outcomes are all multivariate, and both exposures and mediators may be high-dimensional, and (2) predict multivariate outcomes considering the selected exposures (e.g., selected SNPs) and predicted mediators (e.g., brain representations predicted by the SNPs) as features or predictors (see \textbf{Fig.}~\ref{fig:framework}\textit{(b-d)}). To illustrate the efficacy of the MMM method, we first establish its theoretical properties, including asymptotic normality, consistency, and error bounds for the indirect effect matrix. We then evaluate its finite-sample performance and apply it to investigate how high-dimensional human brain features mediate high-dimensional genetic information and eleven multivariate disease-behavior and diagnostic outcomes in Alzheimer's disease (AD). Finally, we demonstrate that the identified genetic factors and the brain representations estimated using the mediation methods are predictive of multiple AD-related outcomes in previously unseen subjects.

\subsection{Outline}

The rest of the article is organized as follows.
In \textbf{Sec.}~\ref{sec:method}, we introduce
the model and the estimation method for MMM.
In \textbf{Sec.}~\ref{sec:theoretical_results}, we discuss the theoretical properties of the estimator.
In \textbf{Sec.}~\ref{sec:simulation}, we illustrate the finite sample performance of the estimators using simulation studies.
In \textbf{Sec.}~\ref{sec:data_analysis}, we apply MMM
to study how high-dimensional brain data mediate high-dimensional genetic variables and multivariate disease outcomes and use the selected genetic features and predicted brain representations as features to predict eleven disease-related outcomes in AD.
We provide proofs of theorems and lemmas in the Supplementary Materials.

\vspace{-1em}
\section{Method} \label{sec:method}
\vspace{-.75em}
\subsection{Notations}

In this paper, we consider mediation analysis with multiple exposures, mediators, and outcomes, and that both the exposure and mediator variables are potentially high-dimensional.

We begin by defining key notations used throughout the paper. We consider a mediation problem where the mediator $\mathbf{m} \in \mathbb{R}^{p}$ is interposed between an exposure vector $\mathbf {x} \in \mathbb{R}^{q}$ and an outcome vector $\mathbf {y} \in \mathbb{R}^{T}$, with $p$ and $q$ allowed to exceed the sample size $n$ (see \textbf{Fig.} \ref{fig:framework}).

More specifically, for subject $i \in \left\{ 1, ..., n \right\}$, let $\mathbf{x}_{i}^{\top} = \big( x_{i1}, ..., x_{i q} \big) \in \R^{1 \times  q}$ denote the $q$-dimensional exposures with $q$ the number of all exposures (e.g., genes). Consider $\mathbf {m}_{i}^{\top} = \big( m_{i1},... , m_{ip}  \big) \in \R^{1  \times  p}$ denote the $p$-dimensional mediators with $p$ the number of mediators (e.g., the number of brain regions). Let $\mathbf{y}_{i}^{\top} = \big( y_{i1}, ..., y_{iT}  \big) \in \R^{1  \times  T}$ denote the $T$-dimensional outcome with $T$ the number of outcomes (e.g., cognitive-behavior outcomes). Finally, we denote $\mathbf{z}_{i}^{\top} = \big( z_{i1}, ..., z_{is} \big)  \in \R^{1 \times  s}$ as the $s$-dimensional covariates (with the first element of $1$ for the intercept in linear regression equation), where $s-1$ is the number of covariates (such as age and gender).


\subsection{The model}
\vspace{-.5em}

To formalize the many-to-many-to-many mediation analysis, we consider a multivariate linear structural equation model (LSEM):
\begin{align}
	\ddot{\mathbf{m}}_{i}  &=   \bm{\alpha}^{\top} \mathbf{x}_{i}
    +
       \bm{\zeta}^{\top} \mathbf{z}_{i}
    +
    \bm{\epsilon}_{i}\label{equa:MMAMA.model:selected.genes:equa1} \\
	\ddot{\mathbf{y}}_{i}
    &=    \bm{\beta}^{\top} \mathbf{m}_{i}
    +
        \bm{\gamma}^{\top}\mathbf{x}_{i}
    +
      \bm{\eta}^{\top} \mathbf{z}_{i} +    
     \bm{\xi}_{i},  	\label{equa:MMAMA.model:selected.genes:equa2}
\end{align}
where $\bm{\alpha} \in \R^{q \times p}$,  $\bm{\beta} \in \R^{p \times T}$ and $\bm{\gamma} \in \R^{q  \times T}$ denote the coefficient matrices related to the input pathway from the exposure to the mediator, from mediator to the outcome (conditioning on the exposure), and from the exposure to the outcome, respectively. Additionally,
$\bm{\zeta} \in \R^{s \times p}$ and $\bm{\eta} \in \R^{s \times T}$ are coefficient matrices corresponding to covariates in each equation.
Finally, $\bm{\epsilon_{i}}{\sim} \mathcal{N} \left( 0 ; \mathbf{I}\right)$
and 
$\bm{\xi_{i}}{\sim} \mathcal{N} \left( 0 ; \mathbf{I} \right)$ are Gaussian random noise vectors with covariance matrices in $\mathbb{R}^p$ and $\mathbb{R}^T$, respectively.

We have, from \textbf{Eqs.}~\eqref{equa:MMAMA.model:selected.genes:equa1} and \eqref{equa:MMAMA.model:selected.genes:equa2}, the \textit{many--to--many--to--many mediation effect}, or the matrix version of the indirect effect:
{\setlength{\abovedisplayskip}{4pt}%
\setlength{\belowdisplayskip}{4pt}%
\setlength{\abovedisplayshortskip}{2pt}%
\setlength{\belowdisplayshortskip}{2pt}%
\begin{align}
	\bm{\alpha} \bm{\beta},
\end{align}}
where $\bm{\alpha} = \{ \alpha_{jk} \in \mathbb{R} \mid 1 \le j \le q, \, 1 \le k \le p \}$ is a $q \times p$ matrix whose $\{j, k\}^{th}$ entry quantifies the pathway from the $j^{th}$ exposure to the $k^{th}$ mediator, and
$\bm{\beta} = \{ \beta_{kl} \in \mathbb{R} \mid 1 \le k \le p, \, 1 \le l \le T \}$ is a $p \times T$ matrix whose $\{k, l\}^{th}$ entry quantifies the pathway from the $k^{th}$ mediator to the $l^{th}$ outcome. Here, $q$ denotes the number of exposures, $p$ the number of mediators, and $T$ the number of outcomes.

It follows that \textit{the global mediation effect} from the $j^{th}$ exposure to the $l^{th}$ outcome that is carried through the \textit{entire} set of high-dimensional mediators is the $(j,l)^{th}$ entry of the mediation matrix:
\begin{align}
	\left \{  \bm{\alpha} \bm{\beta} \right  \}_{j,l}: = \bm{\alpha}_j \bm{\beta}_l = \sum_{k = 1}^{p} \alpha_{jk}\beta_{kl},
\end{align}
where $\bm{\alpha}_j$ is the $j^{th}$ row of $\bm{\alpha}$ and $\bm{\beta}_l$ is the $l^{th}$ column of $\bm{\beta}$.

\noindent The global mediation effect, $\left \{\bm{\alpha} \bm{\beta} \right  \}_{j,l}$, said in another way, quantifies the mediation effect from the $j^{th}$ exposure to the $l^{th}$ outcome that is summed over all mediators. In our context, this represents the indirect effect of the entire brain mediating the $j^{th}$ gene and the $l^{th}$ AD outcome.

As $\left \{  \bm{\alpha} \bm{\beta} \right \}_{j,l}$ summarizes the overall mediation effect due to the full set of high-dimensional mediators, one may also be interested in estimating individual mediation effects. To do so, we write \textbf{Eqs.} \eqref{equa:MMAMA.model:selected.genes:equa1} and \eqref{equa:MMAMA.model:selected.genes:equa2} as follows:
\begin{align}
	m_{ik}
     \label{equa:MMAMA.model:selected.genes:equa3}
    &=
    \sum_{j=1}^q x_{ij} \alpha_{jk}
    +   \Z_{i}^{\top} \bm{\zeta}_{k}
    +   \epsilon_{ik}, \\ 
    y_{il}
    &=
	\label{equa:MMAMA.model:selected.genes:equa4}
    \sum_{k=1}^p m_{ik} \beta_{kl}
    +
      \X_{i}^{\top} \bm{\gamma}_{l}
    +
    \Z_{i}^{\top} \bm{\eta}_{l}
    +
    \xi_{il},
\end{align}
\noindent where $\alpha_{jk}$ quantifies the effect from the $j^{th}$ exposure to the $k^{th}$  mediator, and $\beta_{kl}$ quantifies the effect from the $k^{th}$ mediator to the $l^{th}$ outcome.


From \textbf{Eqs.} \eqref{equa:MMAMA.model:selected.genes:equa3} and \eqref{equa:MMAMA.model:selected.genes:equa4}, one can obtain the mediation effect from the $j^{th}$ exposure to the $l^{th}$ outcome via the $k^{th}$ mediator:
\begin{align}
	\alpha_{jk} \beta_{kl},
\end{align}
for $1 \le j \le q$, $1 \le k \le p$, and $1 \le l \le T$.

In \textbf{Fig.} \ref{fig:framework}, we present the schematics of the MMM model. The model setting of the MMM poses several challenges. First, in terms of estimation, one needs to estimate $(\bm{\alpha},\bm{\beta}, \bm{\gamma},\bm{\zeta},\bm{\eta})$ across both layers of the LSEM, where the dimensionalities of $\bm{\alpha}$, $\bm{\beta}$, and  $\bm{\gamma}$ are all large or even high-dimensional. Second, beyond estimation, one needs to make inference on the indirect-effect matrix $\bm{\alpha}\bm{\beta}$ which is itself high-dimension. Third, to evaluate the reproducibility of the estimated pathways and parameters, it is necessary to extend the estimation framework to out-of-sample prediction. Whereas the estimated parameters of $\hat{\bm{\alpha}}$ and $\hat{\bm{\beta}}$ may provide insights into feature selection (given the high dimensions of both $\X$ and $\M$ are both high), it is unclear whether the estimated mediation parameters generalize to previously unseen samples.

To address the first challenge, we consider in \textbf{Sec.} \ref{sec:estimation-procedure} a regularized MMM estimation procedure and show that one can simultaneously estimate $(\bm{\alpha},\bm{\beta}, \bm{\gamma},\bm{\zeta},\bm{\eta})$. Regularization not only ensures identifiability, but also facilitate variable selection for prediction tasks.
For inference, we show in \textbf{Sec.} \ref{sec:theoretical_results} that the estimated MMM mediation effects are consistent, asymptotically normal, and achieve bounded error. We then empirically verify in \textbf{Sec.} \ref{sec:simulation} the theoretical properties via simulation studies. Finally, in  \textbf{Sec.} \ref{sec:prediction} and \ref{sec:data_analysis}, we demonstrate that the estimated mediation parameters are reproducible in the sense that they are effective for predicting multivariate outcomes in previously unseen samples, and that the selected genetic variables and brain mediators (enabled by the two regularizations) are biologically meaningful for explaining the genetic and neurological underpinnings of AD.



\vspace{-.75em}
\subsection{Estimation}
\label{sec:estimation-procedure}
\vspace{-.25em}


Let $\bm{\theta} = \big( \bm{\alpha}, \bm{\beta},  \bm{\zeta}, \bm{\gamma}, \bm{\eta} \big)$ denote the full set of model parameters and $\widehat{\bm{\alpha}},   \widehat{\bm{\zeta}},   \widehat{\bm{\beta}},   \widehat{\bm{\gamma}}, \widehat{\bm{\eta}}$ denote their corresponding estimators. For any matrix $\mathbf{A}= (a_{jk})$, we define the Frobenius norm and the $l_1$ norm as $\left\|  \mathbf{A}   \right\|_{2}^{2} = \displaystyle{ \sum_{j,k} }   a_{j,k}^{2}$ and $\left\|  \mathbf{A}   \right\|_{1} = \displaystyle{ \sum_{j,k} }   \vert a_{jk} \vert$, respectively.
For example, for our parameter matrices, the norms are $\left\|  \bm{\beta}   \right\|_{2}^{2} = \displaystyle{ \sum_{1 \leq j \leq  p ; 1 \leq k \leq  T } }   \beta_{j,k}^{2}  $,   \quad   $\left\|  \bm{\gamma}   \right\|_{2}^{2} = \displaystyle{ \sum_{1 \leq j \leq  q ; 1 \leq k \leq  T } }   \gamma_{j,k}^{2}  $ , \quad   $\left\|  \bm{\eta}   \right\|_{2}^{2} = \displaystyle{ \sum_{1 \leq j \leq  s ; 1 \leq k \leq  T } }   \eta_{j,k}^{2}$,
$\left\|  \bm{\alpha}   \right\|_{2}^{2} = \displaystyle{ \sum_{1 \leq j \leq  q ; 1 \leq k \leq  p } }   \alpha_{j,k}^{2}  $  \quad  and  \quad  $\left\|  \bm{\zeta}   \right\|_{2}^{2} = \displaystyle{ \sum_{1 \leq j \leq  s ; 1 \leq k \leq  p } }   \zeta_{j,k}^{2}$, where $q$ is the number of exposures, $p$ is the number of mediators, $T$ is the number of outcomes, and $s$ is the number of covariates.

Let $\bm{\beta} = \big( \bm{\beta}_{1}, \ldots, \bm{\beta}_{T} \big) \in \R^{p \times T}$ with $\bm{\beta_{k}} = \big( \beta_{1k}, ..., \beta_{pk} \big)^{\top} \in \R^{p \times 1}$ corresponding to the coefficients for the $k^{th}$ outcome; and let the other parameters be defined in a similar way. We define a joint penalized estimator that simultaneously incorporates both structural equations, and estimate $\bm{\theta}$ as the minimizer of the penalized least-squares objective in Eq.~\eqref{equa:elastic.net:loss.function}. 
In \textbf{Sec.}~\ref{sec:theoretical_results}, we establish that the estimators $\big( \widehat{\bm{\beta}}, \widehat{\bm{\gamma}}, \widehat{\bm{\eta}} \big)$ and $\big( \widehat{\bm{\alpha}}, \widehat{\bm{\zeta}} \big)$ are consistent and asymptotically normally distributed. Conceptually, the estimation procedure can be viewed as a two-stage penalization scheme that preserves the multivariate dependence structure across $\X$, $\M$, and $\Y$.
In the first stage, we regularize the mapping from exposures and covariates to mediators, shrinking irrelevant rows and columns of $(\bm{\alpha},\bm{\zeta})$ toward zero, and in the second stage we similarly regularize the mapping from mediators, exposures, and covariates to outcomes through $(\bm{\beta} , \bm{\gamma} , \bm{\eta})$.

By working within a coherent MMM framework, this joint regularization yields direct estimates of the entire indirect-effect matrix $\bm{\alpha \beta}$, allowing us to study high-dimensional exposure--mediator--outcome pathways without collapsing them into low-dimensional summaries.
The technical details of the penalized objective, tuning, and theoretical guarantees are in \textbf{Sec.}~\ref{sec:estimation-procedure}. The algorithm for estimating the indirect effect matrix is in \textbf{Algorithm} \ref{alg:MMM}.

\begin{algorithm}[ht]
\caption{Estimation of Indirect Effects in a Many-to-Many-to-Many Mediation Model}
\label{alg:MMM}
\begin{algorithmic}[1]

\State \textbf{Input:}
Exposure $\mathbf X \in \mathbb{R}^{n \times q}$,
mediator matrix $\mathbf M \in \mathbb{R}^{n \times p}$,
covariate matrix $\mathbf Z \in \mathbb{R}^{n \times s}$,
and outcome matrix $\mathbf Y \in \mathbb{R}^{n \times T}$.

\State
$
\mathbf I_1 \leftarrow \bigl[\, \mathbf X \ \mathbf Z \,\bigr] \in \mathbb{R}^{n \times (q+s)} .
$

\State
$
\mathbf I_2 \leftarrow \bigl[\, \mathbf M \ \mathbf X \ \mathbf Z \,\bigr] \in \mathbb{R}^{n \times (p+q+s)} .
$

\State Fit an Elastic Net regression model with predictors $\mathbf I_1$ and responses $\mathbf M$, and denote the estimated coefficient matrix by
$
\bm \theta_1 \in \mathbb{R}^{(q+s) \times p}.
$

\State Fit a separate Elastic Net regression model with predictors $\mathbf I_2$ and responses $\mathbf Y$, and denote the estimated coefficient matrix by
$
\bm \theta_2 \in \mathbb{R}^{(p+q+s) \times T}.
$

\State $\bm \alpha \in \mathbb{R}^{q \times p}\leftarrow$ Submatrix of
$\bm \theta_1$ consisting of the rows corresponding to $\mathbf X$.

\State  $\bm \beta \in \mathbb{R}^{p \times T}\leftarrow$ Submatrix of
$\bm \theta_2$ consisting of the rows corresponding to $\mathbf M$.

\State \textbf{Output:} Coefficient matrices $\bm \alpha$ and $\bm \beta$.

\end{algorithmic}
\end{algorithm}

\subsection{Out-of-sample prediction}
\label{sec:prediction}
\vspace{-.25em}

The estimated mediation parameters facilitate multivariate outcome prediction. The out-of-sample prediction serves two purposes. First, to evaluate the generalizability and reproducibility of the MMM framework parameters, one needs to evaluate the estimated parameters in new, previously unseen samples. Second, if the out-of-sample prediction is satisfactory, then it shows that, thanks to the MMM mediation analysis, once the model is trained (using exposures, mediators, and outcomes from the training data), future prediction can be made using only the exposures. This is possible because of \textbf{Eqs.} \eqref{equa:MMAMA.model:selected.genes:equa5} and \eqref{equa:MMAMA.model:selected.genes:equa6} and this possibility is  potentially advantageous in real-world predictive tasks. For example, when obtaining data from both the exposure (e.g., genetic data) and the mediators (e.g., brain scans) is costly or time-consuming, one can predict outcomes using only exposure data and covariates. This is possible largely because the relationships between the mediators (which are not observed in the testing data) and the exposures, as well as those between the mediators and the outcomes, have been in part captured by the estimated mediation parameters $\hat{\bm{\alpha}}$ and $\hat{\bm{\beta}}$.

Specifically, one performs out-of-sample predictions in two steps. First, one predicts the estimated mediators $\left \{ \hat{m}_{ik}^{\text{new}} \right \}_k$ using the observed exposures $\mathbf{x}_i^{\text{new}}$ and covariates $\mathbf{z}_i^{\text{new}}$ following \textbf{Eq.} \ref{equa:MMAMA.model:selected.genes:equa5}. Second, one predicts the multivariate outcomes $\left \{ \hat{y}_{il}^{\text{new}} \right \}_l$ using the predicted mediators $\left \{ \hat{m}_{ik}^{\text{new}} \right \}_k$, observed exposures $\mathbf{x}_i^{\text{new}}$, and covariates $\mathbf{z}_i^{\text{new}}$ following \textbf{Eq.} \ref{equa:MMAMA.model:selected.genes:equa6}:
\begin{align}
	\hat{m}_{ik}^{\text{new}}
\label{equa:MMAMA.model:selected.genes:equa5}
    &=
    \sum_{j=1}^q x_{ij}^{\text{new}}  \hat{\alpha}_{jk}
    +  (\Z_{i}^{\text{new}})^{\top}  \hat{\bm{\zeta}}_{k},\\
	\hat{y}_{il}^{\text{new}}
    &=
	\label{equa:MMAMA.model:selected.genes:equa6}
    \sum_{k=1}^p \hat{m}_{ik}^{\text{new}} \hat{\beta}_{kl}
    +
    (\X_{i}^{\text{new}})^{\top} \hat{\bm{\gamma}}_{l}
    +
     (\Z_{i}^{\text{new}})^{\top} \hat{\bm{\eta}}_{l},
\end{align}
where the Greek letters with hats indicate their estimated values, the superscript $^{\text{new}}$ indicates observed data in the new sample, and $\hat{m}_{ik}^{\text{new}}$ and $\hat{y}_{il}^{\text{new}}$ represent the predicted $k^{th}$ mediator and predicted $l^{th}$ outcome, respectively, for subject $i$, for $1 \le k \le p$, $1 \le l \le T$, and $1 \le i \le N$.

Note that during multivariate outcome prediction using the above formulation, one only requires exposures and covariates from the testing sample but not the mediators. We present the out-of-sample results in \textbf{Sec.} \ref{sec:data_analysis}.

\subsection{A note on many-to-many-to-many causal mediation}
\label{sec:causal.mediation.model}

\bigskip

In this section, we extend classical causal mediation to a high-dimensional many--to--many--to--many setting. We first define the sequential ignorability assumptions~\citep{Robins-Richardson:2010} for the case of high-dimensional many-to-many-to-many mediation as follows. We adopt the potential outcomes framework \citep{Rubin:1974}, 
and follow the causal mediation formalism of \cite{imai2010identification,vanderweele2015explanation} 
to define direct and indirect effects in multivariate settings. 

Specifically, let $\ddot{\mathbf{m}}(x_j)$ denote the potential mediator vector under an intervention that sets $X_j = x_j$, while other exposures $X_{-j}$ take their natural values. Similarly, let $Y_l(x_j, \ddot{\mathbf{m}}(x_j))$ denote the potential outcome of the $l^{th}$ outcome 
under $X_j = x_j$ and mediators $\ddot{\mathbf{m}}$ generated by this intervention. Finally, let $Y_l(x_j, \m)$ denote the potential outcome under $X_j = x_j$ 
and mediators set to a fixed value $\m$. These definitions extend the standard potential outcomes notation to the high-dimensional many--to--many--to--many setting.


Further, we define $\text{CDE}_l(x_j, \widetilde{x}_j)$ and $\text{CDE}_l(\mathbf{x}, \widetilde{\mathbf{x}})$ as the controlled direct effect (CDE) of the $j^{th}$ exposure on the $l^{th}$ outcome and the CDE of the multivariate exposure on the $l^{th}$ outcome, respectively. We define $\text{NDE}_l(x_j, \widetilde{x}_j)$ and $\text{NDE}_l(\mathbf{x}, \widetilde{\mathbf{x}})$ as the natural direct effect (NDE) of the $j^{th}$ exposure on the $l^{th}$ outcome and the NDE of the multivariate exposure on the $l^{th}$ outcome, respectively. Finally, we define $\text{NIE}_l(x_j, \widetilde{x}_j)$ and $\text{NIE}_l(\mathbf{x}, \widetilde{\mathbf{x}})$ the natural indirect effect (NIE) of the $j^{th}$ exposure on the $l^{th}$ outcome and the NIE of the multivariate exposure on the $l^{th}$ outcome, respectively. All causal effects are defined conditionally on baseline covariates $\Z$.

We consider the sequential ignorability \citep{robins1992identifiability, vanderweele2015explanation} in the many--to--many--to--many setting as follows.

\subsubsection*{Assumptions}

\begin{enumerate}
\item[(C1)] $Y_l \left( x_j, \ddot{\mathbf{m}}(x_j) \right )\perp X_j \vert \ddot{\mathbf{z}}$: no confounding between any exposure $X_j$ and outcome $Y_l$, for $ 1 \leq j \leq q$ and $1 \leq l \leq T$;
\item[(C2)] $Y_l(\mathbf{x}, \mathbf{m}) \perp  M_k  \vert  \ddot{\mathbf{x}}, \ddot{\mathbf{z}}$: no confounding between mediators $M_k$ and outcome $Y_l$, for $ 1 \leq k \leq p$ and $ 1 \leq l \leq T$;
\item[(C3)] $M_k (x_j) \perp X_j \vert \ddot{\mathbf{z}}$: no confounding between exposure $X_j$ and mediators $\ddot{\mathbf{m}}$, for $ 1 \leq j \leq q$ and $ 1 \leq k \leq p$;
\item[(C4)] $\quad Y_l(x_j, \mathbf{m}) \perp M_k(\widetilde{x}_j) \mid \ddot{\mathbf{z}}$: no confounding between mediator $M_k$ and outcome $Y_l$ that is effected by $x_j$, for $ 1 \leq k \leq p$ and for $ 1 \leq l \leq T$. 
\end{enumerate}

Under~(C1)-(C4), we show that the average direct and indirect effects~\citep{imai2010identification,vanderweele2015explanation} are identifiable under the many--to--many--to--many setting. We present the proof of controlled and natural direct effects as well as the natural indirect effects in the Supplementary materials.

More precisely, the controlled direct effect from the $j^{th}$ exposure to the $l^{th}$ outcome is:
\begin{align*}
\text{CDE}_l(x_j, \widetilde{x}_j) =
\E \left \{ Y_l(x_j, \m) -y_l(\widetilde{x}_j, \m) \big| \ddot{\mathbf{z}} = \mathbf{z} \right \}
&= (x_j - \widetilde{x}_j) \gamma_{jl}.
\nonumber
\end{align*}

Relatedly, the controlled direct effect from the multivariate exposure to the $l^{th}$ outcome is:
\begin{align*}
\text{CDE}_l(\mathbf{x}, \widetilde{\mathbf{x}}) = \E \left \{ Y_l(\mathbf{x}, \m) -y_l(\widetilde{\mathbf{x}}, \m) \big| \ddot{\mathbf{z}} = \mathbf{z} \right \}
&= (\mathbf{x} - \widetilde{\mathbf{x}})^{\top}  \bm{\gamma}_{l}.
\nonumber
\end{align*}

The natural direct effect from the $j^{th}$ exposure to the $l^{th}$ outcome is:
\begin{align*}
\text{NDE}_l(x_j, \widetilde{x}_j) = \E \left \{ Y_l(x_j,\ddot{\mathbf{m}}(\widetilde{x}_j)) - Y_l(\widetilde{x}_j,\ddot{\mathbf{m}}(\widetilde{x}_j)) \big| \ddot{\mathbf{z}} = \mathbf{z} \right \}
&=
(x_j - \widetilde{x}_j)  \gamma_{jl}.
\end{align*}

Relatedly, the natural direct effect from the multivariate exposure to the $l^{th}$ outcome
is:
\begin{align*}
\text{NDE}_l(\mathbf{x}, \widetilde{\mathbf{x}}) =
\E \left\{ Y_l(\mathbf{x}, \ddot{\mathbf{m}}(\widetilde{\mathbf{x}}))
           - Y_l(\widetilde{\mathbf{x}}, \ddot{\mathbf{m}}(\widetilde{\mathbf{x}}))
           \,\big|\, \ddot{\mathbf{z}} = \mathbf{z} \right\}
&=
(\mathbf{x} - \widetilde{\mathbf{x}})^{\top}\, \bm{\gamma}_l.
\end{align*}

The natural indirect effect from $j^{th}$ exposure to the $l^{th}$ outcome is:
\begin{align*}
NIE_l(x_j, \widetilde{x}_j) =
\E \left \{ Y_l(x_j,\ddot{\mathbf{m}}(x_j)) - Y_l(x_j,\ddot{\mathbf{m}}(\widetilde{x}_j)) \big | \ddot{\mathbf{z}} = \mathbf{z} \right \}
&=(x_j - \widetilde{x}_j) \sum_{k=1}^p \alpha_{jk}\beta_{kl}.
\end{align*}

The natural indirect effect from the multivariate exposure to the $l^{th}$ outcome is:
\begin{align*}
NIE_l(\mathbf{x}, \widetilde{\mathbf{x}}) =
\E \left \{ Y_l(\mathbf{x},\ddot{\mathbf{m}}(\mathbf{x})) - Y_l(\mathbf{x},\ddot{\mathbf{m}}(\widetilde{\mathbf{x}})) \big | \ddot{\mathbf{z}} = \mathbf{z} \right \}
&=
(\mathbf{x} - \widetilde{\mathbf{x} })^{\top} \bm{\alpha} \bm{\beta}_l.
\end{align*}



%
%
%
%
After establishing the causal properties of the MMM mediation framework, we proceed to examine the theoretical properties of the model and its mediation estimators, including consistency and asymptotic behavior.

\section{Theoretical Results}
\label{sec:theoretical_results}

Let $\bm{\theta} = \big( \bm{\alpha}, \bm{\beta},  \bm{\zeta}, \bm{\gamma}, \bm{\eta} \big)$ denote the parameters set, where, for example, $\bm{\beta} = \big( \bm{\beta}_{1}, ... , \bm{\beta}_{T} \big) \in \R^{p \times T}$ with $\bm{\beta}_{k} = \big( \beta_{1k}, ..., \beta_{pk} \big)^{\top} \in \R^{p \times 1}$. Let   $\widehat{\bm{\alpha}},   \widehat{\bm{\zeta}},   \widehat{\bm{\beta}},   \widehat{\bm{\gamma}},    \widehat{\bm{\eta}}$ be estimators of these parameters, respectively.
The MMM estimate is then:
\begin{align}
	\widehat{\bm{\theta}}  =  \,  &\underset{\theta}{\textrm{ argmin}}  \Bigg[     \sum_{i=1}^{n} \left(    \left\|  \Y_{i}   -   \bm{\beta}^{\top} \M_{i}   -  \bm{\gamma}^{\top}  \X_{i} -    \bm{\eta}^{\top} \Z_{i}   
	\right\|_{2}^{2}
	+  \left\|  \M_{i} -     \bm{\alpha}^{\top}\X_{i}    -   \bm{\zeta}^{\top}\Z_{i}  %
	\right\|_{2}^{2}
	\right)
	\label{equa:elastic.net:loss.function}
	\\
	& \hspace{1.5cm} +  \lambda_{Y,2}   \Big(  \big\| \bm{\beta}  \big\|_{2}^{2} + \left\|   \bm{\gamma}   \right\|_{2}^{2}  + \big\|  \bm{\eta}  \big\|_{2}^{2}  \Big)  +    \lambda_{Y,1}   \Big(  \big\| \bm{\beta}  \big\|_{1} + \left\|   \bm{\gamma}  \right\|_{1}  + \big\|   \bm{\eta}  \big\|_{1}  \Big)
	\nonumber
	\\
	& \hspace{1.5cm}  +   \lambda_{M,2}  \Big(  \left\|  \bm{\alpha}  \right\|_{2}^{2} + \big\|  \bm{\zeta}  \big\|_{2}^{2} \Big)  +   \lambda_{M,1}   \Big(  \left\| \bm{\alpha}  \right\|_{1}  + \big\|   \bm{\zeta} \big\|_{1}   \Big)  \Bigg]    ,
	\nonumber
\end{align}
where $\lambda_{M,1} , \lambda_{M,2} ,  \lambda_{Y,1} ,   \lambda_{Y,2} \geq 0$ are tuning parameters.

We establish that the mediation estimators $\widehat{\bm{\alpha}}$, $\widehat{\bm{\beta}}$, and $\widehat{\bm{\gamma}}$, as well as the covariate parameters $\widehat{\bm{\zeta}}$ and $\widehat{\bm{\eta}}$, are consistent and converge in probability to the true parameters. Furthermore, these estimators are asymptotically normally distributed.
\subsection{Consistency}
\label{subsec:model.selection.consistency}

We first establish the consistency of $\widehat{\bm{\beta}}$; the consistency analysis for other parameters is similar. We then show that the mediation effects $\alpha_{jk} \beta_{kl}$ are also consistently estimated for $1 \le j \le q$, $1 \le k \le p$, and $1 \le l \le T$.

Throughout, we treat the design matrices $\X$, $\M$, and $\Z$ as fixed (non-random). In the case where these matrices are random, the results hold conditionally on $\X$, $\M$, and $\Z$.

Note that, minimizing the loss function in~\eqref{equa:elastic.net:loss.function} with respect to $\bm{\beta}$ is equivalent to minimizing the following loss function:
\begin{align}
	L_{1}(\bm{\beta})
	=  \,  &    \sum_{i=1}^{n} \left(    \left\|  \Y_{i}   -   \bm{\beta}^{\top} \M_{i}   -    \bm{\gamma}^{\top}\X_{i}  -    \bm{\eta}^{\top}\Z_{i}
	\right\|_{2}^{2}
	\right)
	+  \lambda_{Y,2}    \big\|  \bm{\beta}  \big\|_{2}^{2}   +    \lambda_{Y,1}    \big\|  \bm{\beta}   \big\|_{1}
	\nonumber
	\\
	=  \,  &  \sum_{k=1}^{T} \Bigg[    \sum_{i=1}^{n}  \left(   \big(  (\Y_{i})_{k}   -     \bm{\beta}_{k}^{\top} \M_{i}  -   \bm{\gamma}_{k}^{\top}\X_{i}   -   \bm{\eta}_{k}^{\top}\Z_{i}
	\big)^{2}
	\right)
	+  \lambda_{Y,2}  \big\|  \bm{\beta}_{k}   \big\|_{2}^{2}   +    \lambda_{Y,1}    \big\|  \bm{\beta}_{k} \big\|_{1}  \Bigg]  ,
	\label{equa:elastic.net:loss.function:beta_k}
\end{align}
and, similarly, minimizing the loss function in~\eqref{equa:elastic.net:loss.function} with respect to $\bm{\alpha}$ is equivalent to minimizing the following loss function:
\begin{align}
	L_{2}(\bm{\alpha})
	=  \,  &    \sum_{i=1}^{n} \left(    \left\|  \M_{i}   -      \bm{\alpha}^{\top}  \X_{i} -   \bm{\zeta}^{\top}\Z_{i} 
	\right\|_{2}^{2}
	\right)
	+  \lambda_{M,2}    \big\|  \bm{\alpha}  \big\|_{2}^{2}   +    \lambda_{M,1}    \big\|  \bm{\alpha}   \big\|_{1}
	\nonumber
	\\
	=  \,  &  \sum_{\ell=1}^{p} \Bigg[    \sum_{i=1}^{n}  \left(   \big(  (\M_{i})_{\ell}   -     \bm{\alpha}_{\ell}^{\top} \X_{i}  -    \bm{\zeta}_{\ell}^{\top} \Z_{i}
	\big)^{2}
	\right)
	+  \lambda_{M,2}  \big\|  \bm{\alpha}_{\ell}   \big\|_{2}^{2}   +    \lambda_{M,1}    \big\|  \bm{\alpha}_{\ell} \big\|_{1}  \Bigg].
\label{equa:elastic.net:loss.function:alpha_k}
\end{align}
Thus, for $1 \leq k \leq T$ and $1 \leq \ell \leq p$, we define the estimator of $\bm{\beta}_{k}$ and it of $\bm{\alpha}_{l}$, respectively, as:
\begin{align}
\widehat{\bm{\beta}}_{k} &:=  \underset{   {\bm{\beta}_{k} } }{ \textrm{argmin} }  \Bigg[  \sum_{i=1}^{n}  \left(   \big(  (\Y_{i})_{k}   -     \bm{\beta}_{k}^{\top} \M_{i} -   \bm{\gamma}_{k}^{\top}\X_{i}   -   \bm{\eta}_{k}^{\top} \Z_{i}
\big)^{2}
\right)
+  \lambda_{Y,2}  \big\|  \bm{\beta}_{k}   \big\|_{2}^{2}   +    \lambda_{Y,1}    \big\|  \bm{\beta}_{k} \big\|_{1}  \Bigg];
\nonumber
\\
\widehat{\bm{\alpha}}_{\ell} &:=  \underset{   {\bm{\alpha}_{\ell} } }{ \textrm{argmin} }  \Bigg[  \sum_{i=1}^{n}  \left(   \big(  (\M_{i})_{\ell}   -  \bm{\alpha}_{\ell}^{\top} \X_{i}   -    \bm{\zeta}_{\ell}^{\top} \Z_{i}
\big)^{2}   	\right)
+  \lambda_{M,2}  \big\|  \bm{\alpha}_{\ell}   \big\|_{2}^{2}   +    \lambda_{M,1}    \big\|  \bm{\alpha}_{\ell} \big\|_{1}  \Bigg].  	
\label{equa:elastic.net:estimator:beta_k:def}
\end{align}

\medskip

We follow the common notations and definitions of sign consistency used  in~\citep{Jia-Yu:2010:model.consistency.elastic-net} and references therein. More precisely,
we denote $\widehat{\bm{\beta}}_{k} =_{s} \bm{\beta}_{k}$ if the estimated vector  $\widehat{\bm{\beta}}_{k}$ and the true parameter vector $\bm{\beta}_{k}$ have the same sign element-wise.
\begin{definition}[Sign recovery property]  \em{
		Property $\mathcal{R}(\M, \bm{\beta}, \bm{\xi}, \lambda_{Y,1}, \lambda_{Y,2})$: There exists an optimal solution $\widehat{\bm{\beta}}(\lambda_{Y,1}, \lambda_{Y,2})$, depending on the given parameters $\lambda_{Y,1}$ and $\lambda_{Y,2}$, for~\eqref{equa:elastic.net:loss.function} with the property $\widehat{\bm{\beta}}_{k} =_{s} \bm{\beta}_{k}$ for all $1 \leq k \leq T$.
\leavevmode \\
Similarly, 	Property $\mathcal{R}(\X, \bm{\alpha}, \bm{\epsilon}, \lambda_{M,1}, \lambda_{M,2})$: There exists an optimal solution $\widehat{\bm{\alpha}}(\lambda_{M,1}, \lambda_{M,2})$, depending on the given parameters $\lambda_{M,1}$ and $\lambda_{M,2}$, for~\eqref{equa:elastic.net:loss.function} with the property $\widehat{\bm{\alpha}}_{\ell} =_{s} \bm{\alpha}_{\ell}$ for all $1 \leq \ell \leq p$.
}  \end{definition}
\begin{definition}[Sign consistency]  \em{
		The estimate is \textbf{sign consistent} if there exists $\widehat{\lambda}_{Y,1}, \widehat{\lambda}_{Y,2}, \widehat{\lambda}_{M,1}, \widehat{\lambda}_{M,2}$ both of which are functions of $n$ and depend on the data, such that:
		\begin{align*}
			\lim_{n   \rightarrow   +\infty}    \mathbb{P} \Big( \big( \widehat{\bm{\beta}} (\widehat{\lambda}_{Y,1} , \widehat{\lambda}_{Y,2})\big)_{k}  =_{s}   \bm{\beta}_{k} , \textrm{ for all } 1 \leq k \leq T  \Big)  &=   1 ,
            \\
            \textrm{ and }\quad \lim_{n   \rightarrow   +\infty}    \mathbb{P} \Big( \big( \widehat{\bm{\alpha}} (\widehat{\lambda}_{M,1} , \widehat{\lambda}_{M,2})\big)_{\ell}  =_{s}   \bm{\alpha}_{\ell} , \textrm{ for all } 1 \leq \ell \leq p  \Big)  &=   1.
		\end{align*}
}  \end{definition}
\underline{Remarks:} 
The estimate $\widehat{\bm{\beta}} \big( \widehat{\lambda}_{Y,1} , \widehat{\lambda}_{Y,2}   \big)$ and $\widehat{\bm{\alpha}} \big( \widehat{\lambda}_{M,1} , \widehat{\lambda}_{M,2}   \big)$ are  sign~consistent if~and~only~if:
\begin{align*}
\lim_{n   \rightarrow   +\infty} \mathbb{P} \Big( \mathcal{R} \big( \M , \bm{\beta}, \bm{\xi}, \widehat{\lambda}_{Y,1} , \widehat{\lambda}_{Y,2} \big) \Big) = 1
\, \textrm{ and }
\lim_{n   \rightarrow   +\infty} \mathbb{P} \Big( \mathcal{R} \big( \X , \bm{\alpha} , \bm{\epsilon}, \widehat{\lambda}_{M,1} , \widehat{\lambda}_{M,2} \big) \Big) = 1.
\end{align*}
We will show that under mild conditions on the relationship between $p, q$ and $n$, when $p, q$ and $n$ all go to infinity, the Elastic Irrepresentable Condition  (EIC), which will be introduced later, ensures that the regularization consistently selects the true model.

\bigskip

By Assumption (B1), the model is assumed to be ``sparse'', i.e., most of the regression coefficients $\bm{\beta}$ and $\bm{\alpha}$ are exactly zero, corresponding to predictors that are irrelevant to the response. Without loss of generality, assume the first $d$ elements of $\bm{\beta}_{k}$ and the first $\widetilde{d}$ elements of $\bm{\alpha}_{\ell}$ are non-zero with $1 \leq d < p$ and $1 \leq \widetilde{d} < q$. Let $\bm{\beta}_{k,(1)} := \big( \beta_{1k} , ..., \beta_{dk} \big)$ and $\bm{\beta}_{k,(2)} := \big( \beta_{(d+1)k} , ..., \beta_{pk} \big)$;
let
$\bm{\alpha}_{\ell,(1)} := \big( \alpha_{1\ell} ,  ..., \alpha_{\widetilde{d} \ell} \big)$ and $\bm{\alpha}_{\ell,(2)} := \big( \alpha_{(\widetilde{d}+1) \ell} , ..., \alpha_{q \ell} \big)$.
Moreover, write $\M_{i,(1)}$ and $\M_{i,(2)}$ as the first $d$ and the last $(p-d)$ element of $\M_{i}$, respectively.
Similarly, write $\X_{i,(1)}$ and $\X_{i,(2)}$ as the first $\widetilde{d}$ and the last $(q-\widetilde{d})$ element of $\X_{i}$, respectively.

We first state necessary and sufficient conditions for property $\mathcal{R}(\M, \bm{\beta}, \bm{\xi},  \lambda_{Y,1},  \lambda_{Y,2})$ and $\mathcal{R}(\X, \bm{\alpha}, \bm{\epsilon}, \lambda_{M,1},  \lambda_{M,2})$  to hold.
\begin{lem} [Component-wise KKT conditions for $\bm{\beta}_k$]

\label{lem:KKT-condition}
	For any given $\lambda_{Y,1},\lambda_{Y,2} > 0$, and noise vector $\bm{\xi} \in \R^{T}$, property $\mathcal{R}(\M, \bm{\beta}, \bm{\xi}, \lambda_{Y,1}, \lambda_{Y,2})$ holds if and only if:
	\em{  \begin{align}
			\Big\|  2  \sum_{i=1}^{n}  \Big[   \M_{i,(2)} \M_{i,(1)}^{\top}  \Big( \sum_{\ell = 1}^{n}  \M_{\ell , (1)} \M_{\ell , (1)}^{\top}  +  \lambda_{Y,2} \mathbf{I} \Big)^{-1}
			&\Big( \M_{i,(1)} \bm{\xi}_{ik}  -  \dfrac{\lambda_{Y,1}}{2} \textrm{sign} \big( \bm{\beta}_{k,(1)} \big)  -  \lambda_{Y,2} \bm{\beta}_{k,(1)} \Big)
			\nonumber
			\\
			&+  \M_{i,(2)}   \bm{\xi}_{ik}  \Big]   \Big\|_{\infty}
			\leq  \lambda_{Y,1},
			\label{lem:KKT-condition:condition1}
	\end{align} }
	and
	{  \begin{align}
			\textrm{sign} \Bigg(  \Big(  \sum_{\ell = 1}^{n} \M_{\ell,(1)}   \M_{\ell , (1)}^{\top}   +  \lambda_{Y,2} \mathbf{I}      \Big)^{-1}
			&\Big[   \sum_{i=1}^{n} \Big( \M_{i, (1)} \M_{i, (1)}^{\top}  \bm{\beta}_{k , (1)}  +  \M_{i,(1)}    \bm{\xi}_{ik}  \Big)
			\nonumber
			\\
			& \qquad  -   \dfrac{\lambda_{Y,1}}{2} \textrm{sign} \big( \bm{\beta}_{k,(1)} \big)  \Big]    \Bigg)  =  \textrm{sign} \big( \bm{\beta}_{k,(1)} \big) ,
			\label{lem:KKT-condition:condition2}
		\end{align}
		for all $1 \leq  k  \leq  T$. }
\end{lem}
\underline{Remarks:} The proof of Lemma~\ref{lem:KKT-condition} follows the component-wise KKT arguments in \citep{Jia-Yu:2010:model.consistency.elastic-net}.
\begin{lem} [Component-wise KKT conditions for $\bm{\alpha}_l$]
\label{lem:KKT-condition-alpha}
For any given $ \lambda_{M,1},\lambda_{M,2}  > 0$, and noise vector $\bm{\epsilon} \in \R^{p}$, property $\mathcal{R}(\M, \bm{\alpha}, \bm{\epsilon}, \lambda_{M,1}, \lambda_{M,2})$ holds if and only if:
\em{  \begin{align}
    \Big\|  2  \sum_{i=1}^{n}  \Big[   \X_{i,(2)} \X_{i,(1)}^{\top}  \Big( \sum_{l = 1}^{n}  \X_{l , (1)} \X_{l , (1)}^{\top}  +  \lambda_{M,2} \mathbf{I} \Big)^{-1}
    &\Big( \X_{i,(1)} \epsilon_{i \ell}  -  \dfrac{\lambda_{M,1}}{2} \textrm{sign} \big( \bm{\alpha}_{\ell,(1)} \big)  -  \lambda_{M,2} \bm{\alpha}_{\ell,(1)} \Big)
    \nonumber
    \\
    &+  \X_{i,(2)}   \epsilon_{i \ell}   \Big]   \Big\|_{\infty}
    \leq  \lambda_{M,1},
    \label{lem:KKT-condition-alpha:condition1}
\end{align} }
and
{  \begin{align}
    \textrm{sign} \Bigg(  \Big(  \sum_{l = 1}^{n} \X_{l,(1)}   \X_{l , (1)}^{\top}   +  \lambda_{M,2} \mathbf{I}      \Big)^{-1}
    &\Big[   \sum_{i=1}^{n} \Big( \X_{i, (1)} \X_{i, (1)}^{\top}  \bm{\alpha}_{\ell , (1)}  +  \X_{i,(1)}   \epsilon_{i \ell} \Big)
    \nonumber
    \\
    & \qquad  -   \dfrac{\lambda_{M,1}}{2} \textrm{sign} \big( \bm{\alpha}_{\ell,(1)} \big)  \Big]    \Bigg)  =  \textrm{sign} \big( \bm{\alpha}_{\ell,(1)} \big) ,
\label{lem:KKT-condition-alpha:condition2}
\end{align}
for all $1 \leq  \ell  \leq  p$. }
\end{lem}
\underline{Remarks:} The proof of Lemma~\ref{lem:KKT-condition-alpha} is analogous to Lemma~\ref{lem:KKT-condition}.
\bigskip

Now, let $\overrightarrow{\mathbf{b}} := \textrm{sign} \big( \bm{\beta}_{k,(1)} \big)$, $\overrightarrow{\mathbf{b}}_{\bm{\alpha}} := \textrm{sign} \big( \bm{\alpha}_{\ell,(1)} \big)$ and, for $j \in \{1, ..., d\}$, denote by $\mathbf{e}_{j} \in \R^{d \times 1}$ the vector with $1$ in the $j^{th}$ position and zeroes elsewhere. Let $D =  \left\{ 1 , ..., d \right\}$ and $D^{c}$ be  the subset and its complement in $\{1, ..., p\}$, respectively. For each $j \in D$ and $m \in D^{c}$, let us define:
\begin{align*}
	U_{j} &:=  \sum_{i = 1}^{n} \mathbf{e}_{j}^{\top} \Big( \sum_{\ell = 1}^{n} \M_{\ell , (1)} \M_{\ell , (1)}^{\top}  +   \lambda_{Y,2} \mathbf{I} \Big)^{-1} \Big[ \M_{i,(1)} (\bm{\xi}_{i})_{k}  -  \dfrac{\lambda_{Y,1}}{2} \overrightarrow{\mathbf{b}}  \Big]   ;
	\\
	V_{m} &:=   2 \sum_{i = 1}^{n} (\M_{i})_{m} \left\{ \M_{i,(1)}^{\top}  \Big( \sum_{\ell = 1}^{n}  \M_{\ell,(1)} \M_{\ell,(1)}^{\top}  + \lambda_{Y,2} \mathbf{I} \Big)^{-1} \Big( \dfrac{\lambda_{Y,1}}{2} \overrightarrow{\mathbf{b}}  +  \lambda_{Y,2} \bm{\beta}_{k , (1)}   \Big)  \right.
	\\
	& \left. - \Big[ \M_{i,(1)}^{\top}  \Big( \sum_{\ell = 1}^{n}  \M_{\ell,(1)} \M_{\ell,(1)}^{\top}  + \lambda_{Y,2} \mathbf{I} \Big)^{-1} \M_{i,(1)}  - 1 \Big]   (\bm{\xi}_{i})_{k}   \right\} .
\end{align*}
These random variables will play a crucial role in our analysis.
In particular, condition~\eqref{lem:KKT-condition:condition1} holds if and only if the event:
\begin{align*}
	\mathcal{M}(V)  :=  \left\{ \max_{m \in D^{c}} |V_{m}| \leq  \lambda_{Y,1}  \right\}
\end{align*}
holds. On the other hand, if we define:   \[
\rho :=  \underset{1 \leq j \leq d}{ \min } \, \Big| \mathbf{e}_{j}^{\top} \Big[ \sum_{i = 1}^{n}   \Big( \dfrac{1}{n} \sum_{\ell = 1}^{n}  \M_{\ell,(1)}  \M_{\ell,(1)}^{\top}  +  \dfrac{\lambda_{Y,2}}{n} \mathbf{I} \Big)^{-1}   \M_{i,(1)}   \M_{i,(1)}^{\top}   \bm{\beta}_{k , (1)}    \Big]  \Big| ,
\]
then the event:
\begin{align}
	\mathcal{M}(U) :=  \left\{ \max_{j \in  D} |U_{j}| \leq   \rho    \right\}
\end{align}
is sufficient to guarantee that condition~\eqref{lem:KKT-condition:condition2} holds,  if $\lambda_{Y,2}$ is chosen such that:
\begin{align}
	\textrm{sign} \Bigg(  \sum_{i=1}^{n} \Big(   \sum_{\ell = 1}^{n}  \M_{\ell,(1)} \M_{\ell,(1)}^{\top}  + \lambda_{Y,2} \mathbf{I} \Big)^{-1}  \M_{i,(1)}   \M_{i,(1)}^{\top}   \bm{\beta}_{k,(1)}   \Bigg)  =  \textrm{sign}\big( \bm{\beta}_{k,(1)} \big).
	\label{lem:KKT-condition:condition2:equivalent.condition:choose.lambdaY2}
\end{align}
Throughout this paper, we constrain $\lambda_{Y,2}$ such that~\eqref{lem:KKT-condition:condition2:equivalent.condition:choose.lambdaY2} holds.

\bigskip

On the model selection consistency of the estimator, we define the Elastic Irrepresentable Condition (EIC) in the framework of high-dimensional MMM mediation analysis as follows.
Let $\mathbf{C} \equiv \mathbf{C}(n) =  \sum_{i=1}^{n}  \dfrac{ \M_{i}  \M_{i}^{\top} }{n}$ and $\widetilde{\mathbf{C}} \equiv \widetilde{\mathbf{C}}(n) =  \sum_{i=1}^{n}  \dfrac{ \X_{i}  \X_{i}^{\top} }{n}$. We can rewrite $\mathbf{C}$ and $\widetilde{\mathbf{C}}$ in the block-wise form:
\begin{align*}
	\mathbf{C} = \left( \begin{array}{cc}
		\mathbf{C}_{11}  &  \mathbf{C}_{12}
		\\
		\mathbf{C}_{21}  &  \mathbf{C}_{22}
	\end{array}
	\right)
    \quad \textrm{ and } \quad
    \widetilde{\mathbf{C}} = \left( \begin{array}{cc}
		\widetilde{\mathbf{C}}_{11}  &  \widetilde{\mathbf{C}}_{12}
		\\
		\widetilde{\mathbf{C}}_{21}  &  \widetilde{\mathbf{C}}_{22}
	\end{array}
	\right),
\end{align*}
where
\begin{itemize}
\item $\mathbf{C}_{11} =  \sum_{i=1}^{n}  \dfrac{\M_{i,(1)} \M_{i,(1)}^{\top} }{n}$, $\mathbf{C}_{12} =  \sum_{i=1}^{n}  \dfrac{\M_{i,(1)} \M_{i,(2)}^{\top} }{n}$, $\mathbf{C}_{21} =  \sum_{i=1}^{n}  \dfrac{\M_{i,(2)} \M_{i,(1)}^{\top} }{n}$ and $\mathbf{C}_{22} =  \sum_{i=1}^{n}  \dfrac{\M_{i,(2)} \M_{i,(2)}^{\top} }{n}$;
\item  $\widetilde{\mathbf{C}}_{11} =  \sum_{i=1}^{n}  \dfrac{\X_{i,(1)} \X_{i,(1)}^{\top} }{n}$, $\widetilde{\mathbf{C}}_{12} =  \sum_{i=1}^{n}  \dfrac{\X_{i,(1)} \X_{i,(2)}^{\top} }{n}$, $\widetilde{\mathbf{C}}_{21} =  \sum_{i=1}^{n}  \dfrac{\X_{i,(2)} \X_{i,(1)}^{\top} }{n}$ and $\widetilde{\mathbf{C}}_{22} =  \sum_{i=1}^{n}  \dfrac{\X_{i,(2)} \X_{i,(2)}^{\top} }{n}$ .
\end{itemize}
\begin{definition}[Elastic Irrepresentable Condition (EIC)]  
\em{
There exists a  constant $0 < \Psi < 1$ (which does not change with $n$), with:
\begin{align}
    \left\|  \mathbf{C}_{21} \Big(  \mathbf{C}_{11} + \dfrac{\lambda_{Y,2}}{n}\mathbf{I} \Big)^{-1}  \Big( \textrm{sign}\big( \bm{\beta}_{k , (1)} \big)  +  \dfrac{2 \lambda_{Y,2}}{\lambda_{Y,1}} \bm{\beta}_{k , (1)}  \Big) \right\|_{\infty}  &\leq  1 - \Psi;
\nonumber
\\
\left\|  \widetilde{\mathbf{C}}_{21} \Big(  \widetilde{\mathbf{C}}_{11} + \dfrac{\lambda_{M,2}}{n}\mathbf{I} \Big)^{-1}  \Big( \textrm{sign}\big( \bm{\alpha}_{\ell , (1)} \big)  +  \dfrac{2 \lambda_{M,2}}{\lambda_{M,1}} \bm{\alpha}_{\ell , (1)}  \Big) \right\|_{\infty}  &\leq  1 - \Psi.
\label{def:EIC:beta_k:inequa}
\end{align}
} \end{definition}
\underline{Remarks:} The Elastic Irrepresentable Condition (EIC) is a direct extension of the Irrepresentable Condition in~\citep{Jia-Yu:2010:model.consistency.elastic-net} to the high-dimensional MMM mediation framework.
Theorem~\ref{thm:consistency.elastic.net.estimator:beta_k} below shows that, under mild conditions on the relationship between the scalings of $p,q$ and $n$, the EIC is sufficient for the property of $\mathcal{R}(\M, \bm{\beta}, \bm{\xi}, \lambda_{Y,1}, \lambda_{Y,2})$ to hold with probability tending to $1$ as $n \rightarrow \infty$. It extends the results of~\citep{Jia-Yu:2010:model.consistency.elastic-net} to the high-dimensional MMM mediation setting.
\begin{thm}[Sign consistency of $\bm{\beta}$]
	\label{thm:consistency.elastic.net.estimator:beta_k}
	For $1 \leq i \leq n$, suppose that $\ddot{\mathbf{y}}_{i} =  \bm{\beta}^{\top} \M_{i} + \bm{\gamma}^{\top} \X_{i}  + \bm{\eta}^{\top} \Z_{i}  + \bm{\xi}_{i}$, where the vector $\M_{i}$ is normalized to $l_{2}$-norm $\sqrt{n}$ and $\bm{\xi}_{i} {\sim} \mathcal{N}(\bm{0} ; \mathbf{I})$. Assume the  EIC (defined in~\eqref{def:EIC:beta_k:inequa})  holds. Consider $d > 1$ and $p - 1 > 1$. If $\rho = \underset{1 \leq j \leq d}{ \min } \, \Big| \mathbf{e}_{j}^{\top} \Big[   \Big( \mathbf{C}_{11} + \frac{\lambda_{Y,2}}{n} \mathbf{I} \Big)^{-1} \Big(\mathbf{C}_{11} \bm{\beta}_{k , (1)} \Big)  \Big]  \Big| $, $C_{\min} = \Lambda_{\min}(\mathbf{C}_{11}) + \dfrac{\lambda_{Y,2}}{n}$ where  $\Lambda_{\min}(\mathbf{A})$ denotes the minimal eigenvalue of matrix $\mathbf{A}$, and $\lambda_{Y,1}, \lambda_{Y,2}$ are chosen such that:
	\begin{enumerate}
		\item[(a)] $\displaystyle   
		\dfrac{ \sqrt{ \log(p-d) } }{ \lambda_{Y,1}  \sqrt{n} }  \overset{n  \rightarrow \infty}{\longrightarrow}  0 $  ,
		\item[(b)] $\displaystyle   \dfrac{1}{\rho} \left\{ 8 \sqrt{ \dfrac{\log(d) }{ n C_{\min} } } +  \dfrac{\lambda_{Y,1}}{2 n }  \left\| \Big( \mathbf{C}_{11} + \dfrac{\lambda_{Y,2}}{n} \mathbf{I} \Big)^{-1} \,  \overrightarrow{\mathbf{b}}   \right\|_{\infty}  \right\}      \overset{n   \rightarrow   \infty}{\longrightarrow}  0 $,
		or equivalently,
		\[
		\dfrac{1}{\rho}  \Bigg\{ 8 \sqrt{ \dfrac{ \log(d)   }{ n C_{\min} } }  + \dfrac{\lambda_{Y,1}}{2 n }  \left\|  \big( \dfrac{1}{n} \sum_{\ell = 1}^{n} \M_{\ell , (1)} \M_{\ell , (1)}^{\top} +  \dfrac{\lambda_{Y,2}}{n} \mathbf{I} \big)^{-1}    \overrightarrow{\mathbf{b}}   \right\|_{\infty}
		\Bigg\}  \overset{n   \rightarrow   \infty}{\longrightarrow}  0  ,
		\]
	\end{enumerate}
	then $\lim_{n   \rightarrow   +\infty} \mathbb{P} \Big( \mathcal{R} \big( \M, \bm{\beta}, \bm{\xi}, \lambda_{Y,1}, \lambda_{Y,2} \big) \Big) = 1$.
\end{thm}
\begin{thm}[Sign consistency of $\bm{\alpha}$]
\label{thm:consistency.elastic.net.estimator:alpha_l}
For $1 \leq i \leq n$, suppose that $\ddot{\mathbf{m}}_{i} =  \bm{\alpha}^{\top} \X_{i} +   \bm{\zeta}^{\top}\Z_{i} + \bm{\epsilon}_{i}$, where the vector $\X_{i}$ is normalized to $l_{2}$-norm $\sqrt{n}$ and $\bm{\epsilon}_{i} {\sim} \mathcal{N}(\bm{0} ; \mathbf{I})$. Assume the  EIC (defined  in~\eqref{def:EIC:beta_k:inequa})  holds. Consider $\widetilde{d} > 1$ and $q - 1 > 1$. If $\rho_{\bm{\alpha}} = \underset{1 \leq j \leq \widetilde{d}}{ \min } \, \Big| \mathbf{e}_{j}^{\top} \Big[   \Big( \widetilde{\mathbf{C}}_{11} + \frac{\lambda_{M,2}}{n} \mathbf{I} \Big)^{-1} \Big(\widetilde{\mathbf{C}}_{11} \bm{\alpha}_{\ell , (1)} \Big)  \Big]  \Big| $, $\widetilde{C}_{\min} = \Lambda_{\min}(\widetilde{\mathbf{C}}_{11}) + \dfrac{\lambda_{M,2}}{n}$ where $\lambda_{M,1}, \lambda_{M,2}$ are chosen such that:
\begin{enumerate}
    \item[(a)] $\displaystyle
    \dfrac{ \sqrt{ \log(q-\widetilde{d}) } }{ \lambda_{M,1}  \sqrt{n} }  \overset{n  \rightarrow \infty}{\longrightarrow}  0 $  ,
    \item[(b)] $\displaystyle   \dfrac{1}{\rho_{\alpha}} \left\{ 8 \sqrt{ \dfrac{\log(\widetilde{d})}{ n \widetilde{C}_{\min} } } +  \dfrac{\lambda_{M,1}}{2 n }  \left\| \Big( \widetilde{\mathbf{C}}_{11} + \dfrac{\lambda_{M,2}}{n} \mathbf{I} \Big)^{-1} \,  \overrightarrow{\mathbf{b}}_{\bm{\alpha}}  \right\|_{\infty}  \right\}      \overset{n   \rightarrow   \infty}{\longrightarrow}  0 $,
    or equivalently,
    \[
    \dfrac{1}{\rho_{\alpha}}  \Bigg\{ 8 \sqrt{ \dfrac{ \log(\widetilde{d})}{ n \widetilde{C}_{\min} } }  + \dfrac{\lambda_{M,1}}{2 n }  \left\|  \big( \dfrac{1}{n} \sum_{l = 1}^{n} \X_{l , (1)} \X_{l , (1)}^{\top} +  \dfrac{\lambda_{M,2}}{n} \mathbf{I} \big)^{-1}    \overrightarrow{\mathbf{b}_{\bm{\alpha}}}   \right\|_{\infty}
    \Bigg\}  \overset{n   \rightarrow   \infty}{\longrightarrow}  0  ,
    \]
\end{enumerate}
then $\lim_{n   \rightarrow   +\infty} \mathbb{P} \Big( \mathcal{R} \big( \X, \bm{\alpha}, \bm{\epsilon}, \lambda_{M,1}, \lambda_{M,2} \big) \Big) = 1$.
\end{thm}
\noindent
Note that proofs of Theorem~\ref{thm:consistency.elastic.net.estimator:beta_k} and Theorem~\ref{thm:consistency.elastic.net.estimator:alpha_l} are similar.
A proof of Theorem~\ref{thm:consistency.elastic.net.estimator:beta_k} is given in the Supplementary materials. 
\begin{proposition}
\label{thm:consistency.elastic.net.estimator:alpha.beta}
For any $1 \leq  k  \leq  T$, with probability tending to $1$, $(\widehat{\bm{\alpha}}  \widehat{\bm{\beta}})_{k}$ converges to $(\bm{\alpha} \bm{\beta})_{k}$.
\end{proposition}
\noindent
A proof of Proposition~\ref{thm:consistency.elastic.net.estimator:alpha.beta} is given the Supplementary materials.
\subsection{Asymptotic studies}
\label{subsec:asymptotic.normality}
For $1 \leq  k  \leq T$ and $1 \leq \ell \leq q$, to establish the asymptotic normality of the estimators   $\widehat{\bm{\alpha}}_{\ell}, \widehat{\bm{\zeta}}_{\ell}, \widehat{\bm{\beta}}_{k}, \widehat{\bm{\gamma}}_{k}$, and $ \widehat{\bm{\eta}}_{k}$, we assume the following regularity conditions, which are standard in high-dimensional penalized regression (see~\cite[Section 3]{Zou-Zhang:2009}).
\begin{enumerate}
	\item[(A1)] There exist two positive constants $\delta$ and $\Delta$ such that for all $n \in \N^{*}$,
	\begin{align*}
		\delta  &\leq  \Lambda_{\min} \Big( \frac{1}{n} \sum_{l = 1}^{n} \M_{l,(1)}  \M_{l,(1)}^{\top} \Big)
		\leq  \Lambda_{\max} \Big( \frac{1}{n}  \sum_{l = 1}^{n} \M_{l,(1)}  \M_{l,(1)}^{\top} \Big)  \leq  \Delta ;
		\\
		\delta  &\leq  \Lambda_{\min} \Big( \frac{1}{n} \sum_{l = 1}^{n} \X_{l,(1)}  \X_{l,(1)}^{\top} \Big)
		\leq  \Lambda_{\max} \Big( \frac{1}{n}  \sum_{l = 1}^{n} \X_{l,(1)}  \X_{l,(1)}^{\top} \Big)  \leq  \Delta ;
	\end{align*}
	\item[(A2)] $\lim_{n   \rightarrow   +\infty} \dfrac{\max_{1 \leq i \leq n} \sum_{k=1}^{p} \big( \M_{i} \big)_{k}^{2} }{ n }  =  0$ and $\lim_{n   \rightarrow   +\infty} \dfrac{\max_{1 \leq i \leq n} \sum_{\ell=1}^{q} \big( \X_{i} \big)_{\ell}^{2} }{ n }  =  0$;
	\item[(A3)] $\max_{1 \leq i \leq n} \E \Big( \left\| \bm{\xi}_{i}  \right\|_{\infty}^{2 + \vartheta}   \Big) < \infty$ and $\max_{1 \leq i \leq n} \E \Big( \left\| \bm{\epsilon}_{i}  \right\|_{\infty}^{2 + \vartheta}   \Big) < \infty$, for some $\vartheta > 0$;
	\item[(A4)] $\lim_{n   \rightarrow   +\infty}  \dfrac{\log(p)}{\log(n)} = \nu_{1}$ and $\lim_{n   \rightarrow   +\infty}  \dfrac{\log(q)}{\log(n)} = \nu_{2}$, for some $0 \leq \nu_{1}, \nu_{2} < 1$;
	\item[(A5)] $\lim_{n   \rightarrow   +\infty}  \dfrac{\lambda_{Y,1}}{ \sqrt{n} } = 0$ and $\lim_{n   \rightarrow   +\infty} \dfrac{\lambda_{Y,2}}{n} = 0$;
	\\
	$\lim_{n   \rightarrow   +\infty}  \dfrac{\lambda_{M,1}}{ \sqrt{n} } = 0$ and $\lim_{n   \rightarrow   +\infty} \dfrac{\lambda_{M,2}}{n} = 0$;
	\item[(A6)] $\lim_{n   \rightarrow   +\infty}  \dfrac{\lambda_{Y,2}}{ \sqrt{n} } \left\| \bm{\beta}_{k , (1)} \right\|_{2} = 0$, $\lim_{n   \rightarrow   +\infty}  \dfrac{\lambda_{Y,2}}{ \sqrt{n} } \left\| \bm{\gamma}_{k , (1)} \right\|_{2} = 0$ and $\lim_{n   \rightarrow   +\infty}  \dfrac{\lambda_{Y,2}}{ \sqrt{n} } \left\| \bm{\eta}_{k , (1)} \right\|_{2} \mbox{= 0}$;
	\\
	$\lim_{n   \rightarrow   +\infty}  \dfrac{\lambda_{M,2}}{ \sqrt{n} } \left\| \bm{\alpha}_{\ell , (1)} \right\|_{2} = 0$ and $\lim_{n   \rightarrow   +\infty}  \dfrac{\lambda_{M,2}}{ \sqrt{n} } \left\| \bm{\zeta}_{\ell , (1)} \right\|_{2} = 0$.
\end{enumerate}
The following Proposition~\ref{prop:upper-bound:elastic.net.beta_{k}} provides an upper bound for the mean squared error (MSE) of the estimator $\widehat{\bm{\beta}}_{k}$.
\begin{proposition} [Error bound for $\bm{\beta}$]
	\label{prop:upper-bound:elastic.net.beta_{k}}
	With the estimator  $\widehat{\bm{\beta}}_{k}$ defined in~\eqref{equa:elastic.net:estimator:beta_k:def} for nonnegative parameters $\lambda_{Y,1}$ and  $\lambda_{Y,2}$, under condition (A1):
	\begin{align*}
		\E \Big( \big\| \widehat{\bm{\beta}}_{k} - \bm{\beta}_{k}   \big\|_{2}^{2}   \Big)
		& \leq  \dfrac{  4  \lambda_{Y,2}^{2}   \big\| \bm{\beta}_{k}   \big\|_{2}^{2}  +  8 \, np  \,     \big\| \M \big\|_{\infty}^2   +  \lambda_{Y,1}^{2} \,  p
		}{ \big(  \delta \,  n  +   \lambda_{Y,2}  \big)^{2} }.
	\end{align*}
\end{proposition}
A proof of Proposition~\ref{prop:upper-bound:elastic.net.beta_{k}}  is given in the Supplementary materials.  

\underline{Remark:} This preliminary error bound plays a key role in deriving the asymptotic normality of $\widehat{\bm{\alpha}}_\ell$ and $\widehat{\bm{\beta}}_k$ in Theorem~\ref{thm:asymptotic.normality.elastic.net.estimator:beta_k}, as it provides control over the mean squared error of the penalized estimators.

Then, adapting arguments from~\cite{Zou-Zhang:2009}, we have the following.

\medskip

\begin{thm}[Asymptotic normality of $\bm{\alpha}$ and $\bm{\beta}$]\label{thm:asymptotic.normality.elastic.net.estimator:beta_k}
	Under conditions \textrm{(A1)-(A6)}, for any $1 \leq k \leq T$ and $1 \leq \ell \leq p$, we have:
	\begin{align*}
		\sqrt{n}  \Big[   \mathbf{v}^{\top}  \Big( \mathbf{I} + \lambda_{Y,2} \big(   \sum_{\ell = 1}^{n}   \M_{\ell , (1)} \M_{\ell , (1)}^{\top} \big)^{-1}  \Big)  \Big(  \sum_{\ell = 1}^{n}    \M_{\ell , (1)} \M_{\ell , (1)}^{\top} \Big)^{1/2} \Big( \widehat{\bm{\beta}}_{k,(1)} - \bm{\beta}_{k,(1)} \Big)   \Big]
		\overset{ d }{\longrightarrow} \mathcal{N} \big(0 ; 1 \big),
	\end{align*}
    \textrm{ and } \quad
    \begin{align*}
        \sqrt{n}  \Big[   \mathbf{v}^{\top}  \Big(\mathbf{I} + \lambda_{M,2} \big(   \sum_{l = 1}^{n}   \X_{l , (1)} \X_{l , (1)}^{\top} \big)^{-1}  \Big)  \Big(  \sum_{l = 1}^{n}    \X_{l , (1)} \X_{l , (1)}^{\top} \Big)^{1/2} \Big( \widehat{\bm{\alpha}}_{\ell,(1)} - \bm{\alpha}_{\ell,(1)} \Big)   \Big]
		\overset{ d }{\longrightarrow} \mathcal{N} \big( 0;1\big),
	\end{align*}
	where $\mathbf{v}$ is a vector of norm $1$.
\end{thm}

\noindent
\underline{Remark:} This result extends the asymptotic normality of the elastic net estimator to high-dimensional MMM mediation estimators under standard regularity conditions. A proof of Theorem~\ref{thm:asymptotic.normality.elastic.net.estimator:beta_k} is provided in the Supplementary materials.   
\subsection{Asymptotic normality of the many--to--many--to--many mediation effect}
To study the mediation effect $(\bm{\alpha} \bm{\beta})$, we first rewrite the model ($\mathfrak{M}$) in~\eqref{equa:MMAMA.model:selected.genes:equa1}-\eqref{equa:MMAMA.model:selected.genes:equa2}, for $1 \leq i \leq n$, as below:
\begin{align}
	\ddot{\mathbf{y}}_{i}  &=  { \bm{\beta}}^{\top} \Big(  {\bm{\alpha}}^{\top}\X_{i}  +   \bm{\zeta}^{\top} \Z_{i}  +  \bm{\epsilon}_{i}  \Big)  +     \bm{\gamma}^{\top} \X_{i}  +     \bm{\eta}^{\top} \Z_{i} +  \bm{\xi}_{i}
	\nonumber
	\\
	&=    \Big( \bm{\beta}^{\top} \bm{\alpha}^{\top}    +  \bm{\gamma}^{\top} \Big) { \X_{i} }
	+    \Big(  \bm{\beta}^{\top} \bm{\zeta}^{\top}     +   \bm{\eta}^{\top}  \Big) \Z_{i}  +  \Big( \bm{\beta}^{\top} \bm{\epsilon}_{i}  +  { \bm{\xi} }_{i}  \Big) ,
\end{align}
with $ \bm{\epsilon}_{i}   {\sim} \mathcal{N}(\bm{0} ; \mathbf{I}) $, $ \bm{\xi}_{i}  {\sim} \mathcal{N}(\bm{0} ; \mathbf{I})$, and $ \bm{\epsilon}_{i} $ and $\bm{\xi}_{i}$ are independent for all $1 \leq i \leq n$. Thus, 
we obtain:
\[
\varkappa_{i} := \bm{\beta}^{\top} \bm{\epsilon}_{i}   +  \bm{\xi}_{i}
{\sim} \mathcal{N}(\bm{0} ;  \bm{\Sigma}_{\varkappa})  ,
\quad
\textrm{ with $\bm{\Sigma}_{\varkappa}= \bm{\beta}^{\top}\bm{\beta}  + \mathbf{I}$. }
\]

The following theorem proves the asymptotic normality for the estimator of the mediation effect $(\bm{\alpha} \bm{\beta})_{k}$.
\begin{thm}[Asymptotic normality of the mediation effect]\label{thm:asymptotic.normality.elastic.net.estimator:alpha.beta_k}
	For any $1 \leq  k  \leq  T$, under conditions \textrm{(A1)-(A6)}, the estimator $(\widehat{\bm{\alpha}}  \widehat{\bm{\beta}})_{k,(1)}$ has the following asymptotic normality:
	\begin{align*}
		\sqrt{n}  \Big[   \mathbf{v}^{\top}  \Big( \mathbf{I} +  \lambda_{Y,2} \big(  \sum_{\ell = 1}^{n}   \X_{\ell , (1)} \X_{\ell , (1)}^{\top} \big)^{-1}\Big) \Big( \sum_{\ell = 1}^{n}    \X_{\ell , (1)} \X_{\ell , (1)}^{\top} \Big)^{1/2}  & \Big( (\widehat{\bm{\alpha}}  \widehat{\bm{\beta}})_{k,(1)} - (\bm{\alpha} \bm{\beta})_{k,(1)} \Big)   \Big]
		\\
		&\overset{ d }{\longrightarrow} \mathcal{N} \big(0 ; \bm{\beta}_k^{\top}  \bm{\beta}_k \big),
	\end{align*}
where $\bm {v}$ is a vector of norm $1$.
\end{thm}
\noindent
A proof of Theorem~\ref{thm:asymptotic.normality.elastic.net.estimator:alpha.beta_k} is given in the Supplementary materials.  

\section{Simulation Studies}
\label{sec:simulation}

\subsection{Objectives}

We conducted simulation studies to evaluate the finite-sample behavior of the proposed MMM mediation estimator under controlled and structured multivariate settings.
The experiments examine five questions:
\textit{(i)}~the accuracy of recovery of the coefficient matrices $(\bm{\alpha},\bm{\beta})$ and the indirect-effect matrix $\bm{\alpha}\bm{\beta}$,
\textit{(ii)}~stability of the estimated indirect effects across bootstrap runs,
\textit{(iii)}~false-positive behavior on null paths,
\textit{(iv)}~robustness of the MMM estimator to varying noise levels and sample sizes,
and \textit{(v)}~empirical agreement between the estimated indirect-effect matrices and the asymptotic theory in \textbf{Sec.}~\ref{sec:theoretical_results}.
\textbf{Fig.}~\ref{fig:simulation} summarizes these results.

\subsection{Data-generating mechanisms and parameter design}

We generated data from the multivariate LSEM in~\eqref{equa:MMAMA.model:selected.genes:equa1}--\eqref{equa:MMAMA.model:selected.genes:equa2}.
Each simulated dataset contains $q=20$ exposures, $p=20$ mediators, $T=10$ outcomes, and $s=2$ covariates; namely, $\mathbf X,\mathbf M\in\mathbb{R}^{n\times 20}$, $\mathbf Y\in\mathbb{R}^{n\times 10}$, and $\mathbf Z\in\mathbb{R}^{n\times 2}$.
The ground-truth matrices $\bm{\alpha}_0$ and $\bm{\beta}_0$ were designed to be structured and sparse, with block-like nonzero regions motivated by imaging--genetics settings, which in turn induces an indirect-effect matrix $\bm{\alpha}_0\bm{\beta}_0$ with clustered pathways.
The covariate matrix $\mathbf Z$ was generated from a truncated normal variable (age) and a Bernoulli variable (sex).
The exposure matrix $\mathbf X$ was sampled from a correlated multivariate normal distribution with decaying off-diagonal dependence.
Conditional on $(\mathbf X,\mathbf Z)$, the mediators and outcomes were generated from the two structural equations using fixed coefficients $(\bm{\alpha}_0,\bm{\beta}_0,\bm{\gamma}_0,\bm{\zeta}_0,\bm{\eta}_0)$.

Noise magnitude was indexed by $\sigma\in\{50,100,200,500,1000\}$, corresponding respectively to low, low-to-moderate, moderate, high, and very high noise.
We set the mediator noise as $\bm{\varepsilon}_i=\sigma \bm{\varepsilon}_i^{(0)}$ and the outcome noise as $\bm{\xi}_i=\sigma^2\bm{\xi}_i^{(0)}$, where $\bm{\varepsilon}_i^{(0)}\sim\mathcal{N}(\bm{0},\mathbf{I})$ and $\bm{\xi}_i^{(0)}\sim\mathcal{N}(\bm{0},\mathbf{I})$.
The quadratic scaling in the outcome equation was chosen deliberately to make the second stage harder than the first and to reflect the fact that variability introduced at the mediator level propagates into the outcome model, so that accurate recovery of $\beta$ and especially of $\alpha\beta$ is tested under a stringent regime.

We varied the sample size over
$n\in\{50,100,500,1000,5000,10000\}$.
For each $(n,\sigma)$ configuration, we estimated $(\widehat{\bm{\alpha}},\widehat{\bm{\beta}},\widehat{\bm{\alpha}}\widehat{\bm{\beta}})$ using the proposed MMM mediation analysis procedure.
For the stability analysis in \textbf{Fig.}~\ref{fig:simulation}\textit{(b)}, each configuration was additionally evaluated across $10$ bootstrap runs, and the reported stability index is the average agreement of the estimated indirect-effect maps across those runs.

\begin{figure}[!t]
    \centering
    \includegraphics[width=\textwidth]{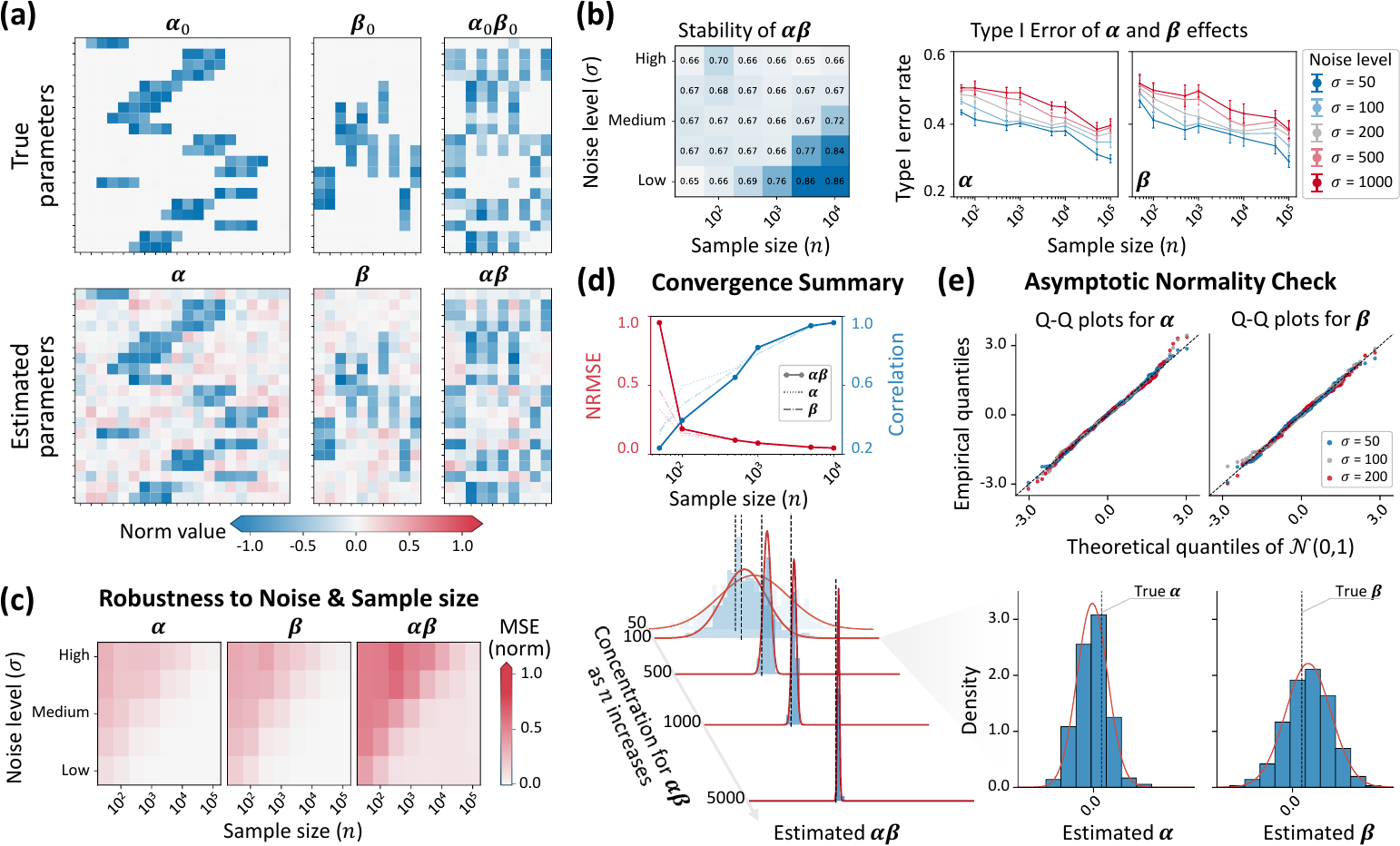}
    \caption{
        \textbf{Simulation results for the MMM mediation framework.}
        { \footnotesize
        \textit{(a)} Heatmaps comparing the ground-truth coefficient matrices
        $\bm{\alpha}_0$, $\bm{\beta}_0$, and the indirect-effect matrix $\bm{\alpha}_0\bm{\beta}_0$
        with their corresponding estimates
        $\widehat{\bm{\alpha}}$, $\widehat{\bm{\beta}}$, and $\widehat{\bm{\alpha}}\widehat{\bm{\beta}}$
        under a representative simulation setting.
        \textit{(b)} Stability of
        $\bm{\alpha}\bm{\beta}$ across combinations of sample size and noise level,
        and Type~I error rates under null mediation paths.
        \textit{(c)} Estimation error of $(\bm{\alpha},\bm{\beta},\bm{\alpha}\bm{\beta})$ as a function
        of sample size and noise level, summarizing robustness to high-noise regimes.
        \textit{(d)} Convergence patterns of Normalized Root Mean Square Error (NRMSE) and correlation with the ground truth as a function of sample size ($n$),
        with representative histograms showing concentration of the estimates as $n$ increases.
        \textit{(e)} Empirical distributions and Q--Q plots of normalized estimators
        illustrating asymptotic normality for selected entries of
        $\bm{\alpha}$ and $\bm{\beta}$.
        }
    }
    \label{fig:simulation}
\end{figure}

\subsection{Simulation experiments}

\subsubsection{Parameter recovery}
\textbf{Fig.}~\ref{fig:simulation}\textit{(a)} compares the ground-truth matrices $(\bm{\alpha}_0,\bm{\beta}_0,\bm{\alpha}_0\bm{\beta}_0)$ with their estimates.
The estimator successfully recovers the main block structure of $\bm{\alpha}_0$ and $\bm{\beta}_0$, preserving the dominant nonzero regions while keeping most null regions near zero.
Furthermore, the estimated indirect-effect matrix $\widehat{\bm{\alpha}}\widehat{\bm{\beta}}$ identifies the key mediation pathways present in $\bm{\alpha}_0\bm{\beta}_0$. Together, these highlight the ability of the two-stage regularization in MMM to uncover the underlying many-to-many-to-many mediation structure.

\subsubsection{Stability and false-positive behavior of indirect-effect matrix estimates}
\textbf{Fig.}~\ref{fig:simulation}\textit{(b)} evaluates two complementary aspects of reliability.
First, the stability index of $\widehat{\bm{\alpha}}\widehat{\bm{\beta}}$ (higher values indicate greater stability) remains uniformly moderate to high across the full grid, ranging from $0.65$ to $0.86$.
Stability improves with sample size and is highest in the large-$n$ settings, where repeated bootstrap runs produce nearly identical indirect-effect maps.
This is consistent with the visual concentration of the estimates in the large-$n$ panels and indicates that the MMM mediation analysis pipeline is reproducible under resampling.

Second, the empirical Type~I error curves for null $\bm{\alpha}$ and $\bm{\beta}$ entries decrease as $n$ grows and increase with the noise level.
For both coefficient matrices, the error rates are highest in the smallest and noisiest regimes and decline steadily toward the largest sample sizes (these rates are computed from thresholded penalized estimates not post-selection-corrected tests, and should be interpreted as a relative measure of spurious activation). Overall, larger sample sizes and lower noise yield a cleaner separation between true and null effects. 

\subsubsection{Robustness of the MMM mediation estimator to noise and sample size}

\textbf{Fig.}~\ref{fig:simulation}\textit{(c)} reports the normalized mean-squared error across the $(n,\sigma)$ grid for $\bm{\alpha}$, $\bm{\beta}$, and $\bm{\alpha}\bm{\beta}$.
The pattern is monotone in the expected direction: estimation error decreases with sample size and increases with noise magnitude.
The indirect-effect matrix $\bm{\alpha}\bm{\beta}$ is more sensitive than the individual $\bm{\alpha}$ or $\bm{\beta}$ matrices, which is natural since its estimation accumulates errors from both  the estimated $\bm{\alpha}$ and $\bm{\beta}$ components.
Even so, the heatmaps show substantial improvement once the sample size enters the mid-to-large regime, indicating that the estimator remains stable and accurate outside the most adverse settings.

\subsubsection{Convergence behavior of the estimators}
\textbf{Fig.}~\ref{fig:simulation}\textit{(d)} summarizes convergence using the normalized root mean squared error (NRMSE) and correlation as functions of $n$.
NRMSE decreases rapidly with sample size for all three targets ($\bm{\alpha}$, $\bm{\beta}$, and $\bm{\alpha}\bm{\beta}$), while the corresponding correlations with their  truth counterparts increase toward one.
The convergence is fastest for the individual coefficient matrices ($\bm{\alpha}$ and $\bm{\beta}$) and slightly slower for the indirect-effect matrix ($\bm{\alpha}\bm{\beta}$), again reflecting the two-stage nature of mediation estimation.
The overlaid histograms of selected entries of $\widehat{\bm{\alpha}}\widehat{\bm{\beta}}$ further support this pattern: as $n$ increases, the empirical distributions become more concentrated and tightly centered around the true values.
Taken together, these results are consistent with the large-sample consistency guarantees established in \textbf{Sec.} \ref{sec:theoretical_results}.

\subsubsection{Assessment of asymptotic normality}
\textbf{Fig.}~\ref{fig:simulation}\textit{(e)} evaluates the behavior of asymptotic normal approximations for representative entries of $\widehat{\bm{\alpha}}$ and $\widehat{\bm{\beta}}$.
Across the displayed noise levels, the Q--Q plots lie close to the $45^\circ$ line, with only mild deviations in the extreme tails.
The accompanying zoomed-in histograms (shown for $n=100$) are already close to bell-shaped for both $\bm{\alpha}$ and $\bm{\beta}$, indicating that the Gaussian approximation becomes reasonable at a moderate sample size.
These finite-sample diagnostics provide empirical support for the element-wise asymptotic normality established in Theorem~\ref{thm:asymptotic.normality.elastic.net.estimator:alpha.beta_k}.

Overall, \textbf{Fig.}~\ref{fig:simulation} shows that the proposed MMM mediation estimator recovers the dominant many-to-many-to-many mediation structure, remains stable under resampling, degrades gracefully with increasing noise, and exhibits the convergence and distributional behavior consistent with the theoretical results.

\section{Application to Alzheimer's Disease}
\label{sec:data_analysis}

After presenting the theoretical framework, methodology, and simulation studies, we apply the MMM mediation analysis method to genetic, brain imaging, and cognitive-behavioral data from subjects who are cognitively normal, individuals with mild cognitive impairment, and patients with Alzheimer’s disease (AD) to study many-to-many-to-many genetic–neural–cognitive mediation effects in AD.

The application of the MMM mediation analysis to AD is suitable for two reasons. First, AD is a polygenic disorder in which multiple genes are associated with disease risk~\citep{harrisonPolygenicScoresPrecision2020, bellenguez2022new}. Second, structural and functional changes in the brains of AD patients affect multiple cognitive and behavioral domains, including memory, executive function, visuospatial abilities, and language~\citep{Diaz2025OPTIMUS}. 

The application of the MMM mediation analysis to AD is also advantageous. Despite advances, previous studies of the multivariate genetic and neural underpinnings of brain diseases have been largely conducted in isolation: either focusing on the relationship between high-dimensional genetic variables and disease outcomes~\citep{kunkle2019genetic}, or between high-dimensional brain features and disease outcomes~\citep{smith2015positive}. In parallel, some studies have examined many-to-many pathways between genetic signatures and brain features~\citep{tissink2024abundant}. While these approaches are informative for pairwise relationships, they do not capture the complex interplay among genes, brain regions, and cognitive outcomes, particularly when both genetic and neural data are high-dimensional and multiple outcomes are considered.
The proposed MMM mediation analysis framework provides a platform to identify and estimate many-to-many-to-many genetic–neural–cognitive mediation effects, and it can also be used to predict multivariate outcomes based on genetic information and brain representations derived from genetic signatures. 

In this study, we apply MMM to the Alzheimer’s Disease Neuroimaging Initiative (ADNI) database (\url{adni.loni.usc.edu}). Specifically, the exposure layer $\X$ consists of 688 genome-wide significant SNPs selected from approximately 1.5 million genotyped variants via genome-wide association studies (GWAS). The mediator layer $\M$ comprises MRI-derived cortical thickness measurements from 202 brain regions. The outcome layer $\Y$ includes 11 AD-related endpoints: diagnosis (AD versus non-AD) and ten cognitive-behavior scores, including MMSE, CDRSB, FAQ, TRABSCOR, ADAS11, ADAS13, ADASQ4, RAVLT immediate, RAVLT \%forgetting, and RAVLT learning scores. The covariate vector $\Z$ contains sex, age, race, education, and marital status. \textbf{Fig.}~\ref{fig:adni}\textit{(a)} summarizes the resulting many-to-many-to-many mediation analysis design.

\begin{figure}[!ht]
    \centering
    \includegraphics[width=\textwidth]{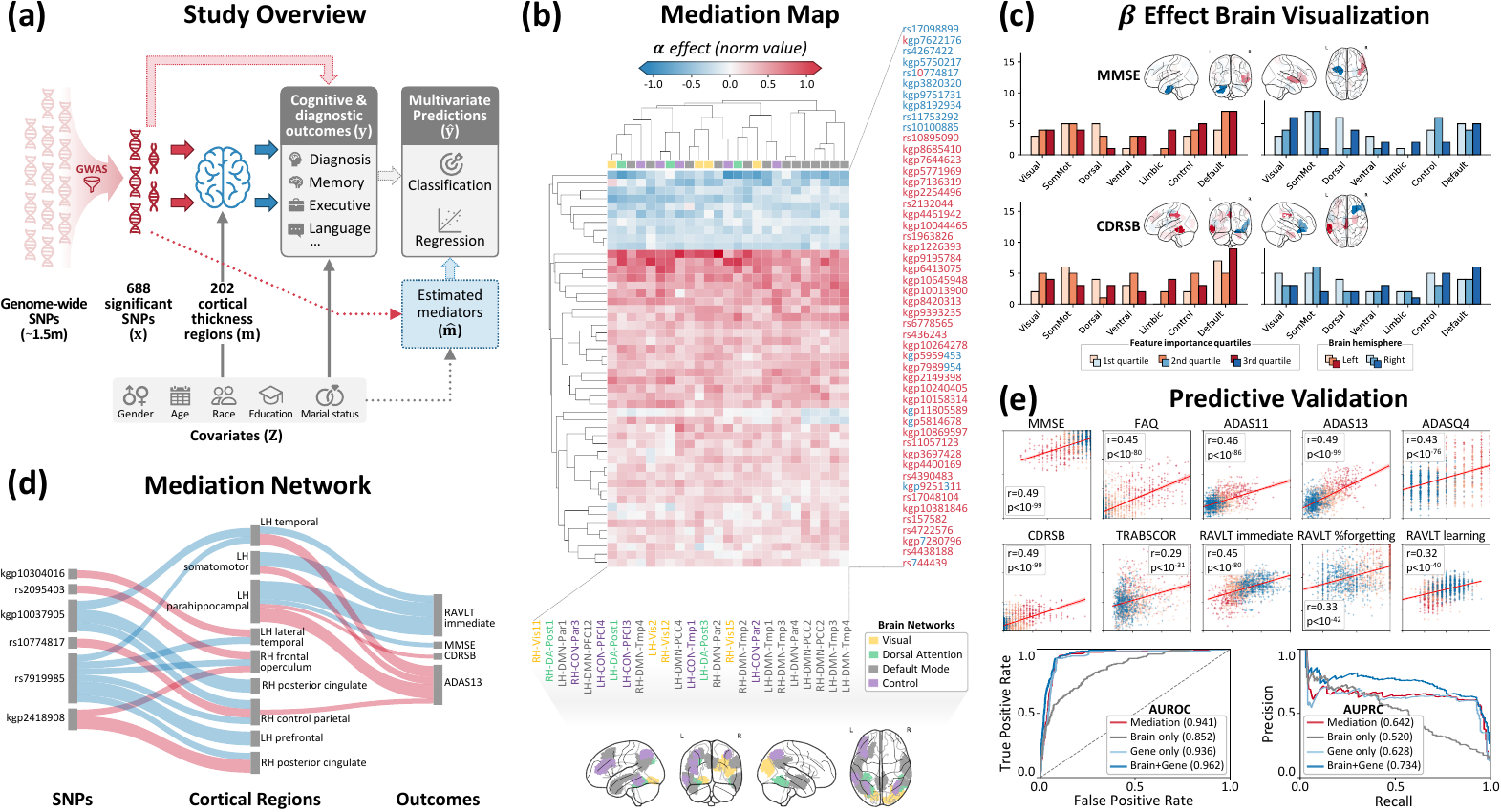}
    \caption{
        \textbf{Application of many-to-many-to-many mediation analysis to Alzheimer's disease.}
        { \footnotesize
        \textit{(a)} A schematic overview of the analysis structure, showing
        688 genome-wide significant SNPs ($\X$),
        202 cortical-thickness mediators ($\M$),
        and 11 cognitive-behavior and diagnostic outcomes ($\Y$),
        together with covariates.
        \textit{(b)} Heatmap of the estimated $\bm{\alpha}$ effects,
        highlighting structured genetic--brain map.
        \textit{(c)} Cortical surface visualization of the top mediating brain regions.
        \textit{(d)} Mediation network connecting the strongest SNPs,
        cortical mediators, and outcomes, with edge thickness proportional to mediation strength.
        \textit{(e)} Predictive performance comparison between baseline models
        and high-dimensional MMM mediation models, including scatter plots of observed versus predicted outcomes.
        }
    }
    \label{fig:adni}
\end{figure}


\subsection{Fitting the MMM model on ADNI data}
We fit the MMM mediation analysis model in \textbf{Sec.}~\ref{sec:method} to the ADNI imaging--genetics data, treating the selected SNPs as exposures, cortical-thickness measures as mediators, and the cognitive-behavior or diagnostic variables as outcomes.
The fitted model produces the exposure--to--mediator matrix $\widehat{\bm{\alpha}}$, the mediator--to--outcome matrix $\widehat{\bm{\beta}}$, and the indirect-effect matrix $\widehat{\bm{\alpha}}\widehat{\bm{\beta}}$ linking genetic variation, cortical thickness organization, and AD-related outcomes.
For interpretation, we focus on the strongest and most stable pathways in $\widehat{\bm{\alpha}}$, $\widehat{\bm{\beta}}$, and $\widehat{\bm{\alpha}}\widehat{\bm{\beta}}$; for prediction, we consider predicting both continuous outcomes and categorical AD diagnosis in previously unseen subjects.

\subsection{The \texorpdfstring{$\bm{\alpha}$}{alpha} map for studying the many--to--many genetic--brain relationship}
\textbf{Fig.}~\ref{fig:adni}\textit{(b)} shows that the estimated exposure--to--mediator effects
($\widehat{\bm{\alpha}}$) are highly structured rather than diffuse.
The strongest signals are concentrated in four large-scale cortical systems:
the Visual network, the Dorsal Attention network, the Default Mode network (DMN),
and the Control/frontoparietal network.
Representative brain regions include the right visual cortex
(\texttt{RH-Vis11}, \texttt{RH-Vis12}, \texttt{RH-Vis15}),
posterior dorsal-attention cortex
(\texttt{RH-DA-Post1}, \texttt{LH-DA-Post1}, \texttt{LH-DA-Post3}),
default-mode parietal, temporal, and posterior-cingulate/precuneus regions
(e.g., \texttt{LH-DMN-Par1}, \texttt{RH-DMN-Par2},
\texttt{LH-DMN-Tmp1/3/4}, \texttt{RH-DMN-Tmp2/3/4},
\texttt{LH-DMN-PCC2/4}, \texttt{RH-DMN-PCC2}),
and control-network parietal and lateral prefrontal regions
(e.g., \texttt{RH-CON-Par3}, \texttt{LH-CON-Par2},
\texttt{LH-CON-PFCl3/4}, \texttt{LH-CON-Tmp1}).

This spatial pattern is biologically explainable.
DMN regions such as the posterior cingulate/precuneus, lateral temporal cortex, and parietal association cortex are among the most consistently implicated systems in AD, and cortical thinning in temporal, parietal, posterior-cingulate, and prefrontal territories is a well-established feature of disease progression~\citep{migliaccio2015mapping,mcevoy2010quantitative}.
The involvement of control and dorsal-attention regions is also consistent with the broader network view of AD, in which cognitive decline is linked not only to canonical memory-related DMN disruption but also to abnormalities in attentional and frontoparietal systems~\citep{katsumi2024greater}.
More generally, the concentration of effects in specific cortical systems agrees with large-scale imaging--genetics evidence that cortical morphology is strongly polygenic and regionally patterned, rather than driven by uniform whole-brain effects~\citep{grasby2020genetic}.

The sign structure in \textbf{Fig.}~\ref{fig:adni}\textit{(b)} is also informative.
A relatively small subset of SNPs shows predominantly negative $\widehat{\bm{\alpha}}$ effects, whereas a larger group shows predominantly positive effects, suggesting that the selected genetic variants act through partially antagonistic cortical pathways.
Rather than interpreting any single SNP--ROI pair in isolation, we view this map as evidence for a coordinated many-to-many genetic influence on cortical organization.
In particular, the dominant signal lies in temporo--parietal, posterior-cingulate/precuneus, and lateral prefrontal territories, which are precisely the association cortices known to be vulnerable in AD and relevant for higher-order cognitive decline~\citep{mcevoy2010quantitative,migliaccio2015mapping}.

\subsection{The \texorpdfstring{$\bm{\beta}$}{beta} map for studying brain regions related to different cognitive-behavior outcomes}

\textbf{Fig.}~\ref{fig:adni}\textit{(c)} summarizes the estimated mediator--to--outcome effects ($\widehat{\bm{\beta}}$) for two representative endpoints, MMSE and CDRSB, and reveals a clear anatomical organization rather than diffuse whole-cortex contributions.
For both outcomes, the largest aggregate effects are concentrated in the Default Mode regions, with additional contributions from Control, Somatomotor, Dorsal Attention, and Visual systems.
This pattern is biologically plausible: temporo--parietal and posteromedial components of the Default Mode Network (DMN), including the posterior cingulate/precuneus and lateral temporal cortex, are among the cortical systems most consistently altered in AD, and these regions are tightly linked to memory loss and global cognitive decline~\citep{mohan2016significance,lee2020posterior,tang2024multimodal}.

The cortical renderings refine this picture by localizing the strongest brain--to--cognitive/behavioral outcome effects to the temporo--limbic and posteromedial association cortex, including the parahippocampal, lateral temporal, posterior cingulate, frontal opercular, and prefrontal regions.
This is again consistent with existing knowledge: medial temporal and parahippocampal circuits are closely related to episodic-memory dysfunction, whereas posterior cingulate and parietal association regions are central nodes of the AD-vulnerable network architecture~\citep{berron2020medial,lee2020posterior,tang2024multimodal}.
The prominence of control and dorsal attention systems is also reasonable, since AD-related cognitive impairment is not confined to canonical memory systems but extends to attentional and executive-control networks, especially in clinically heterogeneous or posterior-predominant phenotypes~\citep{katsumi2024greater,katsumi2023association}.

A further finding from \textbf{Fig.}~\ref{fig:adni}\textit{(c)} is the apparent left-hemispheric predominance of the strongest effects.
We interpret this cautiously, although there is prior evidence that hemispheric asymmetry is altered in AD and that left temporal and left parietal association cortex can be especially informative in AD-related structural and functional phenotypes~\citep{mizrak2024investigation,park2017improved}.

Taken together, these results suggest that the links between the brain and cognitive/behavioral outcomes are anatomically informative: the outcomes are associated primarily with a distributed set of AD-relevant association cortices, with the DMN and adjacent temporo--parietal systems playing the leading role.

\subsection{The many--to--many--to--many genetic--brain--cognitive mediation effect \texorpdfstring{$\bm{\alpha}\bm{\beta}$}{alpha beta}}

\textbf{Fig.}~\ref{fig:adni}\textit{(d)} makes the many--to--many--to--many structure explicit by extracting the strongest SNP--brain region--cognitive/behavioral outcome pathways from the fitted model.
Rather than a diffuse graph, the estimated network has a hub-like organization: a small subset of SNPs converges onto a refined set of cortical mediators, and these cortical mediators then project to a small number of clinically central outcomes.
The key cortical hubs are concentrated in left temporal, left parahippocampal, left lateral temporal, left prefrontal, right posterior cingulate, right frontal operculum, and right control-parietal territories.
This organization is notable because it bridges regions classically linked to episodic memory and AD-vulnerable association cortex with outcomes spanning memory performance, global cognition, and clinical severity.

The strongest memory-related branch of the network runs through temporal and parahippocampal mediators toward RAVLT immediate.
This pattern is consistent with the central role of medial temporal circuitry in memory encoding and retrieval and with evidence that disruption of temporal and parahippocampal systems is among the earliest and most consequential features of AD~\citep{berron2020medial,tang2024multimodal}.
The branches leading to MMSE, CDRSB, and ADAS13 are weighted by posteromedial, parietal-control, and prefrontal mediators, which is biologically reasonable because these outcomes reflect broader global impairment rather than memory in isolation~\citep{lee2020posterior,katsumi2023association}.
The repeated appearance of posterior cingulate and control-parietal regions is especially important, as these areas sit near the interface of Default Mode and higher-order control systems that are repeatedly implicated in AD progression and cognitive deterioration~\citep{mohan2016significance,lee2020posterior,katsumi2024greater}.

The network also contains both positive and negative paths, indicating that the selected genetic variants do not act through a single monotone cortical mechanism.
In particular, the MMM estimates reveal partially opposing cortical effects: some variants are associated with thicker or likely relatively preserved cortical mediators, whereas others are associated with thinner or potentially more adverse cortical patterns.
We, therefore, view \textbf{Fig.}~\ref{fig:adni}\textit{(d)} not as a collection of isolated SNP--ROI--disease links, but as evidence that polygenic variation is funneled through a compact set of biologically meaningful cortical hubs that jointly shape multiple AD-related outcomes.
This is the type of structure that is lost in univariate-exposure/univariate-outcome analyses but retained by the MMM formulation.

\subsection{Multivariate outcome prediction in previously unseen subjects} \label{subsection:prediction}
We evaluate whether the MMM mediation framework yields useful downstream predictors when only genetic exposures and the predicted brain data (brain representation predicted by genetic data) are used during prediction (see \textbf{Fig.}~\ref{fig:adni}\textit{(e)}).

For the continuous outcomes, the MMM-based predictions show clear and strong quantitative agreement with the observed values, with the best performance obtained for global cognition and disease-severity measures such as MMSE, CDRSB, and ADAS13.
This indicates that the estimated mediation effects capture meaningful cross-domain variation in AD-related phenotypes rather than signal restricted to a single cognitive score.

For diagnosis classification, the MMM mediation model outperforms both the brain-only and gene-only baselines, although its performance does not exceed that of the model using both observed brain and genetic features.
This ordering is expected and important for interpretation.
The MMM predictor does not use MRI measurements at test time; instead, it relies solely on genetic inputs and the brain patterns learned through the MMM mediation model.
Its strong performance demonstrates that a substantial portion of the predictive information carried by brain structure can be transferred into a genetically mediated representation, even though some information is inevitably lost relative to a model that directly includes both modalities for prediction.

Taken together, the out-of-sample prediction studies suggest two complementary points.
First, the mediation fit is not merely descriptive: it retains practical predictive value across multiple AD-related outcomes.
Second, the genetically mediated brain representation provides a principled compromise between interpretability and deployability, since it incorporates information learned from brain imaging data during training but does not require imaging data for prediction in new subjects.

\subsection{Interpretation}
The application of the MMM mediation analysis model on ADNI data indicates that polygenic variation influences AD-related phenotypes through a structured set of cortical systems rather than via diffuse, nonspecific areas.
The dominant mediating regions are located in the Default Mode, Control, Dorsal Attention, and Visual networks, whereas the strongest mediator--to--outcome effects are concentrated in temporo--parietal, posteromedial, and prefrontal association cortex.
This overall pattern is consistent with current perspectives on AD as both a polygenic disorder and a network-level brain disease, in which distributed cortical vulnerability gives rise to impairments across multiple cognitive domains~\citep{grasby2020genetic,mohan2016significance,tang2024multimodal, Diaz2025OPTIMUS}.


\section{Conclusion}\label{sec-conc}
In this work, we study many--to--many--to--many (MMM) mediation analysis where exposure $\mathbf{x}$, mediator $\mathbf{m}$ and outcome $\mathbf{y}$ are all multivariate, and both exposures and mediators may be high-dimensional. We first establish the theoretical properties of the MMM mediation estimators, including the consistency, asymptotic normality, and error bounds. We then demonstrate their empirical performance through simulation experiments. Finally, we apply the MMM mediation analysis framework to ADNI data to investigate the genetic-neural-cognitive/behavior mediation effects in Alzheimer's disease. 

\medskip

There are a few potential extensions of this paper. First, one can generalize the MMM mediation analysis framework to accommodate non-linear relationships among exposures, mediators, and outcomes. One way to do so is to consider basis expansions or nonparametric function approximation of maps between exposures to mediators and mediators to outcomes.
Second, in the present work, we assume that the mediators (e.g., brain regions) are independent. A natural extension is to account for the correlation structure among mediator, for example, by combining the MMM mediation analysis framework with the generalized estimation equations.
Third, it may be useful to explore a Bayesian MMM mediation analysis formulation by incorporating prior distributions on the exposures (e.g., genes) and mediators (e.g., brain regions). Such an approach could leverage existing biological knowledge to more precisely identify relevant genes and brain regions, improving interpretability, and potentially increasing estimation accuracy in high-dimensional settings.
Finally, one can consider a longitudinal version of the MMM method to study how mediators evolve over time in transmitting the effects of exposures to outcomes. A beginning can perhaps be made by linking the current MMM mediation method with mixed-effects models, GEE, and partial differential equations (PDEs) (extending $\M$ to $\M(t,s)$, where $t$ denotes time and $s$ indexes spatial location).


\section{Data Availability Statement}\label{data-availability-statement}
Data used in preparation of this article were obtained from the Alzheimer's Disease Neuroimaging Initiative (ADNI) database (\url{adni.loni.usc.edu}). The investigators within ADNI contributed to the design and implementation of ADNI and/or provided data but did not participate in analysis or writing of this paper. A complete listing of ADNI investigators can be found at: \url{http://adni.loni.usc.edu/wp-content/uploads/how_to_apply/ADNI_Acknowledgement_List.pdf}.


\section*{Acknowledgement}

This project is partly funded by the Swiss National Science Foundation (SNSF) grants 3200-0-239967 and CR00I5-235987.

Data collection and sharing for the Alzheimer's Disease Neuroimaging Initiative (ADNI) is funded by the National Institute on Aging (National Institutes of Health Grant U19AG024904). The grantee organization is the Northern California Institute for Research and Education. In the past, ADNI has also received funding from the National Institute of Biomedical Imaging and Bioengineering, the Canadian Institutes of Health Research, and private sector contributions through the Foundation for the National Institutes of Health (FNIH) including generous contributions from the following: AbbVie, Alzheimer's Association; Alzheimer's Drug Discovery Foundation; Araclon Biotech; BioClinica, Inc.; Biogen; Bristol-Myers Squibb Company; CereSpir, Inc.; Cogstate; Eisai Inc.; Elan Pharmaceuticals, Inc.; Eli Lilly and Company; EuroImmun; F. Hoffmann-La Roche Ltd and its affiliated company Genentech, Inc.; Fujirebio; GE Healthcare; IXICO Ltd.; Janssen Alzheimer Immunotherapy Research \& Development, LLC.; Johnson \& Johnson Pharmaceutical Research \& Development LLC.; Lumosity; Lundbeck; Merck \& Co., Inc.; Meso Scale Diagnostics, LLC.; NeuroRx Research; Neurotrack Technologies; Novartis Pharmaceuticals Corporation; Pfizer Inc.; Piramal Imaging; Servier; Takeda Pharmaceutical Company; and Transition Therapeutics.

\section*{Author Contribution}

OYC and DCC conceived and designed the study. TDN provided the theoretical and methodological results. TKT checked the theory and methods, and conducted the simulation studies and data analysis. CKT performed the GWAS analyses. DCC wrote the discussion section. TBN provided guidance to TKT. TDN, TKT, DCC, and OYC wrote the paper, with comments from all authors.  

\bibliography{bibliography.bib}
\clearpage
\appendix

\begin{center}
\textbf{\Large Supplementary Materials to}\\[3mm]
\textbf{\Large High-dimensional Many--to--many--to--many}\\
\textbf{\Large Mediation Analysis}
\end{center}

This document contains the Supplementary Materials to the paper ``High-dimensional Many--to--many--to--many Mediation Analysis''.
Appendix~\ref{sup:consistency} provides the proofs of all theorems and lemmas related to the consistency of the proposed method.
Appendix~\ref{sup:asymptotic_normality} presents the proofs of the results concerning the asymptotic theory. 
Appendix~\ref{sup:alpha_beta} contains the proof of the mediation effect of $\bm{\alpha}\bm{\beta}$. 
Appendix~\ref{sec:proof:causal.mediation} includes proof for causal mediation analysis in the many--to--many--to--many setting.

\section{Proofs}
\label{app. asymptotic theory}
\renewcommand{\theequation}{A\arabic{equation}}
\setcounter{equation}{0}
\subsection{Proof for the consistency}
\label{sup:consistency}
\subsubsection{Proof of Lemma~\ref{lem:KKT-condition}}
\label{section:proof:lem:KKT-condition}
\begin{proof}[Proof of Lemma~\ref{lem:KKT-condition}]
We follow the proof of Lemma 1 in~\cite{Jia-Yu:2010:model.consistency.elastic-net}. 
For $1 \leq k \leq T$,  taking the first derivative w.r.t. $\bm{\beta}_{k}$ for the loss function in~\eqref{equa:elastic.net:loss.function:beta_k}, we get   
	\begin{align*}
		\dfrac{\partial}{\partial \bm{\beta}_{k}} L_{1}(\bm{\beta})   
		= -  &2 \sum_{i=1}^{n} \Big[  (Y_{i})_{k} \M_{i} -   \M_{i}  \M_{i}^{\top}  \bm{\beta}_{k} -  \M_{i} \big( \X_{i}^{\top} \bm{\gamma}_{k} \big) - \M_{i} \big(  \Z_{i}^{\top} \bm{\eta}_{k} \big) \Big] 
		+ 2 \lambda_{Y,2}   \bm{\beta}_{k} + \lambda_{Y,1}    
		\,  \Upsilon ,  
	\end{align*}
	where 
	\begin{align*}
		\Upsilon =  \left\{  \begin{array}{ll}
			\textrm{sign}(\beta_{jk})	&  \textrm{ if }  \beta_{jk}  \neq  0
			\\[0.2cm]   
			\textrm{any real number which } \in [-1 ; 1] 	&  \textrm{ if }  \beta_{jk}  =  0
		\end{array}
		\right. , \quad  \textrm{  with  }  1  \leq  j  \leq  p .   
	\end{align*} 
	Thus, by KKT conditions for optimality in a convex program, the point $\widehat{\bm{\beta}}_{k}$ is optimal if and only if:
	\begin{align}
		& -  \sum_{i=1}^{n} \Big[ 2 (\Y_{i})_{k} \M_{i}  -   2 \M_{i} \M_{i}^{\top}  \widehat{\bm{\beta}}_{k}  - 2 \M_{i} \big( \X_{i}^{\top} \bm{\gamma}_{k} \big) - 2 \M_{i} \big(   \Z_{i}^{\top}\bm{\eta}_{k} \big) \Big] 
		+ 2 \lambda_{Y,2}   \widehat{\bm{\beta}}_{k} + \lambda_{Y,1}   \widehat{ \Upsilon }   
		=  0,    
	\end{align} 
	where 
	\begin{align*}
		\widehat{ \Upsilon }   =  \left\{  \begin{array}{ll}
			\textrm{sign}(\widehat{\beta}_{jk})	&  \textrm{ if }  \widehat{\beta}_{jk}  \neq  0
			\\[0.2cm]   
			\textrm{any real number which } \in [-1 ; 1] 	&  \textrm{ if }  \widehat{\beta}_{jk}  =  0
		\end{array}
		\right. , \quad  \textrm{  with  }  1  \leq  j  \leq  p .   
	\end{align*} 
	Then, substituting $( \Y_{i} )_{k}$ by $ \Big(  \M_{i}^{\top}\bm{\beta}_{k}  +    \X_{i}^{\top} \bm{\gamma}_{k}       +    \Z_{i}^{\top}\bm{\eta}_{k}   +   ( \bm{\xi}_{i} )_{k}  \Big) $ yields: 
	\begin{align}
		& -  \sum_{i=1}^{n} \Big[ 2    \M_{i} \M_{i}^{\top} \big( \bm{\beta}_{k}  - \widehat{\bm{\beta}}_{k} \big) + 2 ( \bm{\xi}_{i} )_{k} \M_{i}   \Big]   
		+   2 \lambda_{Y,2}   \widehat{\bm{\beta}}_{k} + \lambda_{Y,1}  \widehat{ \Upsilon }    
		=  0.   
		\label{proof:lem1:KKT:beta:equa1}
	\end{align}
	Remind that, we assume the first $d$ elements of $\bm{\beta}_{k}$ are non-zeroes, so we  consider $\bm{\beta}_{k,(1)} = \big( \beta_{1k} , ..., \beta_{dk} \big)$. 
    The condition $\mathcal{R}(\M, \bm{\beta}, \bm{\xi}, \lambda_{Y,1}, \lambda_{Y,2})$ holds if and only if we have $\widehat{\bm{\beta}}_{k,(2)} = 0$, $\textrm{sign}(\widehat{\bm{\beta}}_{k,(1)}) = \textrm{sign}(\bm{\beta}_{k,(1)})$ and $
	\big\| \widehat{ \Upsilon}_{(2)} \big\|_{\infty} \leq 1$.   
	From these conditions and using~\eqref{proof:lem1:KKT:beta:equa1}, we conclude that the condition $\mathcal{R}(\M, \bm{\beta}, \bm{\xi}, \lambda_{Y,1}, \lambda_{Y,2})$ holds if and only if 
	\begin{align*}
		-  \sum_{i=1}^{n} \Big[ 2 \M_{i,(1)}  \M_{i,(1)}^{\top} \big( \bm{\beta}_{k,(1)}  - \widehat{\bm{\beta}}_{k,(1)} \big)     + 2 \big( \bm{\xi}_{i} \big)_{k}   \M_{i,(1)}    \Big]   
		+   2 \lambda_{Y,2}   \widehat{\bm{\beta}}_{k,(1)} + \lambda_{Y,1}  
		\big[ \textrm{sign} \big( \bm{\beta}_{j,k} \big)  \big]_{1 \leq j \leq  d}  =  0  ,    
	\end{align*}
	since  $\textrm{sign}(\widehat{\bm{\beta}}_{k,(1)}) = \textrm{sign}(\bm{\beta}_{k,(1)})$, and 
	\begin{align*}
	 \sum_{i=1}^{n} \Big[ 2 \M_{i,(2)} \M_{i,(1)}^{\top}  \big( \bm{\beta}_{k,(1)}  - \widehat{\bm{\beta}}_{k,(1)} \big)    +  2   \big( \bm{\xi}_{i}   \big)_{k}    \M_{i,(2)}    \Big]   
		=  \lambda_{Y,1}  \widehat{ \Upsilon}_{(2)}.   
	\end{align*}
	With $ \textrm{sign} \big( \widehat{\bm{\beta}}_{k,(1)} \big) =  \big[ \textrm{sign} \big( \widehat{\bm{\beta}}_{j,k} \big)  \big]_{1 \leq j \leq  d}$, solve for  $\widehat{\bm{\beta}}_{k,(1)}$ to conclude that:  
	\begin{align*}  
		\widehat{\bm{\beta}}_{k,(1)} = \Big(  \sum_{\ell = 1}^{n} \M_{\ell,(1)}   \M_{\ell , (1)}^{\top}  + \lambda_{Y,2} \mathbf{I}      \Big)^{-1} 
		\Big[ &  \sum_{i=1}^{n} \Big( \M_{i, (1)} \M_{i, (1)}^{\top}  \bm{\beta}_{k , (1)}  +  \M_{i,(1)}   \big( \bm{\xi}_{i} \big)_{k} \Big) 
		-   \dfrac{\lambda_{Y,1}}{2} \textrm{sign} \big( \bm{\beta}_{k,(1)} \big)  \Big]     
	\end{align*}
		\textrm{ and }
	\begin{align*}
		-\lambda_{Y,1} \widehat{ \Upsilon }_{(2)} =  2  \sum_{i=1}^{n}  \Big[   \M_{i,(2)} \M_{i,(1)}^{\top} \Big( \sum_{\ell = 1}^{n}  \M_{\ell , (1)} \M_{\ell , (1)}^{\top}  +  \lambda_{Y,2} \mathbf{I} \Big)^{-1}   
		\Big( & \M_{i,(1)} \big( \bm{\xi}_{i} \big)_{k} -  \dfrac{\lambda_{Y,1}}{2} \textrm{sign} \big( \bm{\beta}_{k,(1)} \big)  
		\\
		& -  \lambda_{Y,2} \bm{\beta}_{k,(1)} \Big) 
		 +  \M_{i,(2)}   \big( \bm{\xi}_{i} \big)_{k}  \Big]  . 
	\end{align*}
	The conditions $\textrm{sign}(\widehat{\bm{\beta}}_{(1)}) = \textrm{sign}(\bm{\beta}_{(1)})$ and $\big\| \widehat{ \Upsilon}_{(2)} \big\|_{\infty} \leq 1$ are exactly~\eqref{lem:KKT-condition:condition1} and~\eqref{lem:KKT-condition:condition2}. This concludes the proof of Lemma~\ref{lem:KKT-condition}.   
\end{proof}
\subsubsection{Proof of Theorem~\ref{thm:consistency.elastic.net.estimator:beta_k}}
\label{section:proof:thm:consistency.elastic.net.estimator:beta_k}
Before proving Theorem~\ref{thm:consistency.elastic.net.estimator:beta_k}, we introduce without proof a well-known comparison result on Gaussian maxima (see~\cite{Ledoux-Talagrand:1991}).  
\begin{lem}
	\label{lem:concentration.inequality:Ledoux.Talagrand.1991}
	For any Gaussian random vector $(W_{1}, ..., W_{n})$, we have:
	\begin{align}
		\mathbb{E} \Big( \max_{1 \leq i \leq n} W_{i} \Big)  \leq  3 \sqrt{\log(n)} \max_{1 \leq i \leq n} \sqrt{ \mathbb{E} \big( W_{i}^{n} \big) } .  
	\end{align}
\end{lem}
With this lemma, we have when $n > 1$, 
\begin{align}
	\mathbb{E} \Big( \max_{1 \leq i \leq n} \big| W_{i} \big| \Big)   &\leq    \mathbb{E} \big( |W_{1}| \big) + 2 \mathbb{E} \Big( \max_{1 \leq i \leq n} W_{i} \Big)   
	\nonumber   
	\\
	&\leq   \sqrt{ \mathbb{E} \big( W_{1}^{2} \big) }  +  6 \sqrt{\log(n)} \max_{1 \leq i \leq n} \sqrt{ \mathbb{E} \big( W_{i}^{2} \big) }  
	\nonumber   
	\\
	&\leq  8 \sqrt{\log(n)} \max_{1 \leq i \leq n}  \sqrt{ \mathbb{E} \big( W_{i}^{2} \big) } ,
	\label{lem:concentration.inequality:Ledoux.Talagrand.1991:inequa2}  
\end{align} 
where the first inequality comes from~\cite{Ledoux-Talagrand:1991}, the second from Jensen's inequality along with Lemma~\ref{lem:concentration.inequality:Ledoux.Talagrand.1991}, and the third from the fact that $2 \log(n) >  1$ when $n > 1$.   

\medskip 

We now start to prove Theorem~\ref{thm:consistency.elastic.net.estimator:beta_k}. 
\begin{proof}[Proof of Theorem~\ref{thm:consistency.elastic.net.estimator:beta_k}]
We adapt the proof of Theorem~1 in~\cite{Jia-Yu:2010:model.consistency.elastic-net}.
	\subsubsection*{Analysis of $\mathcal{M}(V)$}
	We have $V_{m}$ is Gaussian random variable with mean:
	\begin{align*}
		\mu_{V_m} =  \mathbb{E}(V_{m} )  =  \Big(  \sum_{i = 1}^{n} (\M_{i})_{m}  \M_{i,(1)}^{\top}    \Big)  
		\Big( \sum_{\ell = 1}^{n}  \M_{\ell,(1)} \M_{\ell,(1)}^{\top}  + \lambda_{Y,2} \mathbf{I} \Big)^{-1} \Big( \lambda_{Y,1} \overrightarrow{\mathbf{b}}  +  2 \lambda_{Y,2} \bm{\beta}_{k , (1)}   \Big) .
	\end{align*}
	The EIC implies that:
	\begin{align*}  
		\Big| \dfrac{1}{n}  \sum_{i=1}^{n}  \M_{i,(2)} \M_{i,(1)}^{\top}   \Big( \frac{1}{n} \sum_{\ell = 1}^{n}  \M_{\ell,(1)}    \M_{\ell,(1)}^{\top}  +  \dfrac{\lambda_{Y,2}}{n} \mathbf{I}  \Big)^{-1}  \Big[ \textrm{sign} \big( \bm{\beta}_{k , (1)} \big)  +  \dfrac{2 \lambda_{Y,2}}{ \lambda_{Y,1} } \bm{\beta}_{k , (1)}  \big) \Big]    
		\Big|  \leq  1 - \Psi ,     
	\end{align*}
	thus, $\big| \mu_{V_m} \big|  \leq  (1 - \Psi) \lambda_{Y,1}$. 
	Let $\widetilde{V}_{m} :=  2 \sum_{i=1}^{n} \big( \M_{i} \big)_{m} \Big[ 1 - \M_{i,(1)}^{\top} \Big(  \sum_{\ell = 1}^{n}   \M_{\ell,(1)} \M_{\ell,(1)}^{\top}  +  \lambda_{Y,2} \mathbf I \Big)^{-1} \M_{i,(1)}  \Big]  \big( \bm{\xi}_{i}  \big)_{k} $,  
	so $V_{m} = \mu_{V_m} + \widetilde{V}_{m}$. Note that  $\mathcal{M}(V)$ holds if and only if  $\displaystyle{ \max_{m \in D^{c}} } \dfrac{V_{m}}{\lambda_{Y,1}}  \leq  1 $ and  $\displaystyle{ \min_{m \in  D^{c} } }  \dfrac{V_{m}}{\lambda_{Y,1}}  \geq  -1 $. Since:    
	\begin{align*}  
		\dfrac{\max_{m \in D^{c}} V_{m} }{\lambda_{Y,1}}  &=    \dfrac{ \max_{m \in D^{c}} \Big( \mu_{V_m} + \widetilde{V}_{m} \Big) }{\lambda_{Y,1}}  \leq  (1-\Psi) + \dfrac{1}{\lambda_{Y,1}}  \max_{1 \leq m \leq p} \widetilde{V}_{m}  
		\\
		\textrm{ and }   \hspace*{3cm}   &   
		\\
		\dfrac{\min_{m \in D^{c}} V_{m} }{\lambda_{Y,1}}  &=   \dfrac{ \min_{m \in D^{c}} \Big( \mu_{V_m} + \widetilde{V}_{m} \Big) }{\lambda_{Y,1}}  \geq  - (1-\Psi) + \dfrac{1}{\lambda_{Y,1}}  \min_{1 \leq m \leq p} \widetilde{V}_{m};     
	\end{align*}    
	we need to show that:  
	\begin{align*}
		\lim_{n   \rightarrow   +\infty}  \mathbb{P} \Bigg( \Big( \dfrac{1}{\lambda_{Y,1}} \max_{m \in D^{c}} \widetilde{V}_{m}  > \Psi  \Big)  \cup   \Big( \dfrac{1}{\lambda_{Y,1}} \min_{m \in D^{c}} \widetilde{V}_{m}  <  - \Psi  \Big) \Bigg)  =  0  .  
	\end{align*}
	It is, however, sufficient to show that  $\displaystyle{  \lim_{n   \rightarrow   +\infty}  \mathbb{P} \Big( \max_{m \in D^{c}}   \dfrac{ \big|  \widetilde{V}_{m}   \big|  }{ \lambda_{Y,1} }   > \Psi  \Big)  =  0}$. By applying Markov's inequality and~\eqref{lem:concentration.inequality:Ledoux.Talagrand.1991:inequa2}, we obtain: 
	\begin{align*}
		\mathbb{P} \Bigg(    \dfrac{ \displaystyle{  \max_{m \in D^{c}} }   \big|   \widetilde{V}_{m}   \big| }{ \lambda_{Y,1} }  > \Psi  \Bigg)  \leq  \dfrac{ \mathbb{E} \Big( \displaystyle{  \max_{m \in D^{c}}  }  |\widetilde{V}_{m}| \Big) }{ \lambda_{Y,1} \, \Psi }   \leq   \dfrac{ 8 \sqrt{ \log(p - d) } }{ \lambda_{Y,1} \,  \Psi }  \max_{m \in D^{c}} \sqrt{ \mathbb{E} \big( \widetilde{V}_{m}^{2}   \big)  } .  
	\end{align*}   
	Now, by straightforward computation, since $\E(\widetilde{V}_{m}) = 0$, one has $\mathbb{E} \big(  \widetilde{V}_{m}^{2}  \big) =  \textrm{Var} \big(   \widetilde{V}_{m} \big)$, so: 
	\begin{align*}  
		\mathbb{E} \big(  \widetilde{V}_{m}^{2}  \big)  
		&=   4   \sum_{i=1}^{n}        \Big( \big( \M_{i} \big)_{m} \Big)^{2}    \Big[  1 -  \M_{i,(1)}^{\top}  \Big(  \sum_{\ell = 1}^{n}    \M_{\ell,(1)}    \M_{\ell,(1)}^{\top}  + \lambda_{Y,2} \mathbf{I}  \Big)^{-1} \M_{i,(1)}  \Big]^{2}    \textrm{Var} \Big(   \big( \bm{\xi}_{i} \big)_{k}   \Big) 
		\\ 
		&\leq  \dfrac{4}{n^{2}} \sum_{i=1}^{n}       \Big[ 2  +  2 \Lambda_{max}  \Big(  \Big( \dfrac{1}{n} \sum_{\ell = 1}^{n}    \M_{\ell,(1)}    \M_{\ell,(1)}^{\top}  + \dfrac{ \lambda_{Y,2} }{n} \mathbf{I}  \Big)^{-1} \Big)^{2} \left\| \M_{i,(1)} \right\|_{2}^{2} \Big]
		\\ 
		&\leq  \dfrac{8}{n}    \Big[ 1  +   \Lambda_{max}  \Big(  \Big( \dfrac{1}{n}   \sum_{\ell = 1}^{n}    \M_{\ell,(1)}    \M_{\ell,(1)}^{\top}  + \dfrac{ \lambda_{Y,2} }{n} \mathbf{I}  \Big)^{-1} \Big)^{2} \,  d \, \left\|  \M \right\|_{\infty}^{2} \Big], 
	\end{align*}
	where $\Lambda_{max}(\mathbf{A})$ denotes the largest eigenvalue of matrix $\mathbf{A}$ and $\left\|  \M \right\|_{\infty} = \underset{1 \leq i \leq n, 1 \leq k \leq p}{\max}  | (\M_{i})_{k} |$. 
	Then, we get:
	\begin{align*}
		& \mathbb{P} \Bigg(    \dfrac{ \displaystyle{  \max_{m \in D^{c}} }   \big|   \widetilde{V}_{m}   \big| }{ \lambda_{Y,1} }  > \Psi  \Bigg)  
		\\  
		&\leq   \dfrac{ 8 \sqrt{ \log(p - d) } }{ \lambda_{Y,1} \,  \Psi }  \max_{m} \sqrt{ \mathbb{E} \big( \widetilde{V}_{m}^{2}   \big)  } 
		\\
		&\leq  \dfrac{ 32  \sqrt{ \log(p - d) }  \,  \sqrt{ 1  +   \Lambda_{max}  \Big(   \Big( \frac{1}{n}  \sum_{\ell = 1}^{n}    \M_{\ell,(1)}    \M_{\ell,(1)}^{\top}  + \frac{ \lambda_{Y,2} }{n} \mathbf{I}  \Big)^{-1} \Big)^{2} \,  d \, \left\|  \M \right\|_{\infty}^{2} }  }{ \lambda_{Y,1} \,  \Psi \,  \sqrt{n}  }  \overset{n  \rightarrow  +\infty}{\longrightarrow} 0 ,   
	\end{align*}
	under the condition (a) of Theorem~\ref{thm:consistency.elastic.net.estimator:beta_k}. 
	\subsubsection*{Analysis of $\mathcal{M}(U)$}
	Let $W_{j} :=  \sum_{i = 1}^{n} e_{j}^{\top} \big( \sum_{\ell = 1}^{n} \M_{\ell , (1)} \M_{\ell , (1)}^{\top}  + \lambda_{Y,2} \mathbf{I} \big)^{-1}  \M_{i,(1)} (\bm{\xi
    }_{i})_{k} $ , so that: 
	\begin{align*}
		\max_{j \in  D} |U_{j}|  &=  \max_{j \in  D} \Big| W_{j} - \dfrac{1}{2}     e_{j}^{\top} \Big(  \sum_{\ell = 1}^{n} \M_{\ell , (1)} \M_{\ell , (1)}^{\top}  +   \lambda_{Y,2} \mathbf{I} \Big)^{-1}    \lambda_{Y,1} \overrightarrow{\mathbf{b}} 
		\Big|
		\\
		&\leq  
		\max_{j \in  D} |W_{j}|  +  \dfrac{1}{2} \lambda_{Y,1} \left\|   \Big(  \sum_{\ell = 1}^{n} \M_{\ell , (1)}  \M_{\ell , (1)}^{\top}  +   \lambda_{Y,2} \mathbf{I} \Big)^{-1}    \overrightarrow{\mathbf{b}}   \right\|_{\infty}.
	\end{align*}
    We have: 
	\begin{align*}
		& \Var(W_{j})  
		\\  
		&=     \sum_{i = 1}^{n} \mathbf{e}_{j}^{\top} \Big(  \sum_{\ell = 1}^{n} \M_{\ell , (1)} \M_{\ell , (1)}^{\top}  +  \lambda_{Y,2} \mathbf{I} \Big)^{-1}  \Big(  \M_{i,(1)} \M_{i,(1)}^{\top}   \Big)  \Big(  \sum_{\ell = 1}^{n} \M_{\ell , (1)} \M_{\ell , (1)}^{\top}   +   \lambda_{Y,2} \mathbf{I} \Big)^{-1}  \mathbf{e}_{j}
		\\
		&=    \Big\{     \mathbf{e}_{j}^{\top} \Big(  \sum_{\ell = 1}^{n} \M_{\ell , (1)} \M_{\ell , (1)}^{\top}  +  \lambda_{Y,2} \mathbf{I} \Big)^{-1}  \Big(    \sum_{i = 1}^{n}  \M_{i,(1)} \M_{i,(1)}^{\top}  +  \lambda_{Y,2} \mathbf{I}  \Big)  
		\\
		& \hspace{10cm}  \Big( \sum_{\ell = 1}^{n} \M_{\ell , (1)} \M_{\ell , (1)}^{\top}   +   \lambda_{Y,2} \mathbf{I} \Big)^{-1}  \mathbf{e}_{j}  \Big\}
		\\
		& \hspace{1cm} -        \mathbf{e}_{j}^{\top} \Big(  \sum_{\ell = 1}^{n} \M_{\ell , (1)} \M_{\ell , (1)}^{\top}  +  \lambda_{Y,2} \mathbf{I} \Big)^{-1}  \Big( \lambda_{Y,2} \mathbf{I}  \Big)  \Big(  \sum_{\ell = 1}^{n} \M_{\ell , (1)} \M_{\ell , (1)}^{\top}   +   \lambda_{Y,2} \mathbf{I} \Big)^{-1}  \mathbf{e}_{j} 
		\\
		&\leq  \dfrac{1}{n}   \mathbf{e}_{j}^{\top} \Big( \dfrac{1}{n} \sum_{\ell = 1}^{n} \M_{\ell , (1)} \M_{\ell , (1)}^{\top}  + \dfrac{ \lambda_{Y,2} }{n} \mathbf{I} \Big)^{-1}  \mathbf{e}_{j}
		\\
		&\leq   \dfrac{1}{n C_{\min}} ; 
	\end{align*}
	Then, using~\eqref{lem:concentration.inequality:Ledoux.Talagrand.1991:inequa2} we obtain: 
	\begin{align*}
		\E \Big[ \max_{j \in  D} \big| W_{j} \big| \Big]  \leq  8 \sqrt{ \dfrac{ \log(d) }{n \,  C_{\min}} } . 
	\end{align*}
	This implies that: 
	\begin{align*}  
		\mathbb{P} \Big[ \max_{j \in  D} \big| U_{j} \big| \geq  \rho \Big]
		&\leq  \mathbb{P} \Bigg[  \dfrac{1}{\rho} \Bigg\{  \max_{j \in  D} \big| W_{j} \big| + \dfrac{\lambda_{Y,1}}{2}  \left\|  \big(  \sum_{\ell = 1}^{n} \M_{\ell , (1)} \M_{\ell , (1)}^{\top} +  \lambda_{Y,2} \mathbf{I} \big)^{-1}    \overrightarrow{\mathbf{b}}   \right\|_{\infty}
		\Bigg\}  \geq  1  \Bigg]
		\\
		&\leq   \dfrac{1}{\rho}  \Bigg\{ \E \Big[  \max_{j \in  D} \big| W_{j} \big| \Big]  + \dfrac{\lambda_{Y,1}}{2}  \left\|  \big(  \sum_{\ell = 1}^{n} \M_{\ell , (1)} \M_{\ell , (1)}^{\top} +  \lambda_{Y,2} \mathbf{I} \big)^{-1}    \overrightarrow{\mathbf{b}}   \right\|_{\infty}
		\Bigg\}  
		\\
		&\leq   \dfrac{1}{\rho}  \Bigg\{ 8 \sqrt{ \dfrac{ \log(d)   }{ n C_{\min} } }  + \dfrac{\lambda_{Y,1}}{2}  \left\|  \big(  \sum_{\ell = 1}^{n} \M_{\ell , (1)} \M_{\ell , (1)}^{\top} +  \lambda_{Y,2} \mathbf{I} \big)^{-1}    \overrightarrow{\mathbf{b}}   \right\|_{\infty}
		\Bigg\}  
	\end{align*}
	which converges to $0$ as $n \rightarrow \infty$ under the condition (b)  of Theorem~\ref{thm:consistency.elastic.net.estimator:beta_k}.
	This concludes the proof of Theorem~\ref{thm:consistency.elastic.net.estimator:beta_k}. 
\end{proof}
\subsection{Proof of the asymptotic normality}
\label{sup:asymptotic_normality}
\subsubsection{Proof of Proposition~\ref{prop:upper-bound:elastic.net.beta_{k}}}   
\label{sec:proof:prop:upper-bound:elastic.net.beta_{k}}
We adapt the proof  in~\cite[Theorem  3.1]{Zou-Zhang:2009}.  
\begin{proof}[Proof of Proposition~~\ref{prop:upper-bound:elastic.net.beta_{k}}]   
	Let us define: 
	\begin{align}
		\widehat{\bm{\beta}}_{k} ( 0 ; \lambda_{Y,2} )  =  \underset{   {\beta_{k} } }{ \textrm{argmin} }  \Bigg[ \sum_{i=1}^{n}  \left(   \big(  (\Y_{i})_{k}   -     \M_{i}^{\top}\bm{\beta}_{k}  -   \X_{i}^{\top} \bm{\gamma}_{k}   -   \Z_{i}^{\top} \bm{\eta}_{k}     
		\big)^{2}   
		\right)   
		+  \lambda_{Y,2}  \big\|  \bm{\beta}_{k}   \big\|_{2}^{2}  \Bigg]  .
	\end{align}
	We have: 
	\begin{align*}  &   \sum_{i=1}^{n}  \left(   \big(  (\Y_{i})_{k}   -    \M_{i}^{\top} \widehat{\bm{\beta}}_{k}  -   \X_{i}^{\top} \bm{\gamma}_{k}   -   \Z_{i}^{\top} \bm{\eta}_{k}     
		\big)^{2}   
		\right)   
		+  \lambda_{Y,2}  \big\|  \widehat{\bm{\beta}}_{k}   \big\|_{2}^{2}   
		\\
		\geq  &   \,   
		\sum_{i=1}^{n}  \left(   \big(  (\Y_{i})_{k}  -  \M_{i}^{\top} \big(  \widehat{\bm{\beta}}_{k} ( 0 ; \lambda_{Y,2}  )   \big)  -   \X_{i}^{\top} \bm{\gamma}_{k}   -   \Z_{i}^{\top} \bm{\eta}_{k}     
		\big)^{2}   
		\right)   
		+  \lambda_{Y,2}  \big\|   \widehat{\bm{\beta}}_{k} ( 0 ; \lambda_{Y,2} )    \big\|_{2}^{2}   
	\end{align*}
	and  
	\begin{align*}
		&   \sum_{i=1}^{n}    \left(  \big(    (\Y_{i})_{k}  -  \M_{i}^{\top} \big(  \widehat{\bm{\beta}}_{k} ( 0 ; \lambda_{Y,2}  )   \big)  -   \X_{i}^{\top} \bm{\gamma}_{k}   -   \Z_{i}^{\top} \bm{\eta}_{k}     
		\big)^{2}   
		\right)   
		+  \lambda_{Y,2}  \big\|   \widehat{\bm{\beta}}_{k} ( 0 ; \lambda_{Y,2} )    \big\|_{2}^{2}   +  \lambda_{Y,1}  \big\|   \widehat{\bm{\beta}}_{k} ( 0 ; \lambda_{Y,2} )    \big\|_{1}
		\\
		\geq   &  \,  
		\sum_{i=1}^{n}  \left(   \big(  (\Y_{i})_{k}   -    \M_{i}^{\top} \widehat{\bm{\beta}}_{k}  -   \X_{i}^{\top} \bm{\gamma}_{k}   -   \Z_{i}^{\top} \bm{\eta}_{k}     
		\big)^{2}   
		\right)   
		+   \lambda_{Y,2}    \big\|   \widehat{\bm{\beta}}_{k}    \big\|_{2}^{2}   +   \lambda_{Y,1}  \big\|  \widehat{\bm{\beta}}_{k}    \big\|_{1}    
	\end{align*}
	this implies that:  
	\begin{align}
		\lambda_{Y,1}  \Big(   \big\|   \widehat{\bm{\beta}}_{k} (0;\lambda_{Y,2})    \big\|_{1}  -  \big\|   \widehat{\bm{\beta}}_{k}    \big\|_{1}   \Big)  
		&\geq      
		\Big[ \sum_{i=1}^{n}  \left(  \big(  (\Y_{i})_{k}   -    \M_{i}^{\top} \widehat{\bm{\beta}}_{k}  -   \X_{i}^{\top} \bm{\gamma}_{k}   -   \Z_{i}^{\top} \bm{\eta}_{k}     
		\big)^{2}   
		\right)   
		+   \lambda_{Y,2}    \big\|   \widehat{\bm{\beta}}_{k}    \big\|_{2}^{2} \Big]   
		\nonumber    
		\\
		& \qquad - \Big[     \sum_{i=1}^{n}  \left(   \big(  (\Y_{i})_{k}   -     \M_{i} ^{\top}\big(   \widehat{\bm{\beta}}_{k} (  0 ; \lambda_{Y,2}  )   \big)  -   \X_{i}^{\top} \bm{\gamma}_{k}   -   \Z_{i}^{\top} \bm{\eta}_{k}     
		\big)^{2}   
		\right)    \nonumber  
		\\
		& \qquad  \qquad  +    \lambda_{Y,2}    \big\|   \widehat{\bm{\beta}}_{k} ( 0 ; \lambda_{Y,2}  )    \big\|_{2}^{2} \Big].  
	\end{align}
	On the other hand, one has:
	\begin{align*}
		&\Big[ \sum_{i=1}^{n} \left(  \big(  (\Y_{i})_{k}   -    \M_{i}^{\top} \widehat{\bm{\beta}}_{k}  -   \X_{i}^{\top} \bm{\gamma}_{k}   -    \Z_{i}^{\top}\bm{\eta}_{k}  \big)^{2}   
		\right)   +   \lambda_{Y,2}  \big\|   \widehat{\bm{\beta}}_{k}   \big\|_{2}^{2} \Big]  
		\\
		& \quad - \Big[    \sum_{i=1}^{n}  \left(   \big(  (\Y_{i})_{k}   -   \M_{i} ^{\top}\big(   \widehat{\bm{\beta}}_{k} (  0 ; \lambda_{Y,2}  )   \big)  -   \X_{i}^{\top} \bm{\gamma}_{k}   -   \Z_{i}^{\top} \bm{\eta}_{k}   
		\big)^{2}  \right)   
		+   \lambda_{Y,2}    \big\|   \widehat{\bm{\beta}}_{k} ( 0 ; \lambda_{Y,2}  )    \big\|_{2}^{2} \Big]
		\\
		=& \, \Big( \widehat{\bm{\beta}}_{k}  - \widehat{\bm{\beta}}_{k} (0 ; \lambda_{Y,2})  \Big)^{\top}  
		\Big( \sum_{i=1}^{n}  \M_{i} \M_{i}^{\top} 
		+ \lambda_{Y,2} \mathbf{I}  \Big) 
		\Big( \widehat{\bm{\beta}}_{k}  - \widehat{\bm{\beta}}_{k} (0 ; \lambda_{Y,2})  \Big)
	\end{align*}   
	and 
	\begin{align*}
		\big\|   \widehat{\bm{\beta}}_{k} (0;\lambda_{Y,2})    \big\|_{1}   -    \big\|   \widehat{\bm{\beta}}_{k}    \big\|_{1}    \leq   \big\|   \widehat{\bm{\beta}}_{k} (0;\lambda_{Y,2})      -   \widehat{\bm{\beta}}_{k}    \big\|_{1}  \leq  \sqrt{p}  \,   \big\|   \widehat{\bm{\beta}}_{k} (0;\lambda_{Y,2})      -   \widehat{\bm{\beta}}_{k}    \big\|_{2} . 
	\end{align*}  
	Note that  $\Lambda_{\min} \big(  \sum_{i=1}^{n} \M_{i} \M_{i}^{\top}  + \lambda_{Y,2} \mathbf{I} \big) = \Lambda_{\min} \big(  \sum_{i=1}^{n} \M_{i} \M_{i}^{\top} \big)  + \lambda_{Y,2}$. 
	Thus, we get:  
	\begin{align*}
		&\Big( \Lambda_{\min} \big( \sum_{i=1}^{n} \M_{i} \M_{i}^{\top} \big)  + \lambda_{Y,2} \Big)  \big\|   \widehat{\bm{\beta}}_{k} (0;\lambda_{Y,2})      -   \widehat{\bm{\beta}}_{k}     \big\|_{2}^{2}   
		\\
		\leq  & \,   \Big( \widehat{\bm{\beta}}_{k}  - \widehat{\bm{\beta}}_{k} (0 ; \lambda_{Y,2})  \Big)^{\top}  
		\Big( \sum_{i=1}^{n} \M_{i} \M_{i}^{\top}   
		+ \lambda_{Y,2} \mathbf{I}  \Big) 
		\Big( \widehat{\bm{\beta}}_{k}  - \widehat{\bm{\beta}}_{k} (0 ; \lambda_{Y,2})  \Big) 
		\\
		\leq  &  \,  \lambda_{Y,1}  \Big(   \big\|   \widehat{\bm{\beta}}_{k} (0;\lambda_{Y,2})    \big\|_{1}  -  \big\|   \widehat{\bm{\beta}}_{k}    \big\|_{1}   \Big)    
		\\
		\leq  &  \,  \lambda_{Y,1} \sqrt{p} \,  \big\|   \widehat{\bm{\beta}}_{k} (0;\lambda_{Y,2})      -   \widehat{\bm{\beta}}_{k}    \big\|_{2},   
	\end{align*} 
	which implies that:
    \begin{align}
		\big\|   \widehat{\bm{\beta}}_{k} (0;\lambda_{Y,2})      -   \widehat{\bm{\beta}}_{k}    \big\|_{2}  
		\leq   \dfrac{ \lambda_{Y,1} \sqrt{p} }{  \Lambda_{\min} \big(  \sum_{i=1}^{n} \M_{i} \M_{i}^{\top} \big)  + \lambda_{Y,2} } .  
		\label{proof:prop:upper-bound:elastic.net.beta_{k}:bound1}
	\end{align}    
	Moreover, note that: 
	\begin{align*}
		\widehat{\bm{\beta}}_{k} ( 0 ; \lambda_{Y,2} ) - \bm{\beta}_{k}    &=    
		\Big( \sum_{\ell = 1}^{n} \M_{\ell}   \M_{\ell}^{\top}  + \lambda_{Y,2} \mathbf{I}      \Big)^{-1}   \sum_{i=1}^{n} \Big[  \M_{i} \M_{i}^{\top}  \bm{\beta}_{k}  +  ( \bm{\xi}_{i} )_{k} \M_{i}   \Big]  -  \bm{\beta}_{k}  
		\\
		&=   - \lambda_{Y,2}  \Big( \sum_{\ell = 1}^{n} \M_{\ell}   \M_{\ell}^{\top}  + \lambda_{Y,2} \mathbf{I}      \Big)^{-1}  \bm{\beta}_{k}  +  \Big(  \sum_{\ell = 1}^{n} \M_{\ell}   \M_{\ell}^{\top}  + \lambda_{Y,2} \mathbf{I}      \Big)^{-1} \Big(  \sum_{i=1}^{n}  (\bm{\xi}_{i})_{k}  \M_{i} \Big).
	\end{align*}  
	Thus:
	\begin{align}  
		&\E \Big( \big\|   \widehat{\bm{\beta}}_{k} (0;\lambda_{Y,2})      -   \bm{\beta}_{k}    \big\|_{2}^{2} \Big)
		\nonumber  
		\\
		\leq & \,  2 \lambda_{Y,2}^{2}   \big\|  \Big(  \sum_{\ell = 1}^{n} \M_{\ell}   \M_{\ell}^{\top}  + \lambda_{Y,2} \mathbf{I}      \Big)^{-1}  \bm{\beta}_{k}  \big\|_{2}^{2}   
		+  2 \E  \Big( \big\|  \Big(  \sum_{\ell = 1}^{n} \M_{\ell}   \M_{\ell}^{\top}  + \lambda_{Y,2} \mathbf{I}      \Big)^{-1} \Big( \sum_{i=1}^{n}  (\bm{\xi}_{i})_{k}  \M_{i}   \Big)   \big\|_{2}^{2}  \Big)
		\nonumber  
		\\
		\leq & \,  2 \lambda_{Y,2}^{2}  \Big(  \Lambda_{\min} \big( \sum_{\ell = 1}^{n} \M_{\ell}   \M_{\ell}^{\top}  \big) + \lambda_{Y,2}  \Big)^{-2}   \big\| \bm{\beta}_{k}   \big\|_{2}^{2}    \nonumber
		\\
		& \qquad  +   4np \Big(  \Lambda_{\min} \big( \sum_{\ell = 1}^{n} \M_{\ell}   \M_{\ell}^{\top}  \big) + \lambda_{Y,2}  \Big)^{-2}   \E \Big( (\bm{\xi}_{1})_{k}^{2} \Big)  \big\|  
        \M \big\|_{\infty}^{2}
		\nonumber  
		\\
		\leq &  \,  2  \Big(  \Lambda_{\min} \big(  \sum_{\ell = 1}^{n} \M_{\ell}   \M_{\ell}^{\top}  \big) + \lambda_{Y,2}  \Big)^{-2}   \Big(  \lambda_{Y,2}^{2}   \big\| \bm{\beta}_{k}   \big\|_{2}^{2}  +  2 \,n p  \,     \big\| \M \big\|^2_{\infty} \Big) .
		\label{proof:prop:upper-bound:elastic.net.beta_{k}:bound2}
	\end{align} 
	Therefore,   combining~\eqref{proof:prop:upper-bound:elastic.net.beta_{k}:bound1} and~\eqref{proof:prop:upper-bound:elastic.net.beta_{k}:bound2}, we  obtain: 
	\begin{align*}
		\E \Big( \big\|  \widehat{\bm{\beta}}_{k}  -  \bm{\beta}_{k}  \big\|_{2}^{2} \Big)   
		\leq &  \, 2 \E \Big( \big\|  \widehat{\bm{\beta}}_{k}(0 ; \lambda_{Y,2})  -  \widehat{\bm{\beta}}_{k} \big\|_{2}^{2} \Big)    +  2 \E \Big( \big\|    \bm{\beta}_{k}   -  \widehat{\bm{\beta}}_{k} (0 ; \lambda_{Y,2}) \big\|_{2}^{2} \Big)    
		\\
		\leq & \,   \dfrac{ \lambda_{Y,1}^{2} p }{ \Big(  \Lambda_{\min} \big(  \sum_{i=1}^{n} \M_{i} \M_{i}^{\top} \big)  +   \lambda_{Y,2}  \Big)^{2} }  
		\\ &  \qquad \qquad   +
		4  \Big(  \Lambda_{\min} \big(  \sum_{\ell = 1}^{n} \M_{\ell}   \M_{\ell}^{\top}  \big) + \lambda_{Y,2}  \Big)^{-2}   \Big(  \lambda_{Y,2}^{2}   \big\| \bm{\beta}_{k}   \big\|_{2}^{2}  +  2 \, np  \,    \big\| \M \big\|_{\infty} \Big)   
			%
		\\
		\leq  &  \,  
		\dfrac{ \lambda_{Y,1}^{2} p  +  4  \lambda_{Y,2}^{2}   \big\| \bm{\beta}_{k}   \big\|_{2}^{2}  +     8 \,n p  \,     \big\| \M \big\|^2_{\infty} 
		}{ \big( \delta \,  n  +   \lambda_{Y,2}  \big)^{2} },
	\end{align*} 
	where we have used condition (A1) in the last inequality.  This conclude the proof of Proposition~\ref{prop:upper-bound:elastic.net.beta_{k}}. 
\end{proof}   
\subsubsection{Proof of Theorem~\ref{thm:asymptotic.normality.elastic.net.estimator:beta_k}}
\label{sec:proof:thm:asymptotic.normality.elastic.net.estimator:beta_k}
We adapt the proof in~\cite[Theorem 3.3]{Zou-Zhang:2009}. 
For $\lambda_{Y,1}, \lambda_{Y,2} \geq  0$,  define  $\check{\bm{\beta}}_{k,(1)} \equiv  \check{\bm{\beta}}_{k,(1)} \big( \lambda_{Y,1} , \lambda_{Y,2} \big) \in \R^{d \times 1}$ by:  
\begin{align}
	\check{\bm{\beta}}_{k,(1)} = & \,  \underset{ \widetilde{\bm{\beta}}_{k,(1)}    }{\textrm{argmin}}   
	\Bigg[  \sum_{i=1}^{n}  \left(   \big(  (\Y_{i})_{k}   -     \M_{i,(1)} ^{\top}\widetilde{\bm{\beta}}_{k,(1)} -     \X_{i }^{\top}  \bm{\gamma}_{k}   -    \Z_{i }^{\top} \bm{\eta}_{k}     
	\big)^{2}   
	\right)   
	+  \lambda_{Y,2}  \big\|  \widetilde{ \bm{\beta}}_{k, (1)}   \big\|_{2}^{2}  +  \lambda_{Y,1}  \big\|  \widetilde{\bm{\beta}}_{k, (1)}   \big\|_{1}   \Bigg]  \,  .   
	\label{proof-asymptotic.normality:check{beta}_{k,(1)}.definition}    
\end{align}
First, we present  some  preliminary results for $\check{\bm{\beta}}_{k,(1)}$. 
\begin{proposition}
	\label{prop:proof.asymptotic.normality:check{beta}_{k,(1)}}
	With probability tending to $1$, $\big( \check{\bm{\beta}}_{k,(1)} , 0 \big)$ is the solution to~\eqref{equa:elastic.net:loss.function:beta_k}
	. 
\end{proposition}
\begin{lem}
	\label{lem:proof-asymptotic.normality:check{beta}_{k,(1)}.bound}
	For any $\lambda_{Y,1}, \lambda_{Y,2} > 0$, 
	\begin{align*}
		\left\| \check{\bm{\beta}}_{k,(1)} -  \check{\bm{\beta}}_{k,(1)} \big( 0 , \lambda_{Y,2} \big)  \right\|_{2}  \leq  
			\dfrac{ \lambda_{Y,1} \sqrt{p} }{\delta n + \lambda_{Y,2} } 
			.   
	\end{align*} 
\end{lem}
\noindent   
Proofs of Proposition~\ref{prop:proof.asymptotic.normality:check{beta}_{k,(1)}} and Lemma~\ref{lem:proof-asymptotic.normality:check{beta}_{k,(1)}.bound} are given in Section~\ref{sec:proof:prop:proof.asymptotic.normality:check{beta}_{k,(1)}} and~\ref{sec:proof:lem:proof-asymptotic.normality:check{beta}_{k,(1)}.bound}, respectively. 

\medskip

Now, we can start the proof of  Theorem~\ref{thm:asymptotic.normality.elastic.net.estimator:beta_k}. 
\begin{proof}[Proof of Theorem~\ref{thm:asymptotic.normality.elastic.net.estimator:beta_k}]
	For any $1 \leq k \leq T$, by Theorem~\ref{thm:consistency.elastic.net.estimator:beta_k}, with probability tending to $1$, the estimator of $\bm{\beta}_{k}$ is equal to $\Big( \check{\bm{\beta}}_{k,(1)} , 0  \Big)$.   
	Let:  
	\begin{align*}
		\Xi_{n} = \mathbf{v}^{\top} \Big( \mathbf{I} + \lambda_{Y,2}  \big( \sum_{\ell = 1}^{n} \M_{\ell,(1)}  \M_{\ell,(1)}^{\top} \big)^{-1}  \Big)  
		\Big( \sum_{\ell = 1}^{n} \M_{\ell,(1)}  \M_{\ell,(1)}^{\top} \Big)^{1/2}  \Big( \check{\bm{\beta}}_{k,(1)} - \bm{\beta}_{k,(1)} \Big)  ,     
	\end{align*}
	where $\mathbf{v}$ is a vector of norm $1$. 
	We have  
	\begin{align*}
		&\mathbf{v}^{\top}  \Big(  \mathbf{I} + \lambda_{Y,2} \big(  \sum_{\ell = 1}^{n} \M_{\ell,(1)}  \M_{\ell,(1)}^{\top} \big)^{-1} \Big) \, 
		\Big( \sum_{\ell = 1}^{n} \M_{\ell,(1)}  \M_{\ell,(1)}^{\top} \Big)^{1/2}  \Big( \check{\bm{\beta}}_{k,(1)} - \bm{\beta}_{k,(1)} \Big)  
		\\
		=& \,  \mathbf{v}^{\top}  \Big( \mathbf{I} + \lambda_{Y,2} \big(  \sum_{\ell = 1}^{n} \M_{\ell,(1)}  \M_{\ell,(1)}^{\top} \big)^{-1}  \Big) \,      
		\Big( \sum_{\ell = 1}^{n} \M_{\ell,(1)}  \M_{\ell,(1)}^{\top} \Big)^{1/2}  \Big( \check{\bm{\beta}}_{k,(1)} -  \check{\bm{\beta}}_{k,(1)} \big( 0 , \lambda_{Y,2} \big)  \Big) 
		\\
		& + \mathbf{v}^{\top}  \Big( \mathbf{I} + \lambda_{Y,2} \big( \sum_{\ell = 1}^{n} \M_{\ell,(1)}  \M_{\ell,(1)}^{\top} \big)^{-1}  \Big) \,      
		\Big( \sum_{\ell = 1}^{n} \M_{\ell,(1)}  \M_{\ell,(1)}^{\top} \Big)^{1/2}  \Big(  \check{\bm{\beta}}_{k,(1)} \big( 0 , \lambda_{Y,2} \big)  -  \bm{\beta}_{k, (1)} \Big) .
	\end{align*}
	Moreover, solving~\eqref{proof-asymptotic.normality:check{beta}_{k,(1)}.definition} with $\lambda_{Y,1} = 0$ and $\lambda_{Y,2} > 0$ gives:   
	\begin{align*}
		\check{\bm{\beta}}_{k,(1)} \big( 0 , \lambda_{Y,2} \big)
		%
		=  \Big( \sum_{\ell = 1}^{n} \M_{\ell,(1)}   \M_{\ell , (1)}^{\top}  + \lambda_{Y,2} I      \Big)^{-1}  \sum_{i=1}^{n} \Big[      \M_{i,(1)} \M_{i,(1)}^{\top}  \bm{\beta}_{k,(1)}  +  ( \xi_{i} )_{k} \M_{i,(1)}   \Big];
	\end{align*}    
	thus:   
	\begin{align*}
		& \Big( \mathbf{I} + \lambda_{Y,2} \big(  \sum_{\ell = 1}^{n} \M_{\ell,(1)}  \M_{\ell,(1)}^{\top} \big)^{-1}  \Big) \,      
		\Big( \sum_{\ell = 1}^{n} \M_{\ell,(1)}  \M_{\ell,(1)}^{\top} \Big)^{1/2}  \Big(  \check{\bm{\beta}}_{k,(1)} \big( 0 , \lambda_{Y,2} \big)  -  \bm{\beta}_{k, (1)} \Big)  
		\\
		=&  - \lambda_{Y,2}   \Big( \sum_{\ell = 1}^{n} \M_{\ell,(1)}  \M_{\ell,(1)}^{\top} \Big)^{-1/2}  \bm{\beta}_{k, (1)}  
		+  \Big( \sum_{\ell = 1}^{n} \M_{\ell,(1)}  \M_{\ell,(1)}^{\top} \Big)^{-1/2}  \Big(  \sum_{i = 1}^{n} \big( \bm{\xi}_{i} \big)_{k}  \M_{i,(1)}  \Big)  . 
	\end{align*}
	Therefore, by  
	Proposition~\ref{prop:proof.asymptotic.normality:check{beta}_{k,(1)}}, it follows that, with probability tending to $1$, $\Xi_{n} = J_{1} + J_{2} + J_{3}$, where: 
	\begin{align*}
		J_{1}  &=  - \big( \lambda_{Y,2} \big)  \mathbf{v}^{\top}  \Big(  \sum_{\ell = 1}^{n} \M_{\ell,(1)}  \M_{\ell,(1)}^{\top} \Big)^{-1/2}  \bm{\beta}_{k, (1)} , 
		\\
		J_{2}  &=    \mathbf{v}^{\top}  \Big( \mathbf{I} + \lambda_{Y,2} \big(  \sum_{\ell = 1}^{n} \M_{\ell,(1)}  \M_{\ell,(1)}^{\top} \big)^{-1}  \Big) \,      
		\Big( \sum_{\ell = 1}^{n} \M_{\ell,(1)}  \M_{\ell,(1)}^{\top} \Big)^{1/2}  \Big( \check{\bm{\beta}}_{k,(1)} -  \check{\bm{\beta}}_{k,(1)} \big( 0 , \lambda_{Y,2} \big)  \Big) , 
		\\
		J_{3}  &=  \mathbf{v}^{\top}   \Big( \sum_{\ell = 1}^{n} \M_{\ell,(1)}  \M_{\ell,(1)}^{\top} \Big)^{-1/2}  \Big(  \sum_{i = 1}^{n} \big( \bm{\xi}_{i} \big)_{k}  \M_{i,(1)}  \Big). 
	\end{align*}
	\indent   
	First, by the condition~(A1) and $\mathbf{v}^{\top} \mathbf{v} = 1$, we have: 
	\begin{align}
		J_{1}^{2}  =   \left\| \lambda_{Y,2}  \Big( \sum_{\ell = 1}^{n} \M_{\ell,(1)}  \M_{\ell,(1)}^{\top} \Big)^{-1/2}  \bm{\beta}_{k, (1)}   \right\|_{2}^{2} 
		&= \dfrac{1}{n} \left\| \lambda_{Y,2}  \Big( \dfrac{1}{n} \sum_{\ell = 1}^{n} \M_{\ell,(1)}  \M_{\ell,(1)}^{\top} \Big)^{-1/2}  \bm{\beta}_{k, (1)}   \right\|_{2}^{2}  \nonumber   
		\\
		&\leq  \dfrac{1}{n}  \lambda_{Y,2}^{2} \left\| \bm{\beta}_{k, (1)} \right\|_{2}^{2} 
		\Lambda_{\min} \Big( \dfrac{1}{n} \sum_{\ell = 1}^{n} \M_{\ell,(1)}  \M_{\ell,(1)}^{\top} \Big). 
	\end{align}
	Hence, the condition (A6) implies that $\sqrt{n} J_{1} = o(1)$.  

	\medskip
	
	Similarly, using Lemma~\ref{lem:proof-asymptotic.normality:check{beta}_{k,(1)}.bound}, we can bound $J_{2}$ as follows. For $\iota  >  0$,  	
	\begin{align*}
		& \quad J_{2}^{2}  
		\\
		& \leq  \Big( 1 +  \lambda_{Y,2}  \,   \Big[    \Lambda_{\min} \big( \sum_{\ell = 1}^{n} \M_{\ell,(1)}  \M_{\ell,(1)}^{\top} \big)  \Big]^{-1}     \Big)^{2}   \Big\|  \Big( \sum_{\ell = 1}^{n} \M_{\ell,(1)}  \M_{\ell,(1)}^{\top} \Big)^{1/2}  \Big( \check{\bm{\beta}}_{k,(1)} -  \check{\bm{\beta}}_{k,(1)} \big( 0 , \lambda_{Y,2} \big)  \Big)  \Big\|_{2}^{2}
		\\
		&\leq  \Big( 1 +  \lambda_{Y,2}  \,   \Big[    \Lambda_{\min} \big(  \dfrac{1}{n}   \sum_{\ell = 1}^{n} \M_{\ell,(1)}  \M_{\ell,(1)}^{\top} \big)  \Big]^{-1}     \Big)^{2} 
		\Lambda_{\max} \big( \dfrac{1}{n}   \sum_{\ell = 1}^{n} \M_{\ell,(1)}  \M_{\ell,(1)}^{\top} \big)    
		\Big(  \dfrac{ \lambda_{Y,1}   }{\delta  \,  n  +   \lambda_{Y,2}} \Big)^{2}  
	\end{align*} 
	and  $
		\dfrac{ \lambda_{Y,1}  }{\delta  \,  n  +   \lambda_{Y,2}}  = O \Big( \dfrac{1}{n} \Big)  $, this leads to the fact that 
	$n J_{2}^{2}  =     O_{\mathbb{P}}(1) 
	$.

	\medskip 
	
	Eventually, consider $J_{3}$ and let $R_{i} := \,  \mathbf{v}^{\top}   \Big( \sum_{\ell = 1}^{n} \M_{\ell,(1)}  \M_{\ell,(1)}^{\top} \Big)^{-1/2}     \M_{i,(1)} $. We can write: 
	\begin{align*}
		J_{3}  &=  \mathbf{v}^{\top}   \Big( \sum_{\ell = 1}^{n} \M_{\ell,(1)}  \M_{\ell,(1)}^{\top} \Big)^{-1/2}   \Big(  \sum_{i = 1}^{n} \big( \bm{\xi}_{i} \big)_{k}  \M_{i,(1)}  \Big)  
		\\
		&=   \sum_{i = 1}^{n}  \Big[  \mathbf{v}^{\top}  \Big(  \sum_{\ell = 1}^{n} \M_{\ell,(1)}  \M_{\ell,(1)}^{\top} \Big)^{-1/2}     \M_{i,(1)}  \Big]    \big( \bm{\xi}_{i} \big)_{k}  
		=  \sum_{i = 1}^{n}  R_{i}  \,   \big( \bm{\xi}_{i} \big)_{k}   . 
	\end{align*}
	On the other hand, one has: 
	\begin{align}
		\sum_{i = 1}^{n}  R_{i}^{2}   &=        \sum_{i = 1}^{n}   \mathbf{v}^{\top}  \Big(  \sum_{\ell = 1}^{n} \M_{\ell,(1)}  \M_{\ell,(1)}^{\top} \Big)^{-1/2}     \M_{i,(1)}   \M_{i,(1)}^{\top}    \Big( \sum_{\ell = 1}^{n} \M_{\ell,(1)}  \M_{\ell,(1)}^{\top} \Big)^{-1/2}  \mathbf{v}
		\nonumber  
		\\
		&=   \, \mathbf{v}^{\top}  \Big(  \sum_{\ell = 1}^{n} \M_{\ell,(1)}  \M_{\ell,(1)}^{\top} \Big)^{-1/2} 
		\Big( \sum_{i = 1}^{n} \M_{i,(1)}  \M_{i,(1)}^{\top} \Big)     
		\Big( \sum_{\ell = 1}^{n} \M_{\ell,(1)}  \M_{\ell,(1)}^{\top} \Big)^{-1/2}  \mathbf{v} 
		\nonumber   
		\\
		&=  \mathbf{v}^{\top}  \mathbf{v}   =  1  \,. 
		\label{asymptotic.normality:Lyapunov.condition:cond1}
	\end{align}
	So $\Var(J_{3}) = \Var(\xi_{1}) \displaystyle \sum_{i = 1}^{n}  R_{i}^{2} = 1 $. 
	Furthermore, we have, for $\varpi > 0 $, 
	\begin{align*}
		\sum_{i = 1}^{n}  \mathbb{E} \Big[  \big| \big( \bm{\xi}_{i} \big)_{k} \big|^{2 + \varpi} \Big] \big| R_{i} \big|^{2 + \varpi}   \leq   
		\mathbb{E} \big[  \big| \big( \bm{\xi}_{1} \big)_{k} \big|^{2 + \varpi}  \big]  \Big( \max_{1 \leq i \leq n} \big| R_{i} \big|^{ \varpi}  \sum_{i = 1}^{n}  \big| R_{i} \big|^{2}  \Big)  
		=  \,  \mathbb{E} \big[  \big| \big( \bm{\xi}_{1} \big)_{k} \big|^{2 + \varpi}  \big]   \Big(  \max_{1 \leq i \leq n}  R_{i}^{2}  \Big)^{\frac{\varpi}{2}}  \,. 
	\end{align*}
	Moreover, for all $1 \leq i \leq n$, since $\mathbf{v}^{\top} \mathbf{v}  =  1$,  
	\begin{align*}
		R_{i}^{2}  \leq   \Big\| \Big(  \sum_{\ell = 1}^{n} \M_{\ell,(1)}  \M_{\ell,(1)}^{\top} \Big)^{-1/2}     \M_{i,(1)} \Big\|_{2}^{2}  
		&=  \dfrac{1}{n} \Big\| \Big(  \dfrac{1}{n}   \sum_{\ell = 1}^{n} \M_{\ell,(1)}  \M_{\ell,(1)}^{\top} \Big)^{-1/2}     \M_{i,(1)} \Big\|_{2}^{2}  
		\\
		&\leq   \dfrac{1}{n}   \Big(  \Lambda_{\min} \big( \dfrac{1}{n}   \sum_{\ell = 1}^{n} \M_{\ell,(1)}  \M_{\ell,(1)}^{\top} \big)   \Big)^{-1} d  \,  \big\| \M  \big\|_{\infty}^{2}.
	\end{align*}
	Consequently, 
	\begin{align}
		&\dfrac{ 1 }{ \big( \sqrt{ \textrm{Var}(J_{3}) }  \big)^{2 + \varpi} } \sum_{i = 1}^{n}  \mathbb{E} \Big[  \big|  \big( \bm{\xi}_{i}   \big)_{k}   \big|^{2 + \varpi} \Big] \big| R_{i} \big|^{2 + \varpi}   
		\nonumber   
		\\
		\leq  &    \mathbb{E} \big[  \big|  \big( \bm{\xi}_{1} \big)_{k}  \big|^{2 + \varpi}  \big]   
		\Big[ \dfrac{1}{n}  \Big(  \Lambda_{\min} \big( \dfrac{1}{n} \sum_{\ell = 1}^{n} \M_{\ell,(1)}  \M_{\ell,(1)}^{\top} \big)   \Big)^{-1} d  \,  \big\| \M  \big\|_{\infty}^{2} \Big]^{ \frac{\varpi}{2}}  =  O(n^{-\frac{\varpi}{2}})  
		\overset{ n \rightarrow \infty }{\longrightarrow}  0  .  
		\label{asymptotic.normality:Lyapunov.condition:cond2}
	\end{align}
	From~\eqref{asymptotic.normality:Lyapunov.condition:cond1} and~\eqref{asymptotic.normality:Lyapunov.condition:cond2}, Lyapunov conditions for the central limit theorem are established. Hence, $\dfrac{J_{3} - \E(J_{3}) }{ \sqrt{\Var(J_{3})} } \underset{d}{  \overset{ n \rightarrow \infty }{\longrightarrow} } \mathcal{N} \big( 0 ; 1 \big) $, or  equivalently, $ \sqrt{n}  J_{3} \underset{d}{  \overset{ n \rightarrow \infty }{\longrightarrow} } \mathcal{N} \big( 0 ; 1 \big)$.    
	
	\medskip
	
	Therefore, we have shown that $\sqrt{n}  J_{1} = o(1)$, $\sqrt{n} J_{2} = o_{\mathbb{P}}(1)$ and $\sqrt{n} J_{3} \underset{d}{  \overset{ n \rightarrow \infty }{\longrightarrow} } \mathcal{N} \big( 0 ; 1 \big) $. Then, by Slutsky' theorem, we obtain that $ \sqrt{n} \,   \Xi_{n} \underset{d}{ \overset{ n \rightarrow \infty }{\longrightarrow} } \mathcal{N} \big( 0 ; 1\big)$.  
	This concludes the proof of Theorem~\ref{thm:asymptotic.normality.elastic.net.estimator:beta_k}. 
\end{proof}   
\subsubsection{Proof of Proposition~\ref{prop:proof.asymptotic.normality:check{beta}_{k,(1)}} }
\label{sec:proof:prop:proof.asymptotic.normality:check{beta}_{k,(1)}}
We follow the proof of~\cite[Theorem 3.2]{Zou-Zhang:2009}. 
\begin{proof}[Proof of Proposition~\ref{prop:proof.asymptotic.normality:check{beta}_{k,(1)}}]
	We will show that $\big( \check{\bm{\beta}}_{k,(1)} ; 0 \big)$ satisfies that KKT conditions of~\eqref{equa:elastic.net:loss.function:beta_k} with probability tending to $1$. 
	First, as in the proof of Lemma~\ref{lem:KKT-condition}, by the standard KKT conditions for optimality in a convex program, the point $\big( \check{\bm{\beta}}_{k,(1)} ; 0 \big)$ is optimal if and only if 
	\begin{align*}
		& -  \sum_{i=1}^{n} \Big[ 2 (\Y_{i})_{k} \M_{i}  -   2 \M_{i} \M_{i}^{\top}  \big( \check{\bm{\beta}}_{k,(1)} ; 0 \big)  - 2 \M_{i} \big( \X_{i}^{\top} \bm{\gamma}_{k} \big) - 2 \M_{i} \big(   \Z_{i}^{\top}\bm{\eta}_{k} \big) \Big] 
		+ 2 \lambda_{Y,2}   \big( \check{\bm{\beta}}_{k,(1)} ; 0 \big) + \lambda_{Y,1}   \check{ \Upsilon }   
		=  0,    
	\end{align*} 
	where 
$
		\check{ \Upsilon }   =  \left\{  \begin{array}{ll}
			\textrm{sign}(\check{\beta}_{jk})	&  \textrm{ if }  \check{\beta}_{jk}  \neq  0
			\\[0.2cm]   
			\textrm{any real number which } \in [-1 ; 1] 	&  \textrm{ if }  \check{\beta}_{jk}  =  0
		\end{array}
		\right. , \quad  \textrm{  with  }  1  \leq  j  \leq  p .   
$
	\\
	Then, by the definition of $\check{\bm{\beta}}_{k,(1)}$ and following the proof of Lemma~\ref{lem:KKT-condition}, it is sufficient to show that:
	\[
	\mathbb{P} \Bigg(  \bigcap_{j = d+1}^{q} \Omega_{j} \Bigg)   \overset{ n   \rightarrow  \infty    }{\longrightarrow}  1, \] 
	with $j \in \{ d+1, ..., q \}$, and
	\begin{align}
		\Omega_{j}  :=  \left\{  \Big| 2   \sum_{i=1}^{n} \Big[  \big( \M_{i,(2)} \big)_{j} \M_{i,(1)}^{\top}  \big( \sum_{\ell = 1}^{n} \M_{\ell , (1)} \M_{\ell , (1) }^{ \top } + \lambda_{Y, 2} \mathbf{I} \big)^{-1} +   \big(  \xi_{i}   \big)_{k}    \big( \M_{i,(2)} \big)_{j}    \Big]   \Big|  \leq     \lambda_{Y,1}  
		\right\}.
	\end{align}  
	Let $\varrho  = \min_{1 \leq j \leq d} | \beta_{kj} | $ and $\widehat{ \varrho }   = \min_{1 \leq j \leq d} | \widehat{\beta}_{kj} | $. We note that: 
	\begin{align*}
		&\mathbb{P} \Bigg(  \bigcup_{j = d+1}^{q} 
		\big( \Omega_{j} \big)^{c} \Bigg)
		\leq  \sum_{j=d+1}^{q}  \mathbb{P} \Big( \big( \Omega_{j} \big)^{c}   \cap  \Big( \widehat{ \varrho } > \frac{\varrho}{2} \Big) \Big)     +     \mathbb{P} \Big( \widehat{ \varrho }  \leq   \frac{\varrho}{2} \Big)
	\end{align*}
	and   
	\begin{align*}
		\mathbb{P}  \Big( \widehat{\varrho}  \leq  \frac{\varrho}{2} \Big)   \leq  \mathbb{P} \Big( \| \widehat{\bm{\beta}}_{k,(1)} - \bm{\beta}_{k , (1)} \|_{2}  \geq  \frac{\varrho}{2} \Big)  \leq  \dfrac{ \E \Big(  \| \widehat{\bm{\beta}}_{k} - \bm{\beta}_{k} \|_{2}^{2} \Big) }{\frac{ \varrho^{2} }{ 4 }} .   
	\end{align*}
	By Proposition~\ref{prop:upper-bound:elastic.net.beta_{k}}, we obtain:  
	\begin{align}
		\mathbb{P}  \Big( \widehat{\varrho}  \leq  \frac{\varrho}{2} \Big)    
		&\leq   \dfrac{  16  \lambda_{Y,2}^{2}   \big\| \bm{\beta}_{k}   \big\|_{2}^{2}+  32 \, p  \,  \big\| \M \big\|_{\infty}   +  4 \lambda_{Y,1}^{2} \,  p   
		}{ \big(  \delta \, n   +   \lambda_{Y,2}  \big)^{2}  \varrho^{2} }  .  
		\label{proof:prop:proof.asymptotic.normality:check{beta}_{k,(1)}:bound1}
	\end{align} 
	On the other hand, 
	we have:   
	\begin{align*}
		&\sum_{j=d+1}^{q}  \mathbb{P} \Big( \big( \Omega_{j} \big)^{c}   \cap  \Big( \widehat{ \varrho } > \frac{\varrho}{2} \Big)  \Big)   
		\\
		\leq  &  \,    
		\dfrac{4 n^{2}}{\lambda_{Y,1}^{2}} \,  \E \Bigg\{  \sum_{j=d+1}^{q}  \Big| \dfrac{1}{n} \sum_{i=1}^{n} \Big[  \big( \M_{i,(2)}   \big)_{j}  \,   \M_{i,(1)}^{\top}  \big( \bm{\beta}_{k,(1)}  - \check{\bm{\beta}}_{k,(1)} \big)    +     \big( \bm{\xi}_{i}   \big)_{k}     \big( \M_{i,(2)}  \big)_{j}      \Big]   \Big|^{2}
		\mathbf{1}_{ \big\{ \widehat{ \varrho } \, > \frac{\varrho}{2} \big\} }  \Bigg\}.  
	\end{align*}
	Moreover, one gets:    
	\begin{align*}
		& \sum_{j=d+1}^{q}  \Big|  \dfrac{1}{n}  \sum_{i=1}^{n} \Big[  \big( \M_{i,(2)}   \big)_{j}  \,   \M_{i,(1)}^{\top}  \big( \bm{\beta}_{k,(1)}  - \check{\bm{\beta}}_{k,(1)} \big)    +     \big( \bm{\xi}_{i}   \big)_{k}     \big( \M_{i,(2)}  \big)_{j}      \Big]   \Big|^{2}      
		\\
		\leq  &  \,  2   \sum_{j=d+1}^{q}  \Big| \dfrac{1}{n}  \sum_{i=1}^{n} \Big[  \big( \M_{i,(2)}   \big)_{j}  \,   \M_{i,(1)}^{\top}  \big( \bm{\beta}_{k,(1)}  - \check{\bm{\beta}}_{k,(1)} \big)  \Big]   \Big|^{2}   
		+ 2   \sum_{j=d+1}^{q}  \Big| \dfrac{1}{n}  \sum_{i=1}^{n} \big( \bm{\xi}_{i}   \big)_{k}    \big( \M_{i,(2)}  \big)_{j}      \Big|^{2}   
		\\
		\leq  & \,  2 \Big[  \Lambda_{\max} \Big( \dfrac{1}{n} \sum_{i = 1}^{n}  \M_{i} \M_{i}^{\top}  \Big) \Big]^{2}  \big\|  \bm{\beta}_{k,(1)}  - \check{\bm{\beta}}_{k,(1)}  \big\|_{2}^{2}  
		+   2   \sum_{j=d+1}^{q}  \Big| \dfrac{1}{n}   \sum_{i=1}^{n} \big( \xi_{i}   \big)_{k}    \big( \M_{i,(2)}  \big)_{j}      \Big|^{2} 
	\end{align*}
	and 
	\begin{align*}
		\E \Big[ \sum_{j=d+1}^{q}  \Big| \dfrac{1}{n}   \sum_{i=1}^{n} \big( \bm{\xi}_{i}  \big)_{k}  \big( \M_{i,(2)}  \big)_{j}    \Big|^{2}  \Big]  \leq  \dfrac{q}{n^{2}} \big\| \M \big\|_{\infty}^{2}   \sum_{i=1}^{n} \E \Big[ \big( \bm{\xi}_{i}   \big)_{k}^{2} \Big] = \dfrac{q}{n} \big\| \M \big\|_{\infty}^{2}   , 
	\end{align*}
	which give us the inequality:
	\begin{align}  
		&\E \Bigg\{  \sum_{j=d+1}^{q}  \Big| \dfrac{1}{n}   \sum_{i=1}^{n} \Big[  \big( \M_{i,(2)}   \big)_{j}  \,   \M_{i,(1)}^{\top}  \big( \bm{\beta}_{k,(1)}  - \check{\bm{\beta}}_{k,(1)} \big)    +     \big( \bm{\xi}_{i}   \big)_{k}     \big( \M_{i,(2)}  \big)_{j}      \Big]   \Big|^{2}
		\mathbf{1}_{ \big\{ \widehat{ \varrho } \, > \frac{\varrho}{2} \big\} }  \Bigg\}  
		\nonumber   
		\\
		\leq  &  \,    2 \Big[  \Lambda_{\max} \Big( \dfrac{1}{n} \sum_{i = 1}^{n}  \M_{i} \M_{i}^{\top}  \Big) \Big]^{2}   \E \Big[ \big\|  \bm{\beta}_{k,(1)}  - \check{\bm{\beta}}_{k,(1)}  \big\|_{2}^{2}   \mathbf{1}_{ \big\{ \widehat{ \varrho } \, > \frac{\varrho}{2} \big\} }   \Big]
		+ 2 \dfrac{q}{n}   \big\| \M \big\|_{\infty}^{2}       .  
		\label{proof:prop:proof.asymptotic.normality:check{beta}_{k,(1)}:bound2}
	\end{align}
	We now bound $\E \Big[ \big\|  \bm{\beta}_{k,(1)}  - \check{\bm{\beta}}_{k,(1)}  \big\|_{2}^{2}   \mathbf{1}_{ \big\{ \widehat{ \varrho } \, > \frac{\varrho}{2} \big\} }   \Big]$. Let  $  \check{\bm{\beta}}_{k,(1)} \big( 0 , \lambda_{Y,2} \big) \in \R^{d \times 1}$ be defined by:  
	\begin{align*}
		\check{\bm{\beta}}_{k,(1)} \big( 0 , \lambda_{Y,2} \big) = & \,  \underset{ \widetilde{\bm{\beta}}_{k,(1)}    }{\textrm{argmin}} 
		\Bigg[     \sum_{i=1}^{n}  \left(   \big(  (\Y_{i})_{k}   -     \M_{i,(1)} ^{\top}\widetilde{\bm{\beta}}_{k,(1)} -     \X_{i }^{\top}  \bm{\gamma}_{k}   -    \Z_{i }^{\top} \bm{\eta}_{k}     
		\big)^{2}   
		\right)   
		+  \lambda_{Y,2}  \big\|  \widetilde{ \bm{\beta}}_{k, (1)}   \big\|_{2}^{2}     \Bigg]  \,  .   
	\end{align*}
	By 
	Lemma~\ref{lem:proof-asymptotic.normality:check{beta}_{k,(1)}.bound}, we attain:  
	\begin{align}
		\left\| \check{\bm{\beta}}_{k,(1)} -  \check{\bm{\beta}}_{k,(1)} \big( 0 , \lambda_{Y,2} \big)  \right\|_{2}  \leq  
		\dfrac{ \lambda_{Y,1}  \sqrt{p} }{\delta n + \lambda_{Y,2} } .  
	\end{align}
	Moreover, note that:  
	\begin{align*}
		&\check{\bm{\beta}}_{k,(1)} ( 0 ; \lambda_{Y,2} ) - \bm{\beta}_{k,(1)}  
		\\ 
		&=      
		\Big(   \dfrac{1}{n} \sum_{\ell = 1}^{n} \M_{\ell,(1)}   \M_{\ell,(1)}^{\top}  + \lambda_{Y,2} \mathbf{I}      \Big)^{-1}   \dfrac{1}{n} \sum_{i=1}^{n} \Big[   \M_{i,(1)} \M_{i,(1)}^{\top}  \bm{\beta}_{k,(1)}  +  ( \bm{\xi}_{i} )_{k} \M_{i,(1)}   \Big]  -  \bm{\beta}_{k,(1)}  
		\\
		&=   - \lambda_{Y,2}  \Big(   \dfrac{1}{n} \sum_{\ell = 1}^{n} \M_{\ell,(1)}   \M_{\ell,(1)}^{\top}  + \lambda_{Y,2} \mathbf{I}      \Big)^{-1}  \bm{\beta}_{k,(1)}  +  \Big(   \dfrac{1}{n} \sum_{\ell = 1}^{n} \M_{\ell,(1)}   \M_{\ell,(1)}^{\top}  + \lambda_{Y,2} \mathbf{I}      \Big)^{-1} \Big( \dfrac{1}{n} \sum_{i=1}^{n}  (\bm{\xi}_{i})_{k}  \M_{i,(1)}   \Big).
	\end{align*}  
	Thus, 
	\begin{align*}  
		&\E \Big( \big\|   \check{\bm{\beta}}_{k,(1)} (0;\lambda_{Y,2})      -   \bm{\beta}_{k,(1)}    \big\|_{2}^{2} \Big)
		\nonumber  
		\\
		\leq & \,  2 \lambda_{Y,2}^{2}   \big\|    \Big(   \dfrac{1}{n} \sum_{\ell = 1}^{n} \M_{\ell,(1)}   \M_{\ell,(1)}^{\top}  + \lambda_{Y,2} \mathbf{I}      \Big)^{-1}  \bm{\beta}_{k,(1)}  \big\|_{2}^{2}   \nonumber  
		\\
		& \qquad  \qquad   +  2 \E  \Big( \big\|  \Big(   \dfrac{1}{n} \sum_{\ell = 1}^{n} \M_{\ell,(1)}   \M_{\ell,(1)}^{\top}  + \lambda_{Y,2} \mathbf{I}      \Big)^{-1} \Big( \dfrac{1}{n} \sum_{i=1}^{n}  (\bm{\xi}_{i})_{k}  \M_{i,(1)}   \Big)   \big\|_{2}^{2}  \Big)
		\nonumber  
		\\
		\leq & \,  2 \lambda_{Y,2}^{2}  \Big(  \Lambda_{\min} \big( \dfrac{1}{n} \sum_{\ell = 1}^{n} \M_{\ell,(1)}   \M_{\ell,(1)}^{\top}  \big) + \lambda_{Y,2}  \Big)^{-2}   \big\| \bm{\beta}_{k,(1)}   \big\|_{2}^{2}   \nonumber   
		\\
		& \qquad  \qquad  +  2 \Big(  \Lambda_{\min} \big( \dfrac{1}{n} \sum_{\ell = 1}^{n} \M_{\ell,(1)}   \M_{\ell,(1)}^{\top}  \big) + \lambda_{Y,2}  \Big)^{-2} \dfrac{2}{n}  \E \Big( (\bm{\xi}_{1})_{k}^{2} \Big) p \big\|  \M \big\|_{\infty}^{2}
		\nonumber  
		\\
		\leq &  \,  2  \Big(  \Lambda_{\min} \big( \dfrac{1}{n} \sum_{\ell = 1}^{n} \M_{\ell,(1)}   \M_{\ell,(1)}^{\top}  \big) + \lambda_{Y,2}  \Big)^{-2}   \Big(  \lambda_{Y,2}^{2}   \big\| \bm{\beta}_{k,(1)}   \big\|_{2}^{2}  +   \dfrac{2 \, p}{n}  \,    \big\| \M \big\|_{\infty} \Big) .
	\end{align*} 
	Note that $\Lambda_{\min} \Big( \frac{1}{n} \sum_{i=1}^{n} \M_{i,(1)} \M_{i,(1)}^{\top} \Big) \geq \Lambda_{\min} \Big( \frac{1}{n} \sum_{i=1}^{n} \M_{i} \M_{i}^{\top} \Big) \geq \delta$ 
	and $\Lambda_{\max} \Big( \frac{1}{n} \sum_{i=1}^{n} \M_{i,(1)} \M_{i,(1)}^{\top} \Big) \leq \Lambda_{\max} \Big( \frac{1}{n} \sum_{i=1}^{n} \M_{i} \M_{i}^{\top} \Big) \leq \Delta$. Following similar arguments in proof of  of Proposition~\ref{prop:upper-bound:elastic.net.beta_{k}}, we deduce:  
	\begin{align}
		\E \Big[ \big\|  \bm{\beta}_{k,(1)}  - \check{\bm{\beta}}_{k,(1)}  \big\|_{2}^{2}   \mathbf{1}_{ \big\{ \widehat{ \varrho } \, > \frac{\varrho}{2} \big\} }   \Big]  
		\leq & \,  2 \E \Big[ \big\|  \bm{\beta}_{k,(1)}  - \check{\bm{\beta}}_{k,(1)} \big( 0 , \lambda_{Y,2} \big) \big\|_{2}^{2}   \mathbf{1}_{ \big\{ \widehat{ \varrho } \, > \frac{\varrho}{2} \big\} }   \Big]   \nonumber  
		\\
		& \qquad  +  2  \E \Big[ \big\|  \check{\bm{\beta}}_{k,(1)} \big( 0 , \lambda_{Y,2} \big)  - \check{\bm{\beta}}_{k,(1)}  \big\|_{2}^{2}   \mathbf{1}_{ \big\{ \widehat{ \varrho } \, > \frac{\varrho}{2} \big\} }   \Big]    
		\nonumber   
		\\  
		\leq & \,   \dfrac{   4  \lambda_{Y,2}^{2}   \big\| \bm{\beta}_{k,(1)}   \big\|_{2}^{2}  +    \dfrac{8 \, p}{n}  \,     \big\| \M \big\|_{\infty} 
			+  2 \lambda_{Y,1}^{2} p 
		}{ \Big(  \Lambda_{\min} \big( \frac{1}{n} \sum_{i=1}^{n} \M_{i} \M_{i}^{\top} \big)  +   \lambda_{Y,2}  \Big)^{2} }
		\nonumber   
		\\
		\leq  &  \,  
		\dfrac{ 4  \lambda_{Y,2}^{2}   \big\| \bm{\beta}_{k,(1)}   \big\|_{2}^{2}  +    \dfrac{8 \, p}{n}  \,    \big\| \M \big\|_{\infty} 
			+  2 \lambda_{Y,1}^{2} p  
		}{ \big(  \delta \, n   +   \lambda_{Y,2}  \big)^{2} }
		.
		\label{proof:prop:proof.asymptotic.normality:check{beta}_{k,(1)}:bound3}
	\end{align}
	Now, the combination of~\eqref{proof:prop:proof.asymptotic.normality:check{beta}_{k,(1)}:bound1},~\eqref{proof:prop:proof.asymptotic.normality:check{beta}_{k,(1)}:bound2}~and~\eqref{proof:prop:proof.asymptotic.normality:check{beta}_{k,(1)}:bound3} yields:  
	\begin{align*}
		\mathbb{P} \Bigg(  \bigcup_{j = d+1}^{q} 
		\big( \Omega_{j} \big)^{c} \Bigg)  &\leq 
		\dfrac{8 n^{2}}{\lambda_{Y,1}^{2}}   \Bigg \{  \Big[  \Lambda_{\max} \Big( \dfrac{1}{n}   \sum_{i = 1}^{n}  \M_{i} \M_{i}^{\top}  \Big) \Big]^{2}   \E \Big[ \big\|  \bm{\beta}_{k,(1)}  - \check{\bm{\beta}}_{k,(1)}  \big\|_{2}^{2}   \mathbf{1}_{ \big\{ \widehat{ \varrho } \, > \frac{\varrho}{2} \big\} }   \Big]
		+  \dfrac{q}{n}      \big\| \M   \big\|_{\infty}^{2}    \Bigg\} 
		\\
		& \quad  +
		\dfrac{  16  \lambda_{Y,2}^{2}   \big\| \bm{\beta}_{k}   \big\|_{2}^{2}+ 32\, p \,  \big\| \M \big\|_{\infty}   +  4 \lambda_{Y,1}^{2} \,  p   
		}{ \big(  \delta \, n   +   \lambda_{Y,2}  \big)^{2}  \varrho^{2} } 
		\\
	& \leq   \dfrac{8 n^{2}}{\lambda_{Y,1}^{2}}\Bigg \{   \Delta^{2}  \times   \dfrac{ 4  \lambda_{Y,2}^{2}   \big\| \bm{\beta}_{k,(1)}   \big\|_{2}^{2}  +    \dfrac{8 \, p}{n}  \,    \big\| \M \big\|_{\infty} 
			+  2 \lambda_{Y,1}^{2} p  
		}{ \big(  \delta \, n   +   \lambda_{Y,2}  \big)^{2} }   
		+  \dfrac{q}{n}    \big\| \M   \big\|_{\infty}^{2}    \Bigg\} 
		\\
		& \quad  +
		\dfrac{  16  \lambda_{Y,2}^{2}   \big\| \bm{\beta}_{k}   \big\|_{2}^{2}+ 32 \, p \,  \big\| \M \big\|_{\infty}   +  4 \lambda_{Y,1}^{2} \,  p   
		}{ \big(  \delta \, n   +   \lambda_{Y,2}  \big)^{2}  \varrho^{2} },
	\end{align*}
	which implies that $\mathbb{P} \Bigg(  \bigcup_{j = d+1}^{q} 
	\big( \Omega_{j} \big)^{c} \Bigg)    \overset{ n \rightarrow \infty }{ \longrightarrow } 0 $.
	The proof of Proposition~\ref{prop:proof.asymptotic.normality:check{beta}_{k,(1)}} is complete. 
\end{proof}
\subsubsection{Proof of Lemma~\ref{lem:proof-asymptotic.normality:check{beta}_{k,(1)}.bound}}
\label{sec:proof:lem:proof-asymptotic.normality:check{beta}_{k,(1)}.bound}
\begin{proof}[Proof of Lemma~\ref{lem:proof-asymptotic.normality:check{beta}_{k,(1)}.bound}]
	We have: 
	\begin{align*}  &    \sum_{i=1}^{n}  \left(   \big(  (\Y_{i})_{k}   -    \M_{i,(1)}^{\top} \check{\bm{\beta}}_{k,(1)}  -   \X_{i}^{\top} \bm{\gamma}_{k}   -   \Z_{i}^{\top} \bm{\eta}_{k}    \big)^{2}   
	\right)   
	+  \lambda_{Y,2}  \big\|  \check{\bm{\beta}}_{k,(1)}   \big\|_{2}^{2}   
	\\
	\geq  & \,  \sum_{i=1}^{n}  \left(   \big(  (\Y_{i})_{k}  -  \big( \M_{i,(1)} ^{\top} \check{\bm{\beta}}_{k,(1)} ( 0 ; \lambda_{Y,2}  )   \big)  -   \X_{i}^{\top} \bm{\gamma}_{k}   -   \Z_{i}^{\top} \bm{\eta}_{k}     
	\big)^{2}   
	\right)   
	+  \lambda_{Y,2}  \big\|   \check{\bm{\beta}}_{k,(1)} ( 0 ; \lambda_{Y,2} )    \big\|_{2}^{2}   
	\end{align*}
	and  
	\begin{align*}
		&  \sum_{i=1}^{n}   \left(   \big(    (\Y_{i})_{k}  -  \big( \M_{i,(1)} ^{\top} \check{\bm{\beta}}_{k,(1)} ( 0 ; \lambda_{Y,2}  )   \big)  -   \X_{i}^{\top} \bm{\gamma}_{k}   -   \Z_{i}^{\top} \bm{\eta}_{k}     
		\big)^{2}   
		\right)   
		\\
		& \qquad \qquad  +  \lambda_{Y,2}  \big\|   \check{\bm{\beta}}_{k,(1)} ( 0 ; \lambda_{Y,2} )    \big\|_{2}^{2}   
	 +  \lambda_{Y,1}  \big\|   \check{\bm{\beta}}_{k,(1)} ( 0 ; \lambda_{Y,2} )    \big\|_{1}
		\\
		\geq  &  \, \sum_{i=1}^{n}  \left(   \big(  (\Y_{i})_{k}   -    \M_{i,(1)}^{\top} \check{\bm{\beta}}_{k,(1)}  -   \X_{i}^{\top} \bm{\gamma}_{k}   -   \Z_{i}^{\top} \bm{\eta}_{k}     
		\big)^{2}   
		\right)   
		+   \lambda_{Y,2}    \big\|   \check{\bm{\beta}}_{k,(1)}    \big\|_{2}^{2}   +   \lambda_{Y,1}  \big\|  \check{\bm{\beta}}_{k,(1)}    \big\|_{1}.    
	\end{align*}
	This implies that:  
	\begin{align}
		&\lambda_{Y,1}  \Big(   \big\|   \check{\bm{\beta}}_{k,(1)} (0;\lambda_{Y,2})    \big\|_{1}  -  \big\|   \check{\bm{\beta}}_{k,(1)}    \big\|_{1}   \Big)   \nonumber  
		\\  
		\geq  &  \,   
		\Big[     \sum_{i=1}^{n}  \left(   \big(  (\Y_{i})_{k}   -    \M_{i,(1)}^{\top} \check{\bm{\beta}}_{k,(1)}  -   \X_{i}^{\top} \bm{\gamma}_{k}   -   \Z_{i}^{\top} \bm{\eta}_{k}     
		\big)^{2}   
		\right)   
		+   \lambda_{Y,2}    \big\|   \check{\bm{\beta}}_{k,(1)}    \big\|_{2}^{2} \Big]  
		\\
		& - \Big[ \sum_{i=1}^{n}  \left(   \big(  (\Y_{i})_{k}   -      \M_{i,(1)}^{\top}\big(   \check{\bm{\beta}}_{k,(1)} (  0 ; \lambda_{Y,2}  )   \big)  -   \X_{i}^{\top} \bm{\gamma}_{k}   -   \Z_{i}^{\top} \bm{\eta}_{k}     
		\big)^{2}   
		\right)   
		+   \lambda_{Y,2}    \big\|   \check{\bm{\beta}}_{k,(1)} ( 0 ; \lambda_{Y,2}  )    \big\|_{2}^{2} \Big]    .  
		\nonumber   
	\end{align}   
	On the other hand, one has: 
	\begin{align*}
		&\Big[ \sum_{i=1}^{n}  \left(   \big(  (\Y_{i})_{k}   -    \M_{i,(1)}^{\top} \check{\bm{\beta}}_{k,(1)}  -   \X_{i}^{\top} \bm{\gamma}_{k}   -   \Z_{i}^{\top} \bm{\eta}_{k}     
		\big)^{2}   
		\right)   
		+   \lambda_{Y,2}    \big\|   \check{\bm{\beta}}_{k,(1)}    \big\|_{2}^{2} \Big]  
		\\
		& - \Big[  \sum_{i=1}^{n}  \left(   \big(  (\Y_{i})_{k}   -      \M_{i,(1)}^{\top}\big(   \check{\bm{\beta}}_{k,(1)} (  0 ; \lambda_{Y,2}  )   \big)  -   \X_{i}^{\top} \bm{\gamma}_{k}   -   \Z_{i}^{\top} \bm{\eta}_{k}     
		\big)^{2}   
		\right)   
		+   \lambda_{Y,2}    \big\|   \check{\bm{\beta}}_{k,(1)} ( 0 ; \lambda_{Y,2}  )    \big\|_{2}^{2} \Big]    .  
		\\
		=& \, \Big( \check{\bm{\beta}}_{k,(1)}  - \check{\bm{\beta}}_{k,(1)} (0 ; \lambda_{Y,2})  \Big)^{\top}  
		\Big( \sum_{i=1}^{n} \M_{i,(1)} \M_{i,(1)}^{\top}   
		+ \lambda_{Y,2} \mathbf{I}  \Big) 
		\Big( \check{\bm{\beta}}_{k,(1)}  - \check{\bm{\beta}}_{k,(1)} (0 ; \lambda_{Y,2})  \Big)
	\end{align*}   
	and 
	\begin{align*}
		\big\|   \check{\bm{\beta}}_{k,(1)} (0;\lambda_{Y,2})    \big\|_{1}   -    \big\|   \check{\bm{\beta}}_{k,(1)}    \big\|_{1}    \leq   \big\|   \check{\bm{\beta}}_{k,(1)}   (0;\lambda_{Y,2})      -   \check{\bm{\beta}}_{k,(1)}    \big\|_{1}  \leq  \sqrt{p}  \,   \big\|   \check{\bm{\beta}}_{k,(1)} (0;\lambda_{Y,2})      -   \check{\bm{\beta}}_{k,(1)}    \big\|_{2} . 
	\end{align*}  
	Note that $\Lambda_{\min} \big(  \sum_{i=1}^{n} \M_{i,(1)} \M_{i,(1)}^{\top}  + \lambda_{Y,2} \mathbf{I} \big) = \Lambda_{\min} \big(  \sum_{i=1}^{n} \M_{i,(1)} \M_{i,(1)}^{\top} \big)  + \lambda_{Y,2}  $. 
	Thus, we get:  
	\begin{align*}
		&\Big( \Lambda_{\min} \big( \sum_{i=1}^{n} \M_{i,(1)} \M_{i,(1)}^{\top} \big)  + \lambda_{Y,2} \Big)  \big\|   \check{\bm{\beta}}_{k,(1)} (0;\lambda_{Y,2})      -   \check{\bm{\beta}}_{k,(1)}     \big\|_{2}^{2}   
		\\
		\leq  & \,   \Big( \check{\bm{\beta}}_{k,(1)}  - \check{\bm{\beta}}_{k,(1)} (0 ; \lambda_{Y,2})  \Big)^{\top}  
		\Big( \sum_{i=1}^{n}  \M_{i} \M_{i}^{\top}   
		+ \lambda_{Y,2} \mathbf{I}  \Big) 
		\Big( \check{\bm{\beta}}_{k,(1)}  - \check{\bm{\beta}}_{k,(1)} (0 ; \lambda_{Y,2})  \Big) 
		\\
		\leq  &  \,  \lambda_{Y,1}  \Big(   \big\|   \check{\bm{\beta}}_{k,(1)} (0;\lambda_{Y,2})    \big\|_{1}  -  \big\|   \check{\bm{\beta}}_{k,(1)}    \big\|_{1}   \Big)    
		\\
		\leq  &  \,  \lambda_{Y,1} \sqrt{p} \,  \big\|   \check{\bm{\beta}}_{k,(1)}  (0;\lambda_{Y,2})      -   \check{\bm{\beta}}_{k,(1)}    \big\|_{2},   
	\end{align*} 
	which implies that:
    \begin{align}
		\big\|   \check{\bm{\beta}}_{k,(1)} (0;\lambda_{Y,2})      -   \check{\bm{\beta}}_{k,(1)}    \big\|_{2}  
		\leq   \dfrac{ \lambda_{Y,1} \sqrt{p} }{  \Lambda_{\min} \big(  \sum_{i=1}^{n} \M_{i,(1)} \M_{i,(1)}^{\top} \big)  + \lambda_{Y,2} }  
		\leq  \dfrac{ \lambda_{Y,1} \sqrt{p} }{\delta \, n + \lambda_{Y,2}}.  
		\label{proof:lem:proof-asymptotic.normality:check{beta}_{k,(1)}.bound}
	\end{align}    
	This conclude the proof of Lemma~\ref{lem:proof-asymptotic.normality:check{beta}_{k,(1)}.bound}. 
\end{proof}
\subsection{Proof of the mediation effect \texorpdfstring{$\bm{\alpha}\bm{\beta}$}{alpha beta}}
\label{sup:alpha_beta}
\label{sec:proof:thm:asymptotic.normality.elastic.net.estimator:alpha.beta_k}
\begin{proof}[Proof of Theorem~\ref{thm:asymptotic.normality.elastic.net.estimator:alpha.beta_k}]
	To study the mediation effect $(\bm{\alpha}\bm{\beta})$,  from model ($\mathfrak{M}$) in~\eqref{equa:MMAMA.model:selected.genes:equa1}-\eqref{equa:MMAMA.model:selected.genes:equa2}, for $1 \leq i \leq n$, we can rewrite the model as below:     
	\begin{align}
		\ddot{\mathbf{y}}_{i}  &=  \bm{\beta}^{\top} \Big(  { \X_{i} }^{\top}\bm{\alpha}  + \Z_{i}^{\top}  \bm{\zeta}  +  \bm{\epsilon}_{i}  \Big)  +  { \X_{i} }^{\top}  \bm{\gamma}  +     \Z_{i}^{\top} \bm{\eta} +  \bm{\xi}_{i} 
		\nonumber    
		\\
		&=    { \X_{i} } ^{\top}\Big( \bm{\alpha}   \bm{\beta}   +  \bm{\gamma} \Big) 
		+   \Z_{i}^{\top}\Big( \bm{\zeta}  \bm{\beta}  +  \bm{\eta} \Big)   +  \Big( \bm{\beta}^{\top} \bm{\epsilon}_{i}  +  \bm{\xi}_{i}  \Big).
	\end{align}
	Moreover, $\big \{ ( \bm{\epsilon}_{i} )_{j} \big \}_{1 \leq  j \leq p}  \overset{\textrm{i.i.d.}}{\sim} \mathcal{N}(0 ; 1) $, $\big \{ ( \bm{\xi}_{i} )_{\ell} \big \}_{1 \leq  \ell \leq T}    \overset{\textrm{i.i.d.}}{\sim} \mathcal{N}(0 ; 1) $, and $( \bm{\epsilon}_{i} )_{k}$ and $( \bm{\xi}_{i} )_{\ell}$ are independent. Thus, for $1 \leq k \leq T$, we obtain:     
	\[ 
	(\varkappa_{i})_{k} := ( \bm{\beta}^{\top} \bm{\epsilon}_{i})_{k}  +  (\bm{\xi}_{i})_{k}   =  \sum_{j=1}^{p} {\beta}_{j, k}  ( \bm{\epsilon}_{i} )_{j}  +  (\bm{\xi}_{i})_{k}  
	\overset{\textrm{i.i.d.}}{\sim} \mathcal{N}(0 ;  \sigma_{\varkappa}^{2})  , 
	\]
	with $\sigma_{\varkappa}^{2} =  \big\| \bm{\beta}_{k} \big\|_{2}^{2} + 1$.  
	Let $\bm{\psi} = \bm{\alpha}  \bm{\beta} +  \bm{\gamma}$ and $\widehat{\bm{\psi}}$ be the corresponding estimator of $\bm{\psi}$. 

	\medskip 
	
	Now, for the consistency of $(\widehat{\bm{\alpha}} \widehat{\bm{\beta}})$, using similar arguments in the proof of Theorem~\ref{thm:consistency.elastic.net.estimator:beta_k}, we obtain that for $1 \leq k \leq T$, with probability tending to $1$, $\widehat{\bm{\psi}}_{k}   \longrightarrow  \bm{\psi}_{k}$. 
	Note that:  
	\begin{align*}
		\bm{\psi}_{k} = \big( \bm{\alpha} \bm{\beta} \big)_{k}  +  \bm{\gamma
        }_{k}  
		= \left(
		\begin{array}{c}
			\sum_{j=1}^{p} \alpha_{1j} \beta_{jk} 
			\\
			\vdots 
			\\
			\sum_{j=1}^{p} \alpha_{q j} \beta_{jk} 
		\end{array}
		\right)  +  \bm{\gamma}_{k} 
		. 
	\end{align*}
	Furthermore, 
	we have already showed that, for any $1 \leq \ell \leq q$ and $1 \leq k \leq T$: 
	\begin{align*}
		\widehat{\bm{\alpha}}_{\ell}    \overset{\mathbb{P}}{\longrightarrow}  \bm{\alpha}_{\ell} , \quad  \widehat{\bm{\beta}}_{k}    \overset{\mathbb{P}}{\longrightarrow}  \bm{\beta}_{k}   \textrm{ and }  \widehat{\bm{\gamma}}_{k}    \overset{\mathbb{P}}{\longrightarrow}  \bm{\gamma}_{k},.
	\end{align*}
	So, for any $1 \leq  k  \leq  T$, with probability tending to $1$, we obtain  $\big( \widehat{\bm{\alpha}} \widehat{\bm{\beta}} \big)_{k}   \longrightarrow  \big( \bm{\alpha} \bm{\beta} \big)_{k} $. 
	
	\medskip

	On the other hand, for the asymptotic normality, using similar arguments in the proof of Theorem~\ref{thm:asymptotic.normality.elastic.net.estimator:beta_k}, we obtain the following result for the estimators of $\bm{\psi}$ and $\bm{\gamma}$.
	\begin{corollary}  
		\label{thm:asymptotic.normality.elastic.net.estimator:theta_k}
		Under conditions \textrm{(A1)-(A6)}, for $1 \leq k \leq T$,  
		\begin{align*}
			\sqrt{n}  \Big[   \mathbf{v}^{\top}  \Big( \mathbf{I} +  \lambda_{Y,2} \big( \sum_{\ell = 1}^{n}   \X_{\ell , (1)} \X_{\ell , (1)}^{\top} \big)^{-1}  \Big)  \Big(     \sum_{\ell = 1}^{n}    \X_{\ell , (1)} \X_{\ell , (1)}^{\top} \Big)^{1/2} \Big( \widehat{\bm{\psi}}   _{k,(1)} - \bm{\psi} _{k,(1)} \Big)   \Big]   
			\overset{ d }{\longrightarrow} \mathcal{N} \big( 0 ; \sigma_{\varkappa}^{2}  \big) \quad  , 
		\end{align*}   
		and 
		\begin{align*}
			\sqrt{n}  \Big[   \mathbf{v}^{\top}  \Big( \mathbf{I} + \lambda_{Y,2} \big( \sum_{\ell = 1}^{n}   \X_{\ell , (1)} \X_{\ell , (1)}^{\top} \big)^{-1}  \Big)  \Big(  \sum_{\ell = 1}^{n}    \X_{\ell , (1)} \X_{\ell , (1)}^{\top} \Big)^{1/2} \Big( \widehat{\bm{\gamma}}_{k,(1)} - \bm{\gamma}_{k,(1)} \Big)   \Big]   
			\overset{ d }{\longrightarrow} \mathcal{N} \big( 0 ; 1\big) \quad  , 
		\end{align*} 
		where $\mathbf{v}$ is a vector of norm $1$. 
	\end{corollary}   
	The asymptotic normality of $\big( \widehat{\bm{\alpha}} \widehat{\bm{\beta}} \big)_{k}$ is derived consequently. 
	This completes the proof of Theorem~\ref{thm:asymptotic.normality.elastic.net.estimator:alpha.beta_k}. 
\end{proof}

\subsection{Proof of causal mediation in the many--to--many--to--many mediation analysis setting}
\label{sec:proof:causal.mediation}
From the regression model~\eqref{equa:MMAMA.model:selected.genes:equa2}, drop the notation for subject $i \in \{1,...,n\}$ for simplicity, we can write for each $1 \leq k \leq T$ and each $1 \leq j \leq p$ as follows: 
\begin{align}
	(\ddot{\mathbf{m}})_{j} &= \big( \bm{\alpha}^{\top} \ddot{\mathbf{x}} \big)_{j}  +  \big( \bm{\zeta}^{\top} \ddot{\mathbf{z}} \big)_{j}  +  (\bm{\epsilon})_{j} 
	= \sum_{\ell=1}^{q} (\bm{\alpha}^{\top} )_{j  \ell} (\ddot{\mathbf{x}})_{\ell}  +  \sum_{\ell=1}^{s} ( \bm{\zeta}^{\top} )_{j \ell} (\ddot{\mathbf{z}})_{\ell}  +  (\bm{\epsilon})_{j}  , 
	\label{proof:causal.mediation:equation.M}
\end{align}
and 
\begin{align}
	(\ddot{\mathbf{y}})_{k} &= \big( \bm{\beta}^{\top} \ddot{\mathbf{m}} \big)_{k}  + \big( \bm{\gamma}^{\top} \ddot{\mathbf{x}} \big)_{k}  +  \big( \bm{\eta}^{\top} \ddot{\mathbf{z}} \big)_{k}  +  (\bm{\xi})_{k} 
	\label{proof:causal.mediation:equation.Y}
	\\
	&= \sum_{\ell=1}^{p} ( \bm{\beta}^{\top} )_{k \ell} (\ddot{\mathbf{m}})_{\ell}  +  \sum_{\ell=1}^{q} ( \bm{\gamma}^{\top})_{k \ell} (\ddot{\mathbf{x}})_{\ell}  +  \sum_{\ell=1}^{s} ( \bm{\eta}^{\top})_{k \ell} (\ddot{\mathbf{z}})_\ell  +  (\bm{\xi})_{k}  .
	\nonumber  
\end{align}
Under assumptions (C1) and (C2), from~\eqref{proof:causal.mediation:equation.Y} we have for the controlled direct effect (CDE) from the multivariate exposure to the $\ell^{th}$ outcome:  
\begin{align*}
& \E \big[ (\ddot{\mathbf{y}}(\mathbf{x},\mathbf{m}))_{k} - (\ddot{\mathbf{y}}(\widetilde{\mathbf{x}},\mathbf{m}))_{k}  \big| \ddot{\mathbf{z}} = \mathbf{z} \big] 
\\
= \quad & \E \big[ (\ddot{\mathbf{y}})_{k} \big| \ddot{\mathbf{x}} = \mathbf{x}, \ddot{\mathbf{m}} = \mathbf{m}, \ddot{\mathbf{z}} = \mathbf{z} \big]  - \E \big[ (\ddot{\mathbf{y}})_{k} \big| \ddot{\mathbf{x}} = \widetilde{\mathbf{x}}, \ddot{\mathbf{m}} = \mathbf{m}, \ddot{\mathbf{z}} = \mathbf{z} \big]
\\
= \quad &  \sum_{\ell=1}^{q} ( \bm{\gamma}^{\top})_{k \ell} (\mathbf{x} - \widetilde{\mathbf{x}})_{\ell} = \big( \bm{\gamma}^{\top} (\mathbf{x} - \widetilde{\mathbf{x}}) \big)_{k}. 
\end{align*}

The argument for the CDE from the $j^{th}$ exposure to the $\ell^{th}$ outcome is analogous.

Moreover, under assumptions (C1)-(C4), from~\eqref{proof:causal.mediation:equation.Y} we have by Pearl's formulations:  
\begin{align*}
\E \big[ \big( \ddot{\mathbf{y}}(\mathbf{x},\M(\widetilde{\mathbf{x}})) \big)_{k} \big| \ddot{\mathbf{z}} = \mathbf{z} \big]  &=  \int_{\mathbf{m}}  \E \big[ (\ddot{\mathbf{y}})_{k} \big| \ddot{\mathbf{x}} = \mathbf{x} , \ddot{\mathbf{z}} = \mathbf{z} , \ddot{\mathbf{m}} = \mathbf{m} \big] d \mathbb{P} (\mathbf{m} | \widetilde{\mathbf{x}}, \mathbf{z})  
\\
&=   \int_{m} \big\{  ( \bm{\beta}^{\top} \mathbf{m} )_{k}  + ( \bm{\gamma}^{\top} \mathbf{x} )_{k}  +  ( \bm{\eta}^{\top} \mathbf{z} )_{k}  \big\}  d \mathbb{P} (\mathbf{m} | \widetilde{\mathbf{x}}, \mathbf{z})   
\\
&= \big( \bm{\beta}^{\top} \E \big[  \ddot{\mathbf{m}}  \big| \ddot{\mathbf{x}} = \widetilde{\mathbf{x}}, \ddot{\mathbf{z}} = \mathbf{z} \big] \big)_{k} + ( \bm{\gamma}^{\top} \mathbf{x} )_{k}  +  ( \bm{\eta}^{\top} \mathbf{z} )_{k}  ,  
\end{align*} 
with $\E \big[  \ddot{\mathbf{m}}  \big| \ddot{\mathbf{x}} = \widetilde{\mathbf{x}}, \ddot{\mathbf{z}} = \mathbf{z} \big] := \big( \E \big[ (\ddot{\mathbf{m}})_{j}  \big| \ddot{\mathbf{x}} = \widetilde{\mathbf{x}}, \ddot{\mathbf{z}} = \mathbf{z} \big] \big)_{j = 1, ..., p}$.  

Thus, the natural direct effect (NDE) from the multivariate exposure to the $\ell^{th}$ outcome is: 
\begin{align*}
\E \big[ \big( \ddot{\mathbf{y}}(\mathbf{x},\ddot{\mathbf{m}}(\widetilde{\mathbf{x}})) \big)_{k} - \big( \ddot{\mathbf{y}}(\widetilde{\mathbf{x}},\ddot{\mathbf{m}}(\widetilde{\mathbf{x}})) \big)_{k} \big| \ddot{\mathbf{z}} = \mathbf{z} \big]   =  \big( \bm{\gamma}^{\top} (\mathbf{x} - \widetilde{\mathbf{x}} ) \big)_{k}  . 
\end{align*}
The argument for the NDE from the $j^{th}$ exposure to the $\ell^{th}$ outcome is analogous.

Similarly,  
\begin{align*}
& \E \big[ \big( \ddot{\mathbf{y}}(\mathbf{x},\ddot{\mathbf{m}}(\mathbf{x})) \big)_{k} - \big( \ddot{\mathbf{y}}(\mathbf{x},\ddot{\mathbf{m}}(\widetilde{\mathbf{x}})) \big)_{k} \big| \ddot{\mathbf{z}} = \mathbf{z} \big]  
\\
= \quad & \Big\{ \bm{\beta}^{\top} \Big( \E \big[  \ddot{\mathbf{m}}  \big| \ddot{\mathbf{x}} = \mathbf{x} , \ddot{\mathbf{z}} = \mathbf{z} \big] - \E \big[  \ddot{\mathbf{m}}  \big| \ddot{\mathbf{x}} = \widetilde{\mathbf{x}} , \ddot{\mathbf{z}}= \mathbf{z} \big] \Big) \Big\}_{k}.
\end{align*}
Moreover,   from~\eqref{proof:causal.mediation:equation.M} one has: 
\begin{align*}
\E \big[  (\ddot{\mathbf{m}})_{j}  \big| \ddot{\mathbf{x}} = \mathbf{x} , \ddot{\mathbf{z}} = \mathbf{z} \big] - \E \big[  (\ddot{\mathbf{m}})_{j}  \big| \ddot{\mathbf{x}} = \widetilde{\mathbf{x}} , \ddot{\mathbf{z}} = \mathbf{z} \big] = \big( \bm{\alpha}^{\top} (\mathbf{x} - \widetilde{\mathbf{x}}) \big)_{j} , \quad j = 1, ..., p,  
\end{align*}
Hence, the natural indirect effect (NIE) from the multivariate exposure to the $\ell^{th}$ outcome is:   
\begin{align*}
\E \big[ \big( \ddot{\mathbf{y}}(\mathbf{x},\ddot{\mathbf{m}}(\mathbf{x})) \big)_{k} - \big( \ddot{\mathbf{y}}(\mathbf{x},\ddot{\mathbf{m}}(\widetilde{\mathbf{x}})) \big)_{k} \big| \ddot{\mathbf{z}} = \mathbf{z} \big]    =  \Big\{   \bm{\beta}^{\top} \big( \bm{\alpha}^{\top} (\mathbf{x} - \widetilde{\mathbf{x}}) \big) \Big\}_{k} .
\end{align*}

The argument for the NIE from the $j^{th}$ exposure to the $\ell^{th}$ outcome is analogous.

This completes the proof for causal mediation in the MMM setting. 

\end{document}